\documentclass[american,11pt,letterpaper]{article}
\usepackage[T1]{fontenc}
\usepackage[utf8]{inputenc}
\usepackage{amsmath}
\usepackage{amsthm}
\usepackage{amssymb}
\usepackage{geometry}

\usepackage{wasysym}
\usepackage{appendix}
\usepackage{tikz}
\usepackage{graphicx}
\usepackage[nottoc,notlot,notlof]{tocbibind} 

  \usepackage[colorlinks, pagebackref, linkcolor=blue, citecolor=blue]{hyperref}

\makeatletter
\theoremstyle{plain}
\newtheorem{thm}{\protect\theoremname}
\theoremstyle{plain}
\newtheorem{lem}[thm]{\protect\lemmaname}
\theoremstyle{plain}
\newtheorem{prop}[thm]{\protect\propositionname}
\theoremstyle{plain}
\newtheorem{cor}[thm]{\protect\corollaryname}
\theoremstyle{plain}
\newtheorem{conj}[thm]{Conjecture}
\theoremstyle{remark}
\newtheorem{rem}[thm]{\protect\remarkname}

\usepackage{amsmath}
\DeclareMathOperator{\per}{per}
\DeclareMathOperator{\fix}{fix}

\makeatother

\usepackage{babel}
\providecommand{\corollaryname}{Corollary}
\providecommand{\lemmaname}{Lemma}
\providecommand{\propositionname}{Proposition}
\providecommand{\remarkname}{Remark}
\providecommand{\theoremname}{Theorem}

\newcommand{\Cat}{\mathrm{Cat}}
\newcommand{\supp}{\mathrm{supp}}

\newcommand{\R}{\mathbb{R}}
\newcommand{\C}{\mathbb{C}}

\newcommand{\NN}{\mathbb{N}}

\newcommand{\E}{\mathbb{E}}
\newcommand{\Var}{\operatorname{Var}}

\newcommand{\CN}{\mathcal{N}_{\mathbb C}}

\begin{document}

\title{Approximating the Permanent of a Random Matrix with Polynomially Small Mean: Zeros and Universality}
\author{Frederic Koehler\thanks{University of Chicago, \texttt{fkoehler@uchicago.edu}.} \and Pui Kuen Leung\thanks{University of Chicago, \texttt{pkl@uchicago.edu}.}}
\maketitle
\begin{abstract}
We study algorithms for approximating the permanent of a random matrix when the entries are slightly biased away from zero. This question is motivated by the goal of understanding the classical complexity of linear optics and \emph{boson sampling} (Aaronson and Arkhipov '11; Eldar and Mehraban '17). Barvinok's interpolation method enables efficient approximation of the permanent, provided one can establish a sufficiently large zero-free region for the polynomial $\per(zJ + W)$, where $J$ is the all-ones matrix and $W$ is a random matrix with independent mean-zero entries.

We show that when the entries of $W$ are standard complex Gaussians, all zeros of the random polynomial $\per(zJ + W)$ lie within a disk of radius $\tilde{O}(n^{-1/3})$, which yields an approximation algorithm when the bias of the entries is $\tilde{\Omega}(n^{-1/3})$. Previously, there were no efficient algorithms at biases smaller than $1/\mathrm{polylog}(n)$, and it was unknown whether there typically exist zeros $z$ with $|z| \ge 1$. As a complementary result, we show that the bulk of the zeros, namely $(1 - \epsilon)n$ of them, have magnitude $\Theta(n^{-1/2})$. This prevents our interpolation method from contradicting the conjectured average-case hardness of approximating the permanent.  
We also establish analogous zero-free regions for the hardcore model on general graphs with complex vertex fugacities. In addition, we prove universality results establishing zero-free regions for random matrices $W$ with i.i.d.\ subexponential entries. 
\end{abstract}
\thispagestyle{empty}
\tableofcontents
\thispagestyle{empty}
\newpage
\setcounter{page}{1}

\section{Introduction}

The permanent of a matrix $A$,
\[
\per(A) = \sum_{\sigma} \prod_i A_{i \sigma(i)}, \]
is a central object in theoretical computer science, combinatorics, and other areas. Despite its simple definition, computing the permanent of an $n\times n$ matrix is $\#P$-hard in the worst case \cite{valiant1979complexity}, and this hardness is widely believed to persist even for approximate computation in broad regimes. At the same time, there are also natural classes of matrices for which the permanent can be efficiently approximated --- notably, including matrices $A$ with nonnegative entries \cite{jerrum1989approximating,jerrum1996markov,jerrum2004polynomial}.

A particularly compelling motivation for studying the permanent comes from the study of linear optics and the \emph{boson sampling} problem \cite{aaronson2011computational}. In this setting, output probabilities of a quantum experiment are expressed in terms of permanents of random matrices, and the conjectured classical hardness of approximating these quantities is one of the main pieces of evidence for quantum computational advantage. Understanding when such permanents can nevertheless be approximated classically is therefore closely tied to the limits of classical simulation of quantum systems. See, e.g., \cite{anand2025simulating,eldar2018approximating,ji2021approximating,bouland2025exponential} for some relevant work and further discussion.

There is a natural analogy with cryptography here. To understand the security of a cryptographic primitive, one develops attacks and then tries to understand why those attacks fail beyond a certain point. The ``primitive'' is the conjectured \emph{average-case} hardness of approximating the permanent of a random matrix, and the ``attacks'' are efficient approximation algorithms. So far, the most powerful line of attack has been based on algorithmic Taylor expansions, as we discuss next.

\paragraph{Previous work.}
Barvinok \cite{barvinok2016computing,barvinok2019computing} showed that it is possible to approximate the permanent of some matrices with positive and negative entries --- those which are sufficiently diagonally dominant. In general, Barvinok's interpolation method reduces approximate counting to proving that a suitable partition function has no zeros in a region connecting a trivial point to the point of interest. Random matrices are very far from diagonally dominant, but Eldar and Mehraban \cite{eldar2018approximating} gave a new analysis of the zeros which shows that interpolation can be used to approximate the permanent of a random matrix with a mean of order $1/\mathrm{polyloglog}(n)$.

Using a different approach (not based on complex analysis), Ji, Jin, and Lu \cite{ji2021approximating} gave an improved approximation algorithm which estimates the permanent of a matrix with a mean of order $1/\mathrm{polylog}(n)$ in quasipolynomial time.
However, the limit of $1/\mathrm{polylog}(n)$ seems fundamental to their approach:
``with our technique only, it is rather hard to go beyond the $1/\mathrm{polylog}(n)$ barrier and essential new ideas seem necessary if this is ever possible.'' \cite{ji2021approximating}

\paragraph{Our results.}
In this paper, we revisit the interpolation approach to approximating the permanent, and show how to obtain dramatically stronger results. The core mathematical problem is to understand the complex zeros of the random polynomial
\[
\per(zJ+W),
\]
where $J$ is the all-ones matrix and $W$ is an $n \times n$ random matrix with independent mean-zero entries. This polynomial interpolates between the easily understood matrix $J$ (at the point $z = \infty$) and the random matrix $W$ (at $z = 0$). If one can show that $\per(zJ+W)$ is zero-free along a path from large $|z|$ down to the desired scale, then Barvinok's method yields an efficient approximation algorithm for the permanent of a random matrix whose entries have a small nonzero bias.

Our first main result concerns the complex Gaussian setting.
When the entries of $W$ are standard complex Gaussians, we prove that with high probability all zeros of $\per(zJ+W)$ lie within a disk of radius $\tilde O(n^{-1/3})$.  Prior to this work, it was not even known whether zeros of constant magnitude typically occur.

As a consequence, we obtain a polynomial time approximation algorithm when the bias of the entries is $\Omega(n^{-1/3 + \beta})$ for any $\beta > 0$. We include a concrete theorem statement below for clarity.
\begin{thm}[Short version of Theorem~\ref{thm:degree-accuracy-tradeoff}]
Let $\gamma, \beta > 0$ and $\delta > 0$. There exists a deterministic algorithm such that
for all sufficiently large $n \ge n_0(\gamma,\beta,\delta)$, 
the algorithm outputs an \(n^{-\gamma}\)-approximation to
\(\log \per(Jz + W)\) for any \(|z|\ge n^{-1/3+\beta}\) 
with probability at least $1 - \delta$ (uniformly over $z$) and
with runtime
\( n^{O(\gamma)}. \)
\end{thm}
The algorithm is based on Taylor expansion to degree $O(\gamma)$; note that $\gamma > 0$ is chosen by the user, so our algorithm can achieve arbitrarily small inverse polynomial accuracy.
This substantially improves on the state of the art, where such an approximation guarantee was only known for biases on the order of $1/\mathrm{polylog}(n)$, and with quasipolynomial runtime \cite{ji2021approximating}. 
We also prove a complementary result showing that the bulk of the zeros, namely $(1-\epsilon)n$ of them, have magnitude $\Theta(n^{-1/2})$. Thus, while all zeros lie in a disk of radius about $n^{-1/3}$, most zeros are in fact concentrated at the smaller scale $n^{-1/2}$. This gives a conceptual check on the scope of our algorithmic result: the zero-free region we prove is genuinely nontrivial, but it remains outside the regime where one might expect to contradict the conjectured average-case hardness of approximating the permanent.

Beyond the Gaussian case, we establish universality by showing that zero-free regions at scale $n^{-1/4}$ continue to exist for much more general random matrices with i.i.d.\ subexponential entries. 
Our methods also extend beyond permanents --- in particular, we establish analogous zero-free regions for the hardcore model on general graphs with complex vertex fugacities. As we shall see later, the permanent morally corresponds to hardcore model on the line graph of $K_{n,n}$, also known as the monomer-dimer model on $K_{n,n}$.

% ...




\subsection{Further related work}
The study of complex zeros of partition functions has a long history in statistical physics and combinatorics, beginning with seminal work of Lee and Yang \cite{lee1952statistical} which established a zero-free region for the Ising model and established connections between zeros and phase transitions. In the case of monomer--dimer systems, the classical theorem of Heilmann and Lieb \cite{heilmann1972theory} plays an analogous role. More broadly, a large literature has developed around zero-free regions for independence, matching, and related polynomials, together with their algorithmic consequences; see for example \cite{brenner1959relations,shearer1985problem,scott2005repulsive,chudnovsky2007roots,bissacot2011improvement,patelregts,peters2019conjecture,liu2019ising}.

There is a substantial literature on approximation algorithms for the permanent outside the random small-bias setting. For nonnegative matrices, we previously mentioned the breakthrough FPRAS of Jerrum, Sinclair, and Vigoda \cite{jerrum1989approximating,jerrum1996markov,jerrum2004polynomial}. Deterministic approximation via Bethe-type variational methods and related ideas has also been studied extensively, see for example \cite{linial1998deterministic,chertkov2013approximating,gurvits2014bounds,anari2025tight,zhou2026complex}. Other recent work has considered structured matrix classes, such as positive semidefinite matrices \cite{ebrahimnejad2025approximability}, as well as exponential time algorithms \cite{clifford2018classical}.

On the probabilistic side, several works study the size and anticoncentration of the permanent of random matrices \cite{tao2009permanent,Tao2010MO,kwan2022permanent,hunter2025exponential,aaronson2011computational}. More broadly, our results fit into a growing body of work connecting the geometry of complex zeros to computational tractability and hardness in counting problems. In particular, for spin systems and graph polynomials, recent works have related the presence of zeros or zero-dense regions to phase transitions, failure of correlation decay, and hardness of approximation --- see, e.g., \cite{sly2010computational,buys2022lee,de2024zeros,galanis2022complexity,bezakova2021complexity,bencs2025complex}.

Our result fits into a line of recent works concerning algorithms for disordered systems. Perhaps most closely relevant, Bencs, Huang, Lee, Liu, and Regts \cite{bencs2025zeros} studied zero-free regions and quasipolynomial-time approximation in mean-field spin glasses using Barvinok's approach combined with a reweighted second moment calculation. In their context the reweighting plays a less critical role: it is used to improve a $O(1)$ bound on the number of zeros, which already gives useful algorithmic consequences, to a zero-freeness result. Another technically relevant work is due to Mohanty and Rajaraman \cite{mohanty2025eigenvalue}, where they analyzed eigenvalues of random matrices by viewing the corresponding determinant as a random polynomial and analyzing its zeros. These works fit into a broader literature rigorously analyzing how algorithmic thresholds shift (compared to worst-case/uniqueness thresholds) when there are random weights in the model --- see, e.g., \cite{gheissari2019spectral, eldan2022spectral,adhikari2024spectral,anari2024trickle,anari2024universality,chen2025rapid,anari2022entropic,kunisky2024optimality}.

Finally, our motivation is tied to the literature on boson sampling and quantum advantage, where permanents of random matrices appear as output amplitudes or probabilities. This has led to sustained interest in the average-case complexity of the permanent, in the power and limits of classical simulation, and in identifying regimes where structured instances admit efficient classical algorithms --- see, e.g., \cite{aaronson2011computational,harrow2020classical,clifford2018classical,anand2025simulating}.
\section{Overview}\label{sec:overview}
\subsection{Riemann sphere perspective}
The following perspective formalizes the point $z = \infty$ and is helpful for keeping track of the scale of different quantities in the analysis. 
We are interested in locating the zero set of the degree $n$ polynomial
\[ z \mapsto \per(zJ + W), \]
which is a zero-dimensional affine variety.
It is natural to homogenize and rephrase the question in terms of the corresponding zero-dimensional projective variety
\[ \mathcal V = \{ [z_1:z_2] : \per(z_1 J + z_2 W) = 0 \} \subset \mathbb P^1 \]
which lives inside the Riemann sphere $\mathbb P^1 = \{ [z:1] : z \in \mathbb C \} \cup \{[1:0]\}$. 
This gives us two symmetrical ways to think about the zero set:
\begin{enumerate}
    \item The zeros of the projective version are in one-to-one correspondence to the original affine version under the map $[z_1:z_2] \to z_1/z_2$, because at the new point at $\infty$, i.e. at $[1:0]$, we have that $\per(J) = n! \ne 0$. 
    \item On the other hand, for generic $W$ we have that $\per(W) \ne 0$, so the zero set is also in one-to-one correspondence with the zero-dimensional affine variety $\{ z : \per(J + zW) = 0 \}$ under the map $[z_1:z_2] \to z_2/z_1$.
\end{enumerate}
\emph{We will almost always use the latter ``inverted'' coordinate system in what follows.} The reason is our original goal, which is to prove there are no zeros of the original polynomial of, e.g., scale $1/n^{1/4}$ or larger, is more naturally rephrased by centering about the point at infinity: we will show there are no solutions of $\per(J + zW) = 0$ such that $|z| \ll n^{1/4}$. 
\subsection{Previous approaches}
The best zero-freeness results known so far came from the original work of Eldar and Mehraban \cite{eldar2018approximating}. By using Jensen's formula from complex analysis, combined with an upper
bound on the sensitivity of the permanent \cite{rempala1999limiting}, they proved (Equation (28) of \cite{eldar2018approximating}) that for any $r > 0$,
\begin{equation} 
\mathbb E \# \{|z| < r : \per(J + zW) = 0 \} \le 4r^2. 
\end{equation}
Depending on the scale of $r$ we choose, this result has various consequences:
\begin{itemize}
    \item \textbf{Zero-free region, Proposition 8 of \cite{eldar2018approximating} part 1.} For any $r < 1/2$, with probability at least $1 - 4r^2$, there are no roots within distance $r$ of the origin.
    \item \textbf{Region with few zeros, Proposition 8 of \cite{eldar2018approximating} part 2.} For any $r > 1$, with probability at least $1 - 1/r$ there are at most $4r^3$ roots within distance $r$ of the origin. This result is used to obtain their main algorithmic result, by arguing they can pick a path likely to avoid these roots and then iteratively Taylor expand along this path.
    \item \textbf{Region with $o(n)$ zeros.} For any $r = o(\sqrt{n})$, with probability $1 - o(1)$ there are at most $o(n)$ roots within distance $r$ of the origin. This can also be deduced by combining coefficient bounds in \cite{eldar2018approximating} with results from the theory of random polynomials (see Theorem 4 of  \cite{hughes2008zeros}).
\end{itemize}
It turns out that the last estimate for $o(n)$ zeros is qualitatively sharp (see Section~\ref{sec:permanent-stability}). However, the former two estimates are the ones which are relevant to algorithms;
To apply Barvinok's method \cite{barvinok2016computing}, we need to have a guarantee for a region that there are \emph{no} zeros (or very few, so that we may hopefully pick a further subregion which is truly zero-free, as done in \cite{eldar2018approximating}). 




\paragraph{Algorithm of Ji, Jin, and Lu \cite{ji2021approximating}.} Interestingly, the best current algorithm for this problem does not use the zero-free machinery at all. Instead, they use the more direct interpretation of $\per(J + zW)$ as a generating function for subpermanents of $W$:
\[ \frac{1}{n!} \per(J + zW) = \frac{1}{n!} \sum_{\sigma} \prod_{i = 1}^n (J + zW)_{i \sigma(i)} = \sum_{k = 0}^n \frac{(n - k)!}{n!} z^k \sum_{B \subset_{k} W} \per(B) \]
where $B$ ranges over $k \times k$ submatrices of $W$. The number of such submatrices $B$ is ${n \choose k}^2 = \left(\frac{(n!)}{(n - k)! k!}\right)^2$, so the standard deviation of the $k$th coefficient is
\begin{equation}\label{eqn:naive-expansion-var}
\sqrt{\Var\left(\frac{(n - k)!}{n!} \sum_{B \subset_{k} W} \per(B)\right)} = \frac{1}{k!}. 
\end{equation}
Because this decays as $1/k! \le (e/k)^k$, we can hope that if $z = O(\mathrm{polylog}(n))$, then truncating
the expansion after $O(\mathrm{polylog}(n))$ terms will achieve small error; this is rigorously justified in \cite{ji2021approximating} and gives a quasipolynomial time algorithm in this regime. 

\subsection{New approach}
We show the zero-free region is actually much larger than previously known, and thereby obtain stronger algorithmic consequences. Our analysis is inspired by an ansatz which suggests that a large cancellation may occur (Equation~\eqref{eqn:guess-cancellation} below). The following analogy is helpful for understanding where this guess comes from.

\paragraph{Analogy between permanent and hardcore model.} 
By definition,
\[ \per(J + zW) = \sum_{\sigma} \prod_i (J + zW)_{i\sigma(i)} \] 
Consider expanding out the terms on the right hand side. For each row $i$, each term encodes a choice of:
\begin{itemize}
    \item A potential column $\sigma(i)$ to match row $i$ with.
    \item Whether to include this edge or not ($zW_{i\sigma(i)}$ if yes, or $1$ if not).
\end{itemize}
This means that every term in the sum can be mapped to a corresponding matching of the complete $n \times n$ bipartite graph. Now, observe that for every matching $M$ in the bipartite graph, there are $(n - |M|)!$ permutations in the sum each generating the matching one time. So we can rewrite
\begin{align}\label{eqn:nonidealized} 
\frac{1}{n!} \per(J + zW) &= \sum_M \frac{(n - |M|)!}{n!} \prod_{(i,j) \in M} (zW_{ij})
\end{align}
Pretend that we can replace $(n - |M|!)/n!$ in the expression for $per(J + zW)/n!$ in \eqref{eqn:nonidealized} by $1/n^{|M|}$ --- then we get the partition function of a hardcore model on the line graph of the complete bipartite graph $K_{n,n}$ with complex fugacities:
\[ Z(z) = \sum_M \prod_{(i,j) \in M} (z W_{ij}/n)\]
This model is also known as a \emph{monomer-dimer model}. The hardcore constraint encodes that each vertex in the original graph will be adjacent to at most one edge from the matching $M$. We call this hardcore model the \emph{idealized model} of the permanent; it makes the next step much more transparent. 
\begin{rem}[Comparison with Heilmann-Lieb theorem \cite{heilmann1972theory}]
Suppose $W_{ij} = 1$ for all $i,j$, then this is the matching polynomial of the complete bipartite graph. By the Heilmann-Lieb theorem \cite{heilmann1972theory}, all of the zeros of $Z(z)$ would be along the negative real axis. More generally, the Heilmann-Lieb theorem applies when $W_{ij} \ge 0$. However, we are interested in weights with mean zero so the behavior will be quite different. 
\end{rem}

\paragraph{Cluster expansion of idealized hardcore model.}
Note that $Z(0) = 1$ because the empty set is the only size 0 matching.
Using the identity $\log(1 + x) = x - x^2/2 + x^3/3 - \cdots$ we can expand the partition function
as a \emph{formal power series}\footnote{The Taylor series $\log(1 + x)$ does converge when $|x| < 1$, but we will be considering the resulting series expansion in situations far outside this regime --- where, e.g., $|Z(z)| = \Omega(e^{n^{1/10}}) \gg 1$.} in $z$ and obtain
\[ \log Z(z) = (z/n) \sum_{ij} W_{ij} - (z/n)^2\left(\sum_v W_{ij}^2/2 + \sum_{ijk : i \ne k }W_{ij} W_{kj} + \sum_{ijk : j \ne k} W_{ij}W_{ik} \right) + O_n(z^3). \]
For the hardcore model, there is a systematic way to compute any of the formal power series coefficients in terms of a sum over connected diagrams, which leads this series to be called the \emph{cluster expansion} \cite{friedli2017statistical}. See Figure~\ref{fig:cluster-diagrams}.

We can compute that the first term in the expansion of $\log Z$ has variance $O(z^2)$ and the second term has variance $O(z^4/n)$. More generally, from the explicit cluster expansion formula (see Section~\ref{sec:preliminaries}) we know that the $k$th term has variance $O_k(z^{2k}/n^{k - 1})$. 
If $|z|^4 \ll n$, the contribution of the second term in the expansion is $o(1)$, which suggests we \emph{might} be able to make the approximations
\begin{equation} \log \frac{\per(J + zW)}{n!}  \approx \log Z(z) \approx (z/n) \sum_{ij} W_{ij} \, ? \label{eqn:guess}
\end{equation}

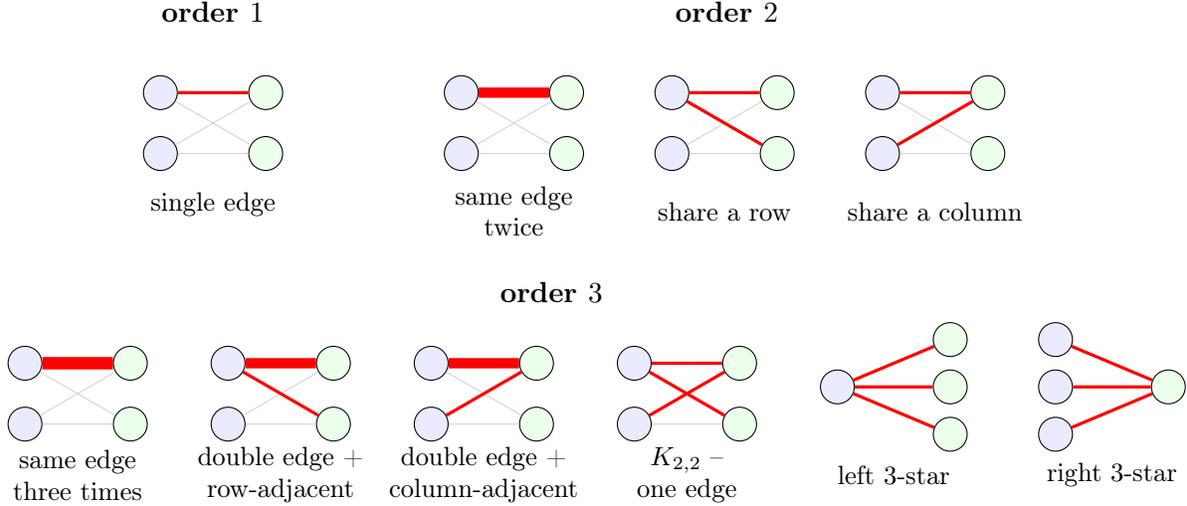
\begin{figure}[t]
\centering
\begin{tikzpicture}[x=1cm,y=0.9cm,
  leftv/.style={circle,draw,inner sep=1pt,minimum size=4.5mm,fill=blue!8},
  rightv/.style={circle,draw,inner sep=1pt,minimum size=4.5mm,fill=green!8},
  edgehi/.style={very thick,red},
  edgeghost/.style={gray!35},
  lab/.style={font=\small},
  orderlab/.style={font=\bfseries}
]

\newcommand{\twobytwo}[3]{%
  \node[leftv]  (#1L1) at (#2,#3+0.9) {};
  \node[leftv]  (#1L2) at (#2,#3) {};
  \node[rightv] (#1R1) at (#2+1.4,#3+0.9) {};
  \node[rightv] (#1R2) at (#2+1.4,#3) {};
  \draw[edgeghost] (#1L1)--(#1R1);
  \draw[edgeghost] (#1L1)--(#1R2);
  \draw[edgeghost] (#1L2)--(#1R1);
  \draw[edgeghost] (#1L2)--(#1R2);
}

\newcommand{\onethreeleft}[3]{%
  \node[leftv]  (#1A)  at (#2,#3+0.7) {};
  \node[rightv] (#1B1) at (#2+1.5,#3+1.4) {};
  \node[rightv] (#1B2) at (#2+1.5,#3+0.7) {};
  \node[rightv] (#1B3) at (#2+1.5,#3) {};
  \draw[edgeghost] (#1A)--(#1B1);
  \draw[edgeghost] (#1A)--(#1B2);
  \draw[edgeghost] (#1A)--(#1B3);
}

\newcommand{\threeoneright}[3]{%
  \node[rightv] (#1A)  at (#2+1.5,#3+0.7) {};
  \node[leftv]  (#1B1) at (#2,#3+1.4) {};
  \node[leftv]  (#1B2) at (#2,#3+0.7) {};
  \node[leftv]  (#1B3) at (#2,#3) {};
  \draw[edgeghost] (#1B1)--(#1A);
  \draw[edgeghost] (#1B2)--(#1A);
  \draw[edgeghost] (#1B3)--(#1A);
}

\begin{scope}[xshift=1.8cm]
  \node[orderlab] at (0.7,5.1) {order $1$};
  \node[orderlab] at (7.53,5.1) {order $2$};

  \twobytwo{a}{0.0}{3.0}
  \draw[edgehi] (aL1)--(aR1);
  \node[lab] at (0.7,2.25) {single edge};

  \twobytwo{b}{4.0}{3.0}
  \draw[edgehi,double distance=1.5pt,double=red] (bL1)--(bR1);
  \node[lab,align=center] at (4.7,2.15) {same edge\\twice};

  \twobytwo{c}{6.8}{3.0}
  \draw[edgehi] (cL1)--(cR1);
  \draw[edgehi] (cL1)--(cR2);
  \node[lab,align=center] at (7.5,2.15) {share a row};

  \twobytwo{d}{9.6}{3.0}
  \draw[edgehi] (dL1)--(dR1);
  \draw[edgehi] (dL2)--(dR1);
  \node[lab,align=center] at (10.3,2.15) {share a column};
\end{scope}
\node[orderlab] at (7,0.95) {order $3$};

\twobytwo{e}{0.0}{-1.0}
\draw[edgehi,double distance=2.3pt,double=red] (eL1)--(eR1);
\draw[edgehi] (eL1)--(eR1);
\node[lab,align=center] at (0.7,-1.75) {same edge\\three times};

\twobytwo{f}{2.7}{-1.0}
\draw[edgehi,double distance=1.5pt,double=red] (fL1)--(fR1);
\draw[edgehi] (fL1)--(fR2);
\node[lab,align=center] at (3.4,-1.75) {double edge +\\row-adjacent};

\twobytwo{g}{5.4}{-1.0}
\draw[edgehi,double distance=1.5pt,double=red] (gL1)--(gR1);
\draw[edgehi] (gL2)--(gR1);
\node[lab,align=center] at (6.1,-1.75) {double edge +\\column-adjacent};

\twobytwo{h}{8.1}{-1.0}
\draw[edgehi] (hL1)--(hR1);
\draw[edgehi] (hL1)--(hR2);
\draw[edgehi] (hL2)--(hR1);
\node[lab,align=center] at (8.8,-1.75) {$K_{2,2}$ -- \\ one edge};

\onethreeleft{i}{10.8}{-1.15}
\draw[edgehi] (iA)--(iB1);
\draw[edgehi] (iA)--(iB2);
\draw[edgehi] (iA)--(iB3);
\node[lab,align=center] at (11.55,-1.75) {left 3-star};

\threeoneright{j}{13.7}{-1.15}
\draw[edgehi] (jB1)--(jA);
\draw[edgehi] (jB2)--(jA);
\draw[edgehi] (jB3)--(jA);
\node[lab,align=center] at (14.45,-1.75) {right 3-star};

\end{tikzpicture}
\caption{Connected diagrams contributing to the first three orders of the formal expansion of $\log Z(z)$ for the idealized monomer--dimer model on $K_{n,n}$.} 
\label{fig:cluster-diagrams}
\end{figure}
We emphasize that the validity of neglecting all of the high-degree terms is not obvious, though it will ultimately be justified by our proof.
Many important perturbative expansions do \emph{not} converge\footnote{For example, the Dyson series in QED (Quantum Electrodynamics, the fundamental theory governing the behavior of light) is known to have a radius of convergence equal to zero based on physical considerations \cite{dyson1952divergence}. }, and are at best \emph{asymptotic series expansions} (see, e.g., the discussion in Chapter 13.1 of \cite{talagrand_qft}). 
The basic reason these types of perturbative expansions can fail to converge is the dependence of the coefficient $a_k$ on $k$. Usually there is a \emph{combinatorial explosion} in the number of terms contributing to the $k$th coefficient, so that the magnitude of $|a_k|$ can grow like $k!$ or worse. The dependence on $k$ can be neglected for $k$ small and $n$ large, but the power series for the cluster expansion ranges from $k = 1$ to $\infty$ for any finite $n$. 

\begin{rem}[Cancellations in cluster expansion]
As stated above, the variance of the $k$th coefficient in the Taylor series for $\log Z$ decays rapidly with $n$ for fixed $k$. However, in the polynomial expansion of $Z$ they do not --- they are $\Theta_k(1)$ by a calculation analogous to \eqref{eqn:naive-expansion-var}. This discrepancy indicates that massive cancellations occurred when we expanded the formal series for $\log(1 + x)$, and gives some indication why we can dramatically improve over the approach of \cite{ji2021approximating}.
\end{rem}

\paragraph{Convergence and zero-free region.} The mathematical heart of this work is to prove that for the log-permanent, the series expansion \emph{does} have a large (polynomial in $n$) radius of convergence. Unfortunately, standard methods for analyzing the convergence of the cluster expansion such as the Koteck\'y-Preiss (KP) condition completely fail at this task; they can show at best that our power series converges at some radius of the order of $O(1)$, which is already known \cite{eldar2018approximating}. The KP condition and similar methods aim to directly show that the power series is absolutely convergent by comparing the magnitude of various terms in the expansion --- this will not work well in our situation, where we want to take advantage of large cancellations caused by the random sign/phase of the weights.

Since the permanent itself is analytic, for the log-permanent to have a convergent expansion we need to show that there are no small \emph{complex zeros} $z$ such that $\per(J + zW) = 0$. Note that because the permanent is a degree $n$ polynomial, we are guaranteed to have $n$ zeros somewhere, but we want to show that they are far away from the origin (or, in the original coordinate system, the point at infinity).

A natural way to prove such a zero-free region exists is by applying Jensen's formula from complex analysis, which is an explicit relation between the location of the zeros of a holomorphic function to its averaged rate of growth. In particular,  if we consider any $R > r$ and count the number of zeros $z_1,\ldots,z_m$ of $\per(J + zW)$ inside the smaller disc of radius $r$, we find
\[ \#\{ |z| < r : \per(J + zW) = 0 \}  \le  \frac{1}{2\pi \log(R/r)} \int_0^{2\pi} \log |\per(J + Re^{i\theta} W)/n!| \, d\theta \]
This observation was already at the heart of the proof of \cite{eldar2018approximating}.

\paragraph{Gain via cancellation.} Heuristically, based upon \eqref{eqn:guess} we would guess that for $|z| \ll n^{1/4}$,
\begin{equation}\label{eqn:guess-cancellation}
\frac{1}{2\pi} \int_0^{2\pi} \log  |\per(J + Re^{i\theta} W)/n!| \approx \frac{1}{2\pi}  \Re \int_0^{2\pi} (R e^{i\theta}/n)\sum_{ij} W_{ij} d\theta = 0 
\end{equation}
since $\int e^{i\theta} d\theta = 0$, which would give us the desired zero-free region for $|z| \ll n^{1/4}$. Taking advantage of this cancellation over $\theta$ is the key improvement over the previous analysis \cite{eldar2018approximating}. To make this calculation rigorous, we need to verify \eqref{eqn:guess} is true. Since \eqref{eqn:guess} says that $\frac{\per(J + zW)}{n!} \approx e^{(z/n)\sum_{ij} W_{ij}}$, we can try to make our argument rigorous by bounding the absolute variance of the quotient:
\[ \mathbb E \left|\frac{\per(J + zW)}{n!}e^{(-z/n)\sum_{ij} W_{ij}} - 1\right|^2 = o(1). \]
The left hand side essentially reduces to the partition function of an \emph{auxiliary} statistical mechanics system, which requires some nontrivial combinatorics to evaluate. We leave further details to Section~\ref{sec:easiest}, but ultimately this works and proves the most basic version of our result. 

\begin{figure}
    \centering
    \includegraphics[width=\linewidth]{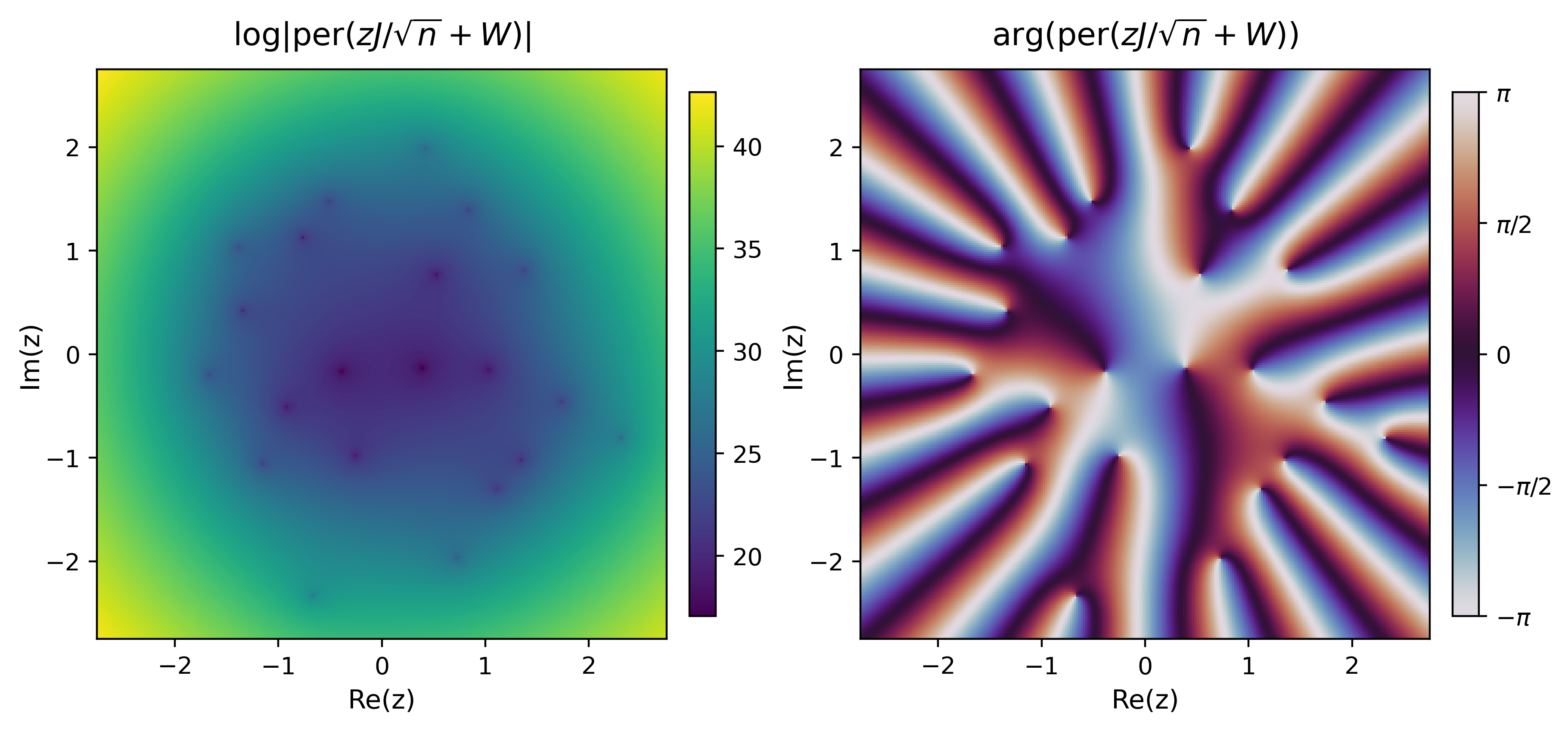}
    \caption{Magnitude (left) and phase (right) of the function $\per(zJ/\sqrt{n} + W)$ of a $21 \times 21$ dimensional random matrix $W$ with i.i.d.\ standard complex Gaussian entries. Zeros are visible as dark blue dots on the left, and by argument principle considerations on the right.}
    \label{fig:placeholder}
\end{figure}

\paragraph{Generalization to hardcore model on arbitrary graphs.} At the level of the above heuristics, the permanent is essentially a monomer-dimer model on the complete bipartite graph. It is relatively straightforward to generalize the heuristic calculations to show that similar cancellations occur for the monomer-dimer model and more generally the hardcore model on \emph{any} graph $G = (V,E)$ with maximum degree $\Delta \ll |V|$. We verify this prediction rigorously in Section~\ref{sec:hardcore-first}. It is worth noting that although the heuristic predictions work the same between permanent and hardcore model, the behavior of the high-degree terms in their expansions is genuinely different --- so the fundamental combinatorics in the proof of the results also ends up to be fairly disparate between the two settings.

\paragraph{Universality.} For both the permanent and the general hardcore model, we are able to show that the zero-free regions derived in the case of $N_{\mathbb C}(0,1)$ entries generalizes to any i.i.d. subexponential distribution. This is fairly involved, because some of the formulas we used in the above heuristic derivation fail to hold, even when the distribution of the entries are \emph{real} Gaussian instead of complex (see Section~\ref{sec:nonuniversal}). We omit the details here, but broadly speaking we are able to show using a softer analysis that the key ``cancellation'' over $\theta$ still occurs in this general setting. See Sections~\ref{sec:permanent-universality} and \ref{sec:hardcore-universality}.

\paragraph{Second-order reweighting for the permanent.} When the entries are $N_{\mathbb C}(0,1)$, the heuristic analysis above naturally generalizes to the setting where $|z| \ll n^{1/3}$ by taking into account the second term in the cluster expansion as well as the first one. The rigorous analysis of this is again fairly involved, but we are able to verify this in the case of the permanent, as well as for its ``idealized'' analogue, the monomer-dimer model on the complete bipartite graph. See Section~\ref{sec:permanent-second}.

\paragraph{Upper bound on radius of convergence via anticoncentration.} Based on a reduction from \cite{eldar2018approximating}, if we could improve our approximation algorithm to succeed when $|z| = n^{1/2 + \epsilon}$ for any $\epsilon > 0$, then we would be able to approximate the permanent of a random matrix better
than expected, which would break the complexity-theoretic assumptions used to argue for the hardness of boson sampling in \cite{aaronson2011computational}. Interestingly, we are able to prove that $(1 - \epsilon)n$ of the zeros of $\per(J + zW)$ are at scale $\Theta(\sqrt{n})$, forming a direct obstruction to convergence of the Taylor series expansion. So if it is possible to break the conjectured hardness of approximating the permanent, new algorithmic ideas would have to be required. 

At a technical level, this result would be much easier to prove if we knew the conjectured anticoncentration of the permanent to be true. However, existing results fall far short of proving this conjecture (see \cite{aaronson2011computational,tao2009permanent,Tao2010MO}) and, in particular, seem to be too weak for our present application. In the end, we are able to recover the result unconditionally by deriving an inductive lower bound on $\mathbb E \log |\per(W)|^2$ (Lemma~\ref{lem:gaussian-per-log-lower}).

\subsection{Organization}
After going through some mathematical preliminaries (Section~\ref{sec:preliminaries}),
we go on to develop our main results for the permanent. In Section~\ref{sec:easiest}, to introduce our method we go through the proof that for complex Gaussians, the radius of convergence is at least $O(n^{1/4})$. This initial result is subsumed by the more sophisticated arguments in the next two sections: in Section~\ref{sec:permanent-second}, we show that for the complex Gaussian case we can improve the radius of convergence to $O(n^{1/3})$, and in Section~\ref{sec:permanent-universality}, we show that the $O(n^{1/4})$-radius zero-free result is true for arbitrary subexponential distributions.

We then proceed to establish analogous results in the setting of the hardcore model on a general graph. We prove that analogous zero-free results hold for the hardcore model on any graph of maximum degree $\Delta$ with subexponential weights --- first in Section~\ref{sec:hardcore-first} for the case of complex Gaussian weights, and then in Section~\ref{sec:hardcore-universality} for general subexponential weights. Conceptually, these developments are parallel to the analysis of the permanent, although the technical details differ significantly.

In Section~\ref{sec:permanent-stability} we analyze the relationship between our findings with the conjectured average-case hardness permanent: we show that if the radius of convergence can be improved to $\omega(\sqrt{n})$, then we can estimate the permanent on a classical computer better than is predicted possible by \cite{aaronson2011computational}, and we show that under the anticoncentration conjecture about Gaussian permanents, the radius of convergence \emph{cannot} be improved beyond $\Theta(\sqrt{n})$.

In Appendix~\ref{apdx:simulations}, we give the results of numerical experiments for small values of $n$. In Appendix~\ref{app:bethe}, we discuss a connection between the formulas we establish and the Bethe approximation used in variational inference. In Appendix~\ref{apdx:coeff-bounds}, we include estimates on the size of the low-degree terms in the cluster expansion, and its permanent analogue.

\section{Preliminaries}
\label{sec:preliminaries}
\subsection{Complex analysis}
See the textbook \cite{stein2010complex} for a more comprehensive reference.

\paragraph{Jensen's formula.} We make extensive use of Jensen's formula from complex analysis, which says that if $f$ is an analytic function, $r > 0$, and $a_1,\ldots,a_N$ are the complex zeros of $f$ inside of the disc of radius $r$ about the origin, repeated according to their multiplicities, then
\[ \log |f(0)| = -\sum_{k = 1}^N \log \frac{r}{|a_k|} + \frac{1}{2\pi} \int_0^{2\pi} \log |f(re^{i\theta})| d\theta. \]
So $\log |f(0)|$ is determined by the locations of the singularities within the disc, and by its average value on the boundary.  A commonly used consequence of this result is that if we consider any $R > r$ and apply the formula at radius $R$ while still considering zeros $a_1,\ldots,a_N$ inside the smaller disc of radius $r$, we find
\begin{equation}\label{eq:jensen-variant}
m \log(R/r) \le \sum_{k = 1}^N \log \frac{R}{|a_k|} \le  \frac{1}{2\pi} \int_0^{2\pi} \log |f(Re^{i\theta})| d\theta - \log |f(0)| 
\end{equation}
which is a convenient way to upper bound the number of zeros $N$. 

The precise way we will use Jensen's formula is stated in the following lemma.
Let $D(0,r) = \{ z : |z| < r \}$.
For a holomorphic function \(f\not\equiv 0\), write \(N_f(r)\) for the number of
zeros of \(f\) in \(\mathbb D(0,r)\), counted with multiplicity. When we apply the below result, we often normalize $g$ so that $|g(0)| = 1$ and hence $\log |g(0)| = 0$.
\begin{lem}\label{lem:core-jensen-app}
Let $R > r > 0$ and let $f$ and $h$ be random functions which are almost surely holomorphic on an open set containing $D(0,R)$ and such that $f(0) \ne 0$ almost surely. Let $g = e^{h} f$, then
\[ \mathbb E\,N_{f}(r)
\le
\frac{1}{4\pi \log(R/r)}
\int_{0}^{2\pi} \log \mathbb E |g(Re^{i\theta})|^2 d\theta - \mathbb E \log |g(0)|. \]
\end{lem}
\begin{proof}
By construction, $f$ and $g$ have the same zeros, so it is equivalent to bound $\mathbb E N_g(r)$. Also, $\log |g(z)| = \frac{1}{2} \log |g(z)|^2$.
So the expected number of zeros can be bounded by Jensen's formula \eqref{eq:jensen-variant}, using that
\[ \mathbb{E}_{\theta} \mathbb{E}\log|g(R e^{i\theta})|^{2} \le \mathbb{E}_{\theta} \log \mathbb{E}|g(R e^{i\theta})|^{2} \]
by Jensen's inequality. Here the expectation over $\theta$ is just the uniform measure from $0$ to $2\pi$. This proves the result.

For an alternative proof, we can apply Jensen's formula directly to $f$, and use that 
\[ \frac{1}{2\pi} \int \log |e^h(Re^{i\theta})| d\theta = \log |e^{h(0)}| \]
since $\log |e^h|$ is a harmonic function. Therefore, the net contribution of $h$ to the right hand side of Jensen's formula is zero, so we can add zero to both sides and then the conclusion follows by Jensen's inequality as above.
\end{proof}

\paragraph{Integrating the logarithmic derivative.} We will also use the standard fact that a nowhere-vanishing analytic function
admits a holomorphic logarithm on any simply connected domain (Theorem 6.2 of \cite{stein2010complex}). Explicitly,
if $\Omega \subset \mathbb{C}$ is simply connected and $f$ is holomorphic on
$\Omega$ with $f(z)\neq 0$ for all $z\in \Omega$, then there exists a
holomorphic function $g$ on $\Omega$ such that
\[
f(z)=e^{g(z)} \qquad (z\in \Omega).
\]
Equivalently, $g'(z)=f'(z)/f(z)$. In particular, fixing a base point
$z_0\in\Omega$ and a choice of logarithm $g(z_0)=\log f(z_0)$, one has
\begin{equation}
\label{eq:log-by-integral}
g(z)=\log f(z_0)+\int_{\gamma} \frac{f'(w)}{f(w)}\,dw,
\end{equation}
where $\gamma$ is any piecewise $C^1$ path in $\Omega$ from $z_0$ to $z$.
Because $\Omega$ is simply connected and $f'/f$ is holomorphic, the integral is
path-independent. 
\subsection{Hardcore model and cluster expansion}
Let $G = (V, E)$ be a finite graph and let $\lambda = (\lambda_v)_{v \in V} \in \mathbb{C}^V$ be a collection of complex fugacities assigned to its vertices. The hardcore model is a statistical mechanics model defined on the collection of all independent sets $\mathcal{I}(G)$ of the graph. As a reminder, an independent set $I$ is a subset of vertices such that no two elements of $I$ are adjacent in the graph. We can think of the hardcore model as a ``gas'' of particles: each vertex can be occupied by at most one particle, and two particles cannot be at the same or at adjacent vertices. Each particle can be thought of as a hard sphere of radius $1$, and the adjacency constraint is to prevent spheres from overlapping. See \cite{perkins2023five,friedli2017statistical} for more background and context.

The partition function $Z(\lambda)$ of the hardcore model is the sum of the weights of all valid independent sets:
\begin{equation}
    Z(\lambda) = \sum_{I \in \mathcal{I}(G)} \prod_{v \in I} \lambda_v.
\end{equation}
Observe from the definition that $Z(0) = 1$, since the product is taken over the empty set.

\paragraph{Cluster Expansion}

The cluster expansion provides a formal power series for $\log Z(\lambda)$. To formulate it, we define a \textit{cluster} as a multiset of vertices, which we can represent as a tuple of multiplicities $\mathbf{m} = (m_v)_{v \in V} \in (\mathbb{Z}_{\geq 0})^V$. 

For any given multiset $\mathbf{m}$, we construct an \textit{incompatibility graph} $H[\mathbf{m}]$ by taking $m_v$ distinguishable copies of each vertex $v \in V$, and drawing an edge between two copies if they correspond to the exact same vertex or to adjacent vertices in $G$.

The cluster expansion of the log partition function is given by:
\begin{equation}
    \log Z(\lambda) = \sum_{\mathbf{m} \neq \mathbf{0}} \frac{1}{\mathbf{m}!} \phi(H[\mathbf{m}]) \prod_{v \in V} \lambda_v^{m_v}
\end{equation}
where $\mathbf{m}! = \prod_{v \in V} m_v!$, and $\phi(H)$ is the Ursell function. For any graph $H = (W, E_H)$, the Ursell function evaluates the sum over all its connected spanning subgraphs:
\begin{equation}
    \phi(H) = \sum_{\substack{A \subseteq E_H \\ (W, A) \text{ is connected}}} (-1)^{|A|}
\end{equation}
Note that if the incompatibility graph $H[\mathbf{m}]$ is disconnected, $\phi(H[\mathbf{m}]) = 0$, which ensures the expansion \emph{only sums over ``connected'' clusters}.

We can think of the cluster expansion as essentially a generalized version of the identity
\[ \log(1 + z) = z - z^2/2 + z^3/3 + \cdots \]
which is convergent when $|z| < 1$. In fact, this is exactly the cluster expansion of the hardcore model on a single-vertex graph, because the corresponding Ursell function is $(-1)^{m - 1} (m - 1)!$.  

\paragraph{First few terms of cluster expansion.}
Let $\lambda_v = \lambda W_v$. The first couple of terms of the cluster expansion are
\[ \log Z(\lambda) = \lambda \sum_v W_v - \lambda^2(\sum_v W_v^2/2 + \sum_{u \sim v}W_uW_v) + O(\lambda^3). \]
Here for the formal power series, when we write $O(\lambda^3)$ it only means that all remaining terms in the formal expansion have degree $3$ or higher.

\paragraph{Koteck\'y--Preiss condition.}
A standard sufficient condition for absolute convergence of the cluster expansion is the Koteck\'y--Preiss (KP) criterion \cite{friedli2017statistical,perkins2023five}. As emphasized earlier, this condition will \emph{not} give very impressive results if directly applied to a hardcore model with random fugacities, since it does not take advantage of phase cancellations. However, the KP condition will end up to be useful to handle auxiliary models which show up in our mathematical analysis.

\begin{prop}[Koteck\'y--Preiss criterion]
Assume there exists a collection of numbers $(a_v)_{v\in V}$ with $a_v \ge 0$ such that for every $v\in V$,
\begin{equation}\label{eq:KP-hardcore}
    \sum_{u \in \{v\}\cup N(v)} |\lambda_u| e^{a_u} \le a_v,
\end{equation}
where $N(v)$ denotes the set of neighbors of $v$ in $G$. Then the cluster expansion for $\log Z(\lambda)$ converges absolutely. More precisely, for every $v\in V$,
\begin{equation}\label{eq:KP-bound-hardcore}
    \sum_{\mathbf{m}:\, m_v\ge 1}
    \frac{1}{\mathbf{m}!}\,
    \bigl|\phi(H[\mathbf{m}])\bigr|
    \prod_{u\in V} |\lambda_u|^{m_u}
    \le a_v.
\end{equation}
In particular, $Z(\lambda)\neq 0$, and hence $\log Z(\lambda)$ is well-defined and analytic in this region.
\end{prop}
Note that by summing \eqref{eq:KP-bound-hardcore} over vertices $v$, we have in particular that
\[ |\log Z(\lambda)| \le \sum_v a_v, \]
so the KP condition is very useful for bounding partition functions.

\paragraph{Spanning tree bound.} The following lemma provides a useful upper bound on the magnitude of the Ursell function. 
\begin{lem}[Lemma 3.4 of \cite{pfister1991large}]
\label{lem:ursell-tree-bound}
Let \(H=(V,E)\) be a finite connected graph.
Then
\[
|\phi(H)|\le \tau(H),
\]
where \(\tau(H)\) denotes the number of spanning trees of \(H\).
In particular, if \(|V|=k\), then
\[
|\phi(H)|\le k^{k-2}.
\]
\end{lem}
This estimate can be recovered as a corollary of the fact that Ursell coefficients admit a representation in terms of spanning trees \cite{penrose1963convergence,pfister1991large}.
\subsection{Wick's formula}
Wick's formula gives a combinatorial expression for Gaussian moments in terms
of sums over pairings. This is a clean example in which moments reduce to
partition functions of an auxiliary system, and it plays an important role in
the second-order analysis of the permanent
(Section~\ref{sec:permanent-second}).

\paragraph{Real case.}
Let $(G_u)_{u\in I}$ be a centered real Gaussian process. Then all odd moments
vanish. Moreover, for any even collection $u_1,\dots,u_{2m}$,
\[
\mathbb{E}\!\left[\prod_{j=1}^{2m} G_{u_j}\right]
=
\sum_{P\in\mathcal{P}_{2m}}
\prod_{\{a,b\}\in P}\mathbb{E}[G_{u_a}G_{u_b}],
\]
where $\mathcal{P}_{2m}$ denotes the set of pairings of $\{1,\dots,2m\}$.

\paragraph{Complex case.}
Let $(G_u)_{u\in I}$ be a centered circularly symmetric complex Gaussian process.
Then
\[
\mathbb{E}\!\left[\prod_{i=1}^m G_{u_i}\prod_{j=1}^n \overline{G_{v_j}}\right]=0
\qquad\text{if } m\neq n,
\]
and when $m=n$,
\[
\mathbb{E}\!\left[\prod_{i=1}^m G_{u_i}\prod_{j=1}^m \overline{G_{v_j}}\right]
=
\sum_{\sigma\in S_m}\prod_{i=1}^m
\mathbb{E}[G_{u_i}\overline{G_{v_{\sigma(i)}}}].
\]

\subsection{Combinatorial identities}
We will use a standard generating function identity for permutations --- see Chapter 2 of \cite{stanley2011enumerative}. It follows from the formula for the number of derangements $D_m$ of $[m]$,
\[ D_{m}=m!\sum_{j=0}^{m}\frac{(-1)^{j}}{j!}, \] 
which is Example 2.2.1 of \cite{stanley2011enumerative}. As a reminder, a derangement is a permutation with no fixed points.
\begin{lem}\label{lem:fixedpts-gen}
Let $\pi\sim\mathsf{Unif}(S_{n})$ and let $\fix(\pi):=\#\{i\in[n]:\pi(i)=i\}$
be the number of fixed points of $\pi$. Then for every $t\in\mathbb{C}$,
\[
\mathbb{E}_{\pi}[t^{\fix(\pi)}]=\sum_{k=0}^{n}\frac{(t-1)^{k}}{k!}.
\]
As a consequence, if $t > 1$ then $\mathbb{E}_{\pi}[t^{\fix(\pi)}] \le e^{t - 1}$.
\end{lem}

\begin{proof}
By decomposing a permutation into its fixed points and a derangement of the remaining points, we find
\[
\mathbb{E}_{\pi}[t^{\fix(\pi)}]=\frac{1}{n!}\sum_{k=0}^{n}\binom{n}{k}D_{n-k}t^{k}.
\]
Using the derangement identity above gives
\begin{align*}
\mathbb{E}_{\pi}[t^{\fix(\pi)}] & =\sum_{k=0}^{n}\frac{t^{k}}{k!}\sum_{j=0}^{n-k}\frac{(-1)^{j}}{j!}
  =\sum_{m=0}^{n}\sum_{k=0}^{m}\frac{t^{k}}{k!}\frac{(-1)^{m-k}}{(m-k)!}
  =\sum_{m=0}^{n}\frac{1}{m!}\sum_{k=0}^{m}\binom{m}{k}t^{k}(-1)^{m-k} \\
  &=\sum_{m=0}^{n}\frac{(t-1)^{m}}{m!},
\end{align*}
and the inequality when $t > 1$ follows by Taylor series expansion of the exponential.
\end{proof}
\begin{rem}
This generating function identity is one example illustrating how the combinatorics of summing over permutations has nice algebraic properties. Because the general hardcore model (which instead sums over matchings on a general graph) does not have this nice permutation structure, we will see that our analysis changes significantly.
\end{rem}

\paragraph{Stirling's formula.}
Since we will be working with permutations, it is helpful to remember Stirling's approximation:
\[ \log(n!) = n\log(n/e) + \log(2\pi n)/2 + o(1) \]
as $n \to \infty$.

\section{First-order permanent analysis}\label{sec:easiest}
In this section, we show how to analyze the zero-free region for matrices with complex Gaussian entries up to $|z| = \tilde{O}(n^{1/4})$. This result will be improved in later sections using more sophisticated arguments, so it is helpful to go through this simpler result first.

\paragraph{Setup.}
Let \(W \in \C^{n\times n}\) be a random matrix with i.i.d. entries
$W_{ij}\sim \CN(0,1)$,
where by \(\CN(0,1)\) we mean the standard complex Gaussian with density
$\frac{1}{\pi}e^{-|w|^2}$ for $w\in\C$.
Let \(J\) denote the \(n\times n\) all-ones matrix, and define
\[
P_n(z)
:=\frac{1}{n!}\per(J+zW)
=\E_{\sigma\sim S_n}\Big[\prod_{i=1}^n \bigl(1+zW_{i,\sigma(i)}\bigr)\Big].
\]
We also define the global linear statistic
\[
D_1(W):=\frac{1}{n}\sum_{i,j=1}^n W_{ij},
\]
and the first-order reweighted permanent
\[
X^{(1)}(z):=P_n(z)\exp\bigl(-zD_1(W)\bigr).
\]

\paragraph{First moment.}
Observe that 
$\mathbb E X^{(1)}(z) = 1$
due to the circular symmetry of the complex Gaussian.  Explicitly, if we Taylor expand
\[
X^{(1)}(z)=\sum_{k\ge 0} c_k(W)\, z^k,
\]
then each \(c_k\) is a homogeneous polynomial of degree \(k\) in the entries of \(W\), so
$c_k(e^{i\theta}W)=e^{ik\theta}c_k(W)$.
Since \(e^{i\theta}W \stackrel{d}{=} W\) for every \(\theta\in\R\), we obtain
\[
\E[c_k(W)] = \E[c_k(e^{i\theta}W)] = e^{ik\theta}\E[c_k(W)].
\]
Hence \(\E[c_k(W)]=0\) for all \(k\ge 1\), while \(c_0(W)=1\).
\subsection{Tilted measure and exact moment formula}

We define the tilted expectation \(\widehat{\E}\) by
\[
\widehat{\E}[f(W)]
:=
\frac{\E\!\left[f(W)\exp\!\bigl(-2\Re(zD_1(W))\bigr)\right]}
{\E\!\left[\exp\!\bigl(-2\Re(zD_1(W))\bigr)\right]}.
\]

\begin{lem}[Linear tilt of a complex Gaussian product measure]
\label{lem:linear_tilt_clean}
Under the tilted law \(\widehat{\E}\), the entries \(W_{ij}\) remain independent complex Gaussians, each with variance \(1\), vanishing pseudo-variance, and common mean
\[
\widehat{\E}[W_{ij}]
=
-\frac{\overline z}{n}.
\]
In other words, under $\widehat{\E}$ we have $W_{ij}\sim \CN\!\left(-\frac{\overline z}{n},\,1\right)$ i.i.d. 
Moreover,
\[
M(z):=\E\!\left[\exp\!\bigl(-2\Re(zD_1(W))\bigr)\right]=e^{|z|^2}.
\]
\end{lem}

\begin{proof}
Since the entries are independent, it suffices to work entrywise. Write
$u:=z/n$, so 
\[
2\Re(zD_1(W))
=
2\Re\!\left(\frac{z}{n}\sum_{i,j}W_{ij}\right)
=
\sum_{i,j}2\Re(uW_{ij}).
\]
Hence
\[
\exp\!\bigl(-2\Re(zD_1(W))\bigr)
=
\prod_{i,j}\exp\!\bigl(-2\Re(uW_{ij})\bigr).
\]
Therefore the tilt factorizes over entries, so under the tilted law the entries remain independent.

Now fix one entry \(W\sim \CN(0,1)\), with density \(\pi^{-1}e^{-|w|^2}\). Under the tilt by \(e^{-2\Re(uw)}\), its density becomes proportional to
\[
e^{-|w|^2}e^{-2\Re(uw)}
=
e^{-(|w|^2+2\Re(uw))}
=
e^{-|w+\overline u|^2+|u|^2}.
\]
Thus, after normalization, the tilted law is exactly \(\CN(-\overline u,1)\). Since \(u=z/n\), this gives
\[
\widehat{\E}[W_{ij}] = -\frac{\overline z}{n}.
\]
The variance and pseudo-variance are unchanged because only the mean is shifted.

Finally, the normalizing constant is
\[
M(z)
=
\prod_{i,j}\E\bigl[e^{-2\Re(uW_{ij})}\bigr].
\]
For one standard complex Gaussian \(W\),
$\E[e^{-2\Re(uW)}]
=
e^{|u|^2}$,
by the same completion-of-squares calculation. Therefore
$M(z)=\bigl(e^{|u|^2}\bigr)^{n^2}
=
e^{n^2|z|^2/n^2}
=
e^{|z|^2}$.
\end{proof}

\begin{lem}[Exact second moment]
\label{lem:exact_moment_clean}
Let
\[
\alpha:=\frac{|z|^2}{n},
\qquad
\kappa=1-\alpha.
\]

Then
\[
\E\bigl[|X^{(1)}(z)|^2\bigr]
=
e^{|z|^2}\,\kappa^{2n}\sum_{k=0}^n \frac{1}{k!}
\left(\frac{|z|^2}{\kappa^2}\right)^k.
\]
\end{lem}

\begin{proof}
By definition of the tilted expectation,
\[
\E\bigl[|X^{(1)}(z)|^2\bigr]
=
\E\!\left[|P_n(z)|^2e^{-2\Re(zD_1(W))}\right]
=
M(z)\,\widehat{\E}\bigl[|P_n(z)|^2\bigr].
\]
By Lemma~\ref{lem:linear_tilt_clean}, \(M(z)=e^{|z|^2}\), so it remains to compute \(\widehat{\E}[|P_n(z)|^2]\).

Under \(\widehat{\E}\), write
\[
W_{ij}=-\frac{\overline z}{n}+\eta_{ij},
\]
where \(\eta_{ij}\) are i.i.d. \(\CN(0,1)\). Then
\[
1+zW_{ij}
=
1-\frac{|z|^2}{n}+z\eta_{ij}
=
\kappa+z\eta_{ij}.
\]
Hence
$P_n(z)
=
\E_{\sigma\sim S_n}\Big[\prod_{i=1}^n (\kappa+z\eta_{i,\sigma(i)})\Big]$
and so
\[
\widehat{\E}[|P_n(z)|^2]
=
\E_{\sigma,\tau\sim S_n}
\widehat{\E}\Bigg[
\prod_{i=1}^n (\kappa+z\eta_{i,\sigma(i)})
\prod_{j=1}^n (\kappa+\overline z\,\overline{\eta_{j,\tau(j)}})
\Bigg].
\]

Fix \(\sigma,\tau\in S_n\). Let
$F(\sigma,\tau):=\{i\in[n]:\sigma(i)=\tau(i)\}$
and
$m:=|F(\sigma,\tau)|$.
Since \(\sigma\) and \(\tau\) are permutations, the edge sets
\[
\{(i,\sigma(i)):1\le i\le n\},
\qquad
\{(j,\tau(j)):1\le j\le n\}
\]
intersect exactly on the \(m\) common edges indexed by \(F(\sigma,\tau)\).

Because the \(\eta_{ab}\) are independent, centered, and satisfy
\[
\widehat{\E}[\eta_{ab}\overline{\eta_{cd}}]=\delta_{ac}\delta_{bd},
\qquad
\widehat{\E}[\eta_{ab}\eta_{cd}]=0,
\]
the expectation factors edge-by-edge. For an edge appearing in only one of the two products, the only surviving contribution comes from taking the constant term \(\kappa\). For a common edge, the contribution is
\[
\widehat{\E}\bigl[(\kappa+z\eta)(\kappa+\overline z\,\overline\eta)\bigr]
=
\kappa^2+|z|^2.
\]
Thus
\[
\widehat{\E}\Bigg[
\prod_{i=1}^n (\kappa+z\eta_{i,\sigma(i)})
\prod_{j=1}^n (\kappa+\overline z\,\overline{\eta_{j,\tau(j)}})
\Bigg]
=
(\kappa^2+|z|^2)^m\,\kappa^{2(n-m)}.
\]
Equivalently,
\[
\widehat{\E}[\cdots]
=
\kappa^{2n}\left(1+\frac{|z|^2}{\kappa^2}\right)^m.
\]
Therefore
\[
\widehat{\E}[|P_n(z)|^2]
=
\kappa^{2n}\,
\E_{\sigma,\tau}\!\left[\left(1+\frac{|z|^2}{\kappa^2}\right)^m\right].
\]

Now if \(\sigma,\tau\) are independent uniform permutations, then \(\pi:=\sigma\tau^{-1}\) is a uniform permutation, and
$m=\#\{i:\sigma(i)=\tau(i)\}=\#\{i:\pi(i)=i\}$ is
the number of fixed points of a uniform permutation. 
Applying Lemma~\ref{lem:fixedpts-gen} with
$t=1+\frac{|z|^2}{\kappa^2}$
gives
\[
\widehat{\E}[|P_n(z)|^2]
=
\kappa^{2n}
\sum_{k=0}^n \frac{1}{k!}
\left(\frac{|z|^2}{\kappa^2}\right)^k.
\]
Multiplying by \(M(z)=e^{|z|^2}\) completes the proof.
\end{proof}

\subsection{Variance asymptotics}

\begin{thm}
\label{thm:first_order_complex_clean}
If \(|z|=o(n^{1/4})\), then
\[
\E\bigl[|X^{(1)}(z)|^2\bigr]=1+o(1).
\]
\end{thm}

\begin{proof}
Since \(|z|=o(n^{1/4})\), we have
$\frac{|z|^4}{n}\to 0$ and
$\alpha = |z|^2/n \to 0$.
By Lemma~\ref{lem:exact_moment_clean},
\[
\E\bigl[|X^{(1)}(z)|^2\bigr]
=
e^{|z|^2}\kappa^{2n}\sum_{k=0}^n \frac{y^k}{k!},
\qquad
y:=\frac{|z|^2}{\kappa^2}.
\]

We first expand the prefactor. Since \(\alpha\to 0\),
$\log(1-\alpha)=-\alpha-\frac{\alpha^2}{2}+O(\alpha^3)$,
so
\begin{align*}
|z|^2+2n\log(1-\alpha)
=
|z|^2-2n\alpha-n\alpha^2+O(n\alpha^3)
=
-|z|^2-\frac{|z|^4}{n}
+O\!\left(\frac{|z|^6}{n^2}\right).
\end{align*}
Hence
\[
e^{|z|^2}\kappa^{2n}
=
\exp\!\left(
-|z|^2-\frac{|z|^4}{n}
+O\!\left(\frac{|z|^6}{n^2}\right)
\right).
\]

Next,
\[
y
=
\frac{|z|^2}{(1-\alpha)^2}
=
|z|^2(1+2\alpha+O(\alpha^2))
=
|z|^2+\frac{2|z|^4}{n}
+O\!\left(\frac{|z|^6}{n^2}\right).
\]
Also \(y=o(n)\), because \(|z|^2=o(n^{1/2})\). Therefore the truncated exponential sum satisfies
\[
\sum_{k=0}^n \frac{y^k}{k!}
=
e^y(1-o(1)).
\]
Indeed, the omitted tail $e^{-y} \sum_{k = n + 1}^{\infty} y^k/k!$ is the probability that a Poisson\((y)\) random variable exceeds \(n\), which is  \(o(1)\) when \(y=o(n)\) --- e.g., by Chebyshev or Chernoff bounds.
Thus
\[
\sum_{k=0}^n \frac{y^k}{k!}
=
\exp\!\left(
|z|^2+\frac{2|z|^4}{n}
+O\!\left(\frac{|z|^6}{n^2}\right)
\right)(1-o(1)).
\]

Combining the two estimates,
\begin{align*}
\E\bigl[|X^{(1)}(z)|^2\bigr]
&=
\exp\!\left(
-|z|^2-\frac{|z|^4}{n}
+O\!\left(\frac{|z|^6}{n^2}\right)
\right)
\cdot
\exp\!\left(
|z|^2+\frac{2|z|^4}{n}
+O\!\left(\frac{|z|^6}{n^2}\right)
\right)(1-o(1)) \\
&=
\exp\!\left(
\frac{|z|^4}{n}
+O\!\left(\frac{|z|^6}{n^2}\right)
\right)(1-o(1)).
\end{align*}
Since \(|z|^4/n\to 0\), the exponent tends to \(0\). Therefore
$\E\bigl[|X^{(1)}(z)|^2\bigr]=1+o(1)$ as claimed.
\end{proof}

From the above result, we can immediately conclude the corresponding zero-free region by applying Jensen's formula (more precisely, Lemma~\ref{lem:core-jensen-app}), which in turn yields an algorithm via Barvinok interpolation. Since we are about to improve on the above result, we leave further details to the next section.

\section{Second-order permanent analysis}\label{sec:permanent-second}
We now proceed to strengthen the result in the previous section by taking into account the quadratic term in the expansion of the log-permanent. 
In order to make the calculations clearer, it is helpful to write the second-order reweighting in a geometric way.

\paragraph{Second-order reweighting in terms of subspaces.}
Let $W\in\mathbb{C}^{n\times n}$ have i.i.d. entries $W_{ij}\sim\mathcal{N}_{\mathbb{C}}(0,\nu)$.
Let $\mathbb{C}^{n\times n}=\mathcal{H}_{0}\oplus\mathcal{H}_{r}\oplus\mathcal{H}_{c}\oplus\mathcal{H}_{m}$
where
\begin{align*}
\mathcal{H}_{0} & :=\operatorname{span}\{J\},\\
\mathcal{H}_{r} & :=\operatorname{span}\{x\mathbf{1}^{T}:x\in\mathbb{C}^{n},\mathbf{1}^{T}x=0\},\\
\mathcal{H}_{c} & :=\operatorname{span}\{\mathbf{1}y^{T}:y\in\mathbb{C}^{n},\mathbf{1}^{T}y=0\},\\
\mathcal{H}_{m} & :=\operatorname{span}\{M\in\mathbb{C}^{n\times n}:M\mathbf{1}=0,\mathbf{1}^{T}M=0\},
\end{align*}
and $P_{0},P_{r},P_{c},P_{m}$ are the corresponding orthogonal projections.
Define $P_{\parallel}:=P_{0}+P_{r}+P_{c}$, so that $P_{m}+P_{\parallel}=I$.
Let $w=\operatorname{vec}(W)\in\mathbb{C}^{n^{2}}$ and $a:=\frac{1}{n}\mathbf{1}$,
and write 
\[
D_{1}(W):=a^{T}w,\quad D_{2}(W):=w^{T}Bw,
\]
where $B:=-\frac{1}{n}P_{\parallel}+\frac{1}{n(n-1)}P_{m}$. These two terms are exactly what one gets by Taylor expanding the log-permanent: explicitly,
 $D_{1}(W)=(\log G_{W}(z))'(0)$ and $D_{2}(W)=(\log G_{W}(z))''(0)$,
where 
\[ G_{W}(z)=\frac{1}{n!}\per(J+zW). \]
Given these definitions, the second-order reweighted permanent is given by
\[
X_{W}^{(2)}(z):=\exp\left(-zD_{1}(W)-\frac{z^{2}}{2}D_{2}(W)\right)G_{W}(z).
\]

\paragraph{First moment.} Note that by circular symmetry of the complex Gaussian law, \(\E X_W^{(2)}(z)=1\). Indeed, if
\[
X_W^{(2)}(z)=\sum_{k\ge 0} c_k(W)\, z^k,
\]
then each \(c_k\) is a homogeneous polynomial of degree \(k\) in the entries of \(W\), so
$c_k(e^{i\theta}W)=e^{ik\theta}c_k(W)$.
Since \(e^{i\theta}W \stackrel{d}{=} W\) for every \(\theta\in\R\), we obtain
\[
\E[c_k(W)] = \E[c_k(e^{i\theta}W)] = e^{ik\theta}\E[c_k(W)].
\]
Hence \(\E[c_k(W)]=0\) for all \(k\ge 1\), while \(c_0(W)=1\). Therefore $\E X_W^{(2)}(z)=1$.
\subsection{Reweighted second moment analysis}
\subsubsection{Exact formulas}
In this section, we show how to give a nice formula for the reweighted second moment by reinterpreting it in terms of a ``tilted'' Gaussian distribution and applying Wick's formula.

\paragraph{Tilted expectation.}
Define the tilted expectation
\[
\widehat{\mathbb{E}}[f(W)]:=\frac{\mathbb{E}[e^{-2\operatorname{Re}(zD_{1})-\operatorname{Re}(z^{2}D_{2})}f(W)]}{\mathbb{E}[e^{-2\operatorname{Re}(zD_{1})-\operatorname{Re}(z^{2}D_{2})}]}.
\]
Also, for $t,s\in\mathbb{C}^{n^{2}}$, define
\[
M(t,s):=\mathbb{E}\exp\left(t^{T}w+s^{T}\overline{w}-\operatorname{Re}(z^{2}w^{T}Bw)\right).
\]

\begin{prop}
We have
\[
\mathbb{E}|X_{W}^{(2)}(z)|^{2}=M(-za,-\overline{z}a)\mathbb{E}_{\sigma,\tau\sim\mathsf{Unif}(S_{n})}\left[\widehat{\mathbb{E}}\left[\prod_{i=1}^{n}(1+zw_{e_{i}})(1+\overline{z}\overline{w_{f_{i}}})\right]\right]
\]
where $e_{i}=(i,\sigma(i))$, $f_{i}=(i,\tau(i))$.
\end{prop}

\begin{proof}
Note that 
\[
\mathbb{E}|X_{W}^{(2)}(z)|^{2}=\frac{1}{(n!)^{2}}\sum_{\sigma,\tau\in S_{n}}T_{\sigma,\tau}(z)=\mathbb{E}_{\sigma,\tau\sim\mathsf{Unif}(S_{n})}[T_{\sigma,\tau}(z)],
\]
where
\begin{align*}
T_{\sigma,\tau}(z) & :=\mathbb{E}\left[\exp\left(-2\operatorname{Re}(za^{T}w)-\operatorname{Re}(z^{2}w^{T}Bw)\right)\prod_{i=1}^{n}(1+zw_{e_{i}})(1+\overline{z}\overline{w_{f_{i}}})\right]\\
 & =M(-za,-\overline{z}a)\widehat{\mathbb{E}}\left[\prod_{i=1}^{n}(1+zw_{e_{i}})(1+\overline{z}\overline{w_{f_{i}}})\right].
\end{align*}
\end{proof}
\begin{lem}
Let $S:=I-\nu^{2}|z|^{4}B^{2}$, $K:=BS^{-1}$. If $S\succ0$, then
\[
M(t,s)=\det(S)^{-1/2}\exp\left(\nu t^{T}S^{-1}s-\frac{\nu^{2}}{2}\overline{z}^{2}t^{T}Kt-\frac{\nu^{2}}{2}z^{2}s^{T}Ks\right).
\]
\end{lem}

\begin{proof}
Diagonalize $B=U\Lambda U^{T}$ with real orthogonal $U$. Set $g:=U^{T}w$,
$t':=U^{T}t$, $s':=U^{T}s$. Since $w\sim\mathcal{N}_{\mathbb{C}}(0,\nu I)$
and $U$ is orthogonal, $g_{k}$ are independent $\mathcal{N}_{\mathbb{C}}(0,\nu)$.
The integral factorizes into one-dimensional complex Gaussian integrals:
\[
M(t,s)=\prod_{k=1}^{N}I_{\lambda_{k}}(t_{k}',s_{k}')
\]
with $N = n^2$ and
\begin{align*}
I_{\lambda}(t,s) & :=\mathbb{E}\exp\left(tg+s\overline{g}-\frac{1}{2}z^{2}\lambda g^{2}-\frac{1}{2}\overline{z}^{2}\lambda\overline{g}^{2}\right)\\
 & =\frac{1}{\sqrt{1-\nu^{2}|z|^{4}\lambda^{2}}}\exp\left(\frac{\nu ts}{1-\nu^{2}|z|^{4}\lambda^{2}}-\frac{\nu^{2}}{2}\frac{\overline{z}^{2}\lambda t^{2}}{1-\nu^{2}|z|^{4}\lambda^{2}}-\frac{\nu^{2}}{2}\frac{z^{2}\lambda s^{2}}{1-\nu^{2}|z|^{4}\lambda^{2}}\right).
\end{align*}
Then
\[
\sum_{k=1}^{N}\frac{\nu t_{k}'s_{k}'}{1-\nu^{2}|z|^{4}\lambda_{k}^{2}}=\nu t^{T}U(I-\nu^{2}|z|^{4}\Lambda^{2})^{-1}U^{T}s=\nu t^{T}S^{-1}s,
\]
\[
\sum_{k=1}^{N}-\frac{\nu^{2}}{2}\frac{\overline{z}^{2}\lambda_{k}t_{k}^{2}}{1-\nu^{2}|z|^{4}\lambda_{k}^{2}}=-\frac{\nu^{2}}{2}\overline{z}^{2}t^{T}U\Lambda(I-\nu^{2}|z|^{4}\Lambda^{2})^{-1}U^{T}t=-\frac{\nu^{2}}{2}\overline{z}^{2}t^{T}Kt,
\]
\[
\sum_{k=1}^{N}-\frac{\nu^{2}}{2}\frac{z^{2}\lambda_{k}s_{k}^{2}}{1-\nu^{2}|z|^{4}\lambda_{k}^{2}}=-\frac{\nu^{2}}{2}z^{2}s^{T}Ks,
\]
so multiplying over $k$ yields the desired matrix form. 
\end{proof}
Let 
\[ \alpha:=\nu|z|^{2}/n \]
be a parameter which will ultimately be required to be smaller than some absolute constant. 
Since $S=(1-\alpha^{2})P_{\parallel}+\left(1-\frac{\alpha^{2}}{(n-1)^{2}}\right)P_{m}$,
we know that $S\succ0$ provided we require $\alpha<1$. 
\begin{lem}
Let 
\begin{equation}\label{eqn:beta-kappa}
\beta:=\alpha/(1-\alpha) \quad\text{and}\quad \kappa:=1-\beta=\frac{1-2\alpha}{1-\alpha}. 
\end{equation}
For $\alpha<1/2$, we have $M(-za,-\overline{z}a)=\det(S)^{-1/2}e^{n\beta}$,
\[
\widehat{\mathbb{E}}[w]=-\frac{\overline{z}\nu}{1-\alpha}a,\quad\Gamma:=\widehat{\mathbb{E}}[(w-\widehat{\mathbb{E}}w)(\overline{w-\widehat{\mathbb{E}}w})^{T}]=\nu S^{-1},\quad\Delta:=\widehat{\mathbb{E}}[(w-\widehat{\mathbb{E}}w)(w-\widehat{\mathbb{E}}w)^{T}]=-\nu^{2}\overline{z}^{2}K.
\]
Consequently, with $\eta:=w-\widehat{\mathbb{E}}w$, for every coordinate
$e$,
\[
1+zw_{e}=\kappa+z\eta_{e},\quad1+\overline{z}\overline{w_{e}}=\kappa+\overline{z}\overline{\eta_{e}}.
\]
\end{lem}

\begin{proof}
We apply the previous lemma with $(t,s)=(-za,-\overline{z}a)$. Note
that
\[
S^{-1}=\frac{1}{1-\alpha^{2}}P_{\parallel}+\frac{1}{1-\alpha^{2}/(n-1)^{2}}P_{m}
\]
from which it follows that $S^{-1}a=(1-\alpha^{2})^{-1}a$. Also,
$Ba=-\frac{1}{n}a$ and $a^{T}a=1$, so we get
\[
M(-za,-\overline{z}a)=\det(S)^{-1/2}\exp\left(\frac{\nu|z|^{2}}{1-\alpha^{2}}+\frac{\nu^{2}|z|^{4}}{n(1-\alpha^{2})}\right)=\det(S)^{-1/2}e^{n\beta}.
\]
The mean and covariances are the first/second derivatives of $\log M(t-za,s-\overline{z}a)$
at $(0,0)$. The identity for $\kappa$ follows from $z\widehat{\mathbb{E}}[w_{e}]=-\nu|z|^{2}/(n(1-\alpha))=-\beta$.
\end{proof}
Fix $(\sigma,\tau)$ and set 
\[
X_{i}:=\eta_{e_{i}}\quad\text{and}\quad Y_{j}:=\overline{\eta_{f_{j}}}.
\]
Then under the tilted law, $(X_{1},\ldots,X_{n},Y_{1},\ldots,Y_{n})$
is centered Gaussian with covariances
\[
\widehat{\mathbb{E}}[X_{i}Y_{j}]=\Gamma_{e_{i},f_{j}},\quad\widehat{\mathbb{E}}[X_{i}X_{i'}]=\Delta_{e_{i},e_{i'}}\quad\widehat{\mathbb{E}}[Y_{j}Y_{j'}]=\overline{\Delta_{f_{j},f_{j'}}}.
\]
From Wick's formula, we get the following. 
\begin{lem} \label{lem:wick} 
\[
\widehat{\mathbb{E}}\left[\prod_{i=1}^{n}(\kappa+zX_{i})\prod_{j=1}^{n}(\kappa+\overline{z}Y_{j})\right]=\kappa^{2n}\sum_{M\in\mathcal{M}}\prod_{e\in M}w_{e},
\]
where $\mathcal{M}$ is the set of partial matchings on $\{1_{t},\ldots,n_{t},1_{s},\ldots,n_{s}\}$,
and edge weights are
\[
w_{(i_{t},j_{s})}=\frac{|z|^{2}\Gamma_{e_{i},f_{j}}}{\kappa^{2}},\quad w_{(i_{t},i'_{t})}=\frac{z^{2}\Delta_{e_{i},e_{i'}}}{\kappa^{2}},\quad w_{(j_{s},j'_{s})}=\frac{\overline{z}^{2}\overline{\Delta_{f_{j},f_{j'}}}}{\kappa^{2}}.
\]
Here, $i_{t}$ corresponds to the $i$th term in $\prod_{i=1}^{n}(\kappa+zX_{i})$
while $j_{s}$ corresponds to the $j$th term in $\prod_{j=1}^{n}(\kappa+\overline{z}Y_{j})$. 
\end{lem}

\begin{proof}
Let
\[
V_{i_{t}}:=X_{i},\quad V_{j_{s}}:=Y_{j},\quad\xi_{i_{t}}:=\frac{z}{\kappa},\quad\xi_{j_{s}}:=\frac{\overline{z}}{\kappa}.
\]
Then
\[
\prod_{i=1}^{n}(\kappa+zX_{i})\prod_{j=1}^{n}(\kappa+\overline{z}Y_{j})=\kappa^{2n}\prod_{u\in U}(1+\xi_{u}V_{u}),
\]
where $U=\{1_{t},\ldots,n_{t},1_{s},\ldots,n_{s}\}$. Expanding, 
\[
\widehat{\mathbb{E}}\left[\prod_{u\in U}(1+\xi_{u}V_{u})\right]=\sum_{A\subset U}\left(\prod_{u\in A}\xi_{u}\right)\widehat{\mathbb{E}}\left[\prod_{u\in A}V_{u}\right].
\]
Since $V$ is centered jointly Gaussian, Wick's formula gives that the expectation
vanishes for $|A|$ odd, and for $|A|=2m$,
\[
\widehat{\mathbb{E}}\left[\prod_{u\in A}V_{u}\right]=\sum_{P\in\mathcal{P}(A)}\prod_{\{u,v\}\in P}\widehat{\mathbb{E}}[V_{u}V_{v}],
\]
with $\mathcal{P}(A)$ the matchings/pairings of $A$. Summing over
all $A$ and pairings is equivalent to summing over all partial matchings
$M$ on $U$; each unmatched vertex contributes 1, and each matched
edge contributes
\[
\xi_{u}\xi_{v}\widehat{\mathbb{E}}[V_{u}V_{v}].
\]
For $(i_{t},j_{s})$ this is $\frac{|z|^{2}}{\kappa^{2}}\widehat{\mathbb{E}}[X_{i}Y_{j}]=\frac{|z|^{2}\Gamma_{e_{i},f_{j}}}{\kappa^{2}}$;
similarly for $(i_{t},i'_{t})$ and $(j_{s},j_{s}')$ one gets the
stated $tt$ and $ss$ weights. Hence
\[
\widehat{\mathbb{E}}\left[\prod_{i=1}^{n}(\kappa+zX_{i})\prod_{j=1}^{n}(\kappa+\overline{z}Y_{j})\right]=\kappa^{2n}\sum_{M\in\mathcal{M}}\prod_{e\in M}w_{e},
\]
\end{proof}
\begin{rem}[Edge weights are real]
Since \(\Gamma=\nu S^{-1}\) and \(\Delta=-\nu^{2}\overline z^{\,2}K\), with
\(S^{-1}\) and \(K\) real symmetric matrices, all of these edge weights are
in fact real. More explicitly,
\[
w_{(i_t,j_s)}=\frac{\nu |z|^{2}}{\kappa^{2}}(S^{-1})_{e_i,f_j}\in\mathbb{R},
\qquad
w_{(i_t,i'_t)}=w_{(j_s,j'_s)}
=-\frac{\nu^{2}|z|^{4}}{\kappa^{2}}K_{e_i,e_{i'}}\in\mathbb{R}.
\]
\end{rem}
\subsubsection{Variance asymptotics}
From the above application of Wick's formula, we have reduced the problem of computing the second moment
to combinatorial sums over matchings. These sums are suitable for analysis via the cluster expansion and KP condition.
\begin{lem}[cluster expansion for matching partition function] \label{lem:cluster}
Let $G=(V,E)$ be a finite simple graph with complex activities $w=(w_{e})_{e\in E}$
and matching partition function
\[
Z_{G}(w):=\sum_{M\text{ is a matching in }G}\prod_{e\in M}w_{e}.
\]
Define
\[
\delta:=\max_{v\in V}\sum_{e\ni v}|w_{e}|.
\]
If $\delta\leq1/(4e)$, then $Z_{G}(w)\neq0$ and the branch of $\log Z_{G}$
with $\log Z_{G}(0)=0$ satisfies
\[
\left|\log Z_{G}(w)-\sum_{e\in E}w_{e}\right|\leq\frac{|V|}{4}\cdot\frac{(2e\delta)^{2}}{1-2e\delta}\leq2e^{2}|V|\delta^{2}.
\]
\end{lem}

\begin{proof}
A matching in $G$ is exactly an independent set in the line graph
$L(G)$ (vertices of $L(G)$ are edges of $G$, and adjacency means
sharing an endpoint). In particular, we can rewrite $Z_{G}$ as a
multivariate independence polynomial
\[
Z_{G}(w)=Z_{L(G)}(\lambda),\quad\lambda_{e}:=w_{e},
\]
where $Z_{L(G)}$ is the multivariate independence polynomial. A standard
consequence of the KP condition (explicitly, Theorem 3.3 in \cite{perkins2023five}) gives that if for some $b\geq0$ and every vertex $v$ we have
\[
\sum_{u\in N_{L(G)}(v)\cup\{v\}}|\lambda_{u}|e^{1+b}\leq1,
\]
then the cluster expansion for $\log Z_{L(G)}(\lambda)$ converges
absolutely and, in particular, $Z_{L(G)}(\lambda)\neq0$. For an edge
$e=xy\in E(G)$ (a vertex in $L(G)$), its closed neighborhood in
$L(G)$ corresponds to edges of $G$ incident to $x$ or $y$. Hence
\[
\sum_{u\in N_{L(G)}(e)\cup\{e\}}|\lambda_{u}|\leq\sum_{f\ni x}|w_{f}|+\sum_{f\ni y}|w_{f}|\leq2\delta.
\]
So the KP condition holds (with $b=0$) whenever $2\delta e\leq1$,
i.e. $\delta\leq1/(2e)$. When this holds, $Z_{G}(w)\neq0$. Define
$Z(t):=Z_{G}(tw)$. By the KP condition (with $b=0$) applied to $tw$,
we have $Z(t)\neq0$ for all $|t|\leq1/(2e\delta)$. $Z(t)$ is a
polynomial of degree at most the maximum matching size, so $\deg Z(t)\leq|V|/2$.
Factor
\[
Z(t)=\prod_{j=1}^{N}(1-t/r_{j})
\]
with $N\leq|V|/2$. Zero-freeness on $|t|\leq1/(2e\delta)$ implies
$|r_{j}|\geq1/(2e\delta)$, i.e. $|1/r_{j}|\leq2e\delta=:\eta$. Then
\[
\log Z(1)-\sum_{e\in E}w_{e}=\sum_{j=1}^{N}\left(\log\left(1-\frac{1}{r_{j}}\right)+\frac{1}{r_{j}}\right)=-\sum_{j=1}^{N}\sum_{k\geq2}\frac{1}{kr_{j}^{k}}.
\]
Taking absolute values and using $\sum_{k\geq2}\eta^{k}/k\leq\eta^{2}/(2(1-\eta))$
gives
\[
\left|\log Z_{G}(w)-\sum_{e\in E}w_{e}\right|\leq N\cdot\frac{\eta^{2}}{2(1-\eta)}\leq\frac{|V|}{2}\cdot\frac{(2e\delta)^{2}}{2(1-2e\delta)}=\frac{|V|}{4}\cdot\frac{(2e\delta)^{2}}{1-2e\delta}.
\]
Under the assumption that $\delta\leq1/(4e)$, this is bounded above
by $2e^{2}|V|\delta^{2}$. 
\end{proof}
Write $F(\sigma,\tau):=\{i\in[n]:\sigma(i)=\tau(i)\}$, and let $m:=|F(\sigma,\tau)|$.
For $i\in F(\sigma,\tau)$, $e_{i}=f_{i}$, and the diagonal $ts$-edge
weight is
\[
w_{\text{diag}}:=\frac{|z|^{2}\Gamma_{e,e}}{\kappa^{2}}.
\]
Note that this is independent of $e$, and positive provided that $\alpha < 1$, because $\Gamma_{e,e} = \nu (S^{-1})_{ee}$ and $S$ is positive definite. 
\begin{lem}
\label{lem:diagonal-decomposition} Fix $(\sigma,\tau)$ and define
\[
Z_{\sigma,\tau}:=\sum_{M\in\mathcal{M}}\prod_{e\in M}w_{e}.
\]
Then
\[
Z_{\sigma,\tau}=\sum_{S\subset F(\sigma,\tau)}w_{\mathrm{diag}}^{|S|}Z_{\sigma,\tau}^{(S)}
\]
where $Z_{\sigma,\tau}^{(S)}$ is the matching partition function
on the reduced graph obtained by starting from the complete graph
on $\{1_{t},\ldots,n_{t},1_{s},\ldots,n_{s}\}$ then deleting $\{i_{t},i_{s}:i\in S\}$
and deleting the remaining edges $\{(i_{t},i_{s}):i\in F(\sigma,\tau)\backslash S\}$. 
\end{lem}

\begin{proof}
The diagonal edges $(i_{t},i_{s})$ for $i\in F(\sigma,\tau)$ are
pairwise disjoint. Partition matchings by
\[
S(M):=\{i\in F(\sigma,\tau):(i_{t},i_{s})\in M\}.
\]
If $S(M)=S$, then the matching contributes $w_{\mathrm{diag}}^{|S|}$
from these chosen edges, and the remaining edges form a matching on
the stated reduced graph. Summing over all such matchings gives $Z_{\sigma,\tau}^{(S)}$. 
\end{proof}
\begin{lem}
[Reduced-graph bound] \label{lem:reduced-graph-local}
Assume $\alpha\leq1/4$ and write
\[
\lambda:=\frac{1}{1-\alpha^{2}},\quad\mu:=\frac{1}{1-\alpha^{2}/(n-1)^{2}},\quad K:=BS^{-1}.
\]
Define
\[
u:=\frac{\nu^{2}|z|^{4}}{\kappa^{2}}\left(\frac{\lambda}{n^{3}}+\frac{\mu}{n^{3}(n-1)}\right).
\]
Then the total sum of the weights of all internal $tt$ edges and
internal $ss$ edges in the unreduced graph is given by
\[
\Sigma_{tt+ss}(n)=-n(n-1)u=-\frac{\alpha^{2}}{\kappa^{2}}\left((n-1)\lambda+\mu\right).
\]
Moreover, for every $(\sigma,\tau)$ and $S\subset F(\sigma,\tau)$,
\[
\delta_{(S)}:=\max_{v\in V^{(S)}}\sum_{e\in E^{(S)}:e\ni v}|w_{e}|\leq C\alpha^{2},
\]
\[
\sum_{e\in E^{(S)}}w_{e}\leq\Sigma_{tt+ss}(n)+C\alpha^{2}|S|+C\alpha^{3}n,
\]
where $(V^{(S)},E^{(S)})$ is the reduced graph from the previous
lemma. 
\end{lem}

\begin{proof}
\textbf{Step 1: exact formula for $\Sigma_{tt+ss}(n)$.} For $i\neq j$,
permutation edges $e_{i},e_{j}$ have distinct row and distinct column,
hence
\[
(P_{\parallel})_{e_{i},e_{j}}=-\frac{1}{n^{2}},\quad(P_{m})_{e_{i},e_{j}}=+\frac{1}{n^{2}}.
\]
Since
\[
K=\Big(-\frac{\lambda}{n}\Big)P_{\parallel}+\Big(\frac{\mu}{n(n-1)}\Big)P_{m},
\]
we get
\[
K_{e_{i},e_{j}}=\frac{\lambda}{n^{3}}+\frac{\mu}{n^{3}(n-1)}.
\]
Substituting this into $w_{(i_{t},j_{t})}=\frac{z^{2}\Delta_{e_{i},e_{j}}}{\kappa^{2}}$
shows that every off-diagonal $tt$ edge has the same weight $-u$.
Similarly, every off-diagonal $ss$ edge has weight $-u$. Therefore
\[
\Sigma_{tt+ss}(n)=\binom{n}{2}(-u)+\binom{n}{2}(-u)=-n(n-1)u.
\]
Now use $\alpha=\nu|z|^{2}/n$:
\[
n(n-1)u=n(n-1)\frac{\nu^{2}|z|^{4}}{\kappa^{2}}\left(\frac{\lambda}{n^{3}}+\frac{\mu}{n^{3}(n-1)}\right)=\frac{\alpha^{2}}{\kappa^{2}}\big((n-1)\lambda+\mu\big).
\]
Hence
\[
\Sigma_{tt+ss}(n)=-\frac{\alpha^{2}}{\kappa^{2}}\big((n-1)\lambda+\mu\big).
\]

\textbf{Step 2: bound $\delta_{(S)}$. }Let $G^{\ast}$ be the graph
before deleting vertices in $S$, but after deleting all diagonal
edges $(i_{t},i_{s})$ for $i\in F(\sigma,\tau)$. Deleting vertices/edges
can only decrease incident absolute mass, so it is enough to bound
the weighted degree in $G^{\ast}$. Take a vertex $i_{t}$. For $j\neq i$,
\[
|K_{e_{i},e_{j}}|=\frac{\lambda}{n^{3}}+\frac{\mu}{n^{3}(n-1)}\le\frac{4}{n^{3}}
\]
(using $\lambda,\mu\le2$ for $\alpha\le1/4$). Thus
\[
\sum_{j\neq i}|w_{(i_{t},j_{t})}|\le(n-1)\frac{\nu^{2}|z|^{4}}{\kappa^{2}}\frac{4}{n^{3}}=\frac{4\alpha^{2}}{\kappa^{2}}\frac{n-1}{n}\le16\alpha^{2}
\]
(since $\kappa\ge1/2$). On remaining $ts$ edges we have $e_{i}\neq f_{j}$,
so
\[
(S^{-1})_{e_{i},f_{j}}=(\lambda-\mu)(P_{\parallel})_{e_{i},f_{j}}
\]
because $S^{-1}=\mu I+(\lambda-\mu)P_{\parallel}$ and off-diagonal
entries of $I$ vanish. Hence
\[
\sum_{j:(i_{t},j_{s})\in E^{\ast}}|w_{(i_{t},j_{s})}|\le\frac{\nu|z|^{2}}{\kappa^{2}}|\lambda-\mu|\sum_{j}|(P_{\parallel})_{e_{i},f_{j}}|.
\]
For fixed $i$,
\[
(P_{\parallel})_{(r,c),(r',c')}=\frac{1}{n}\mathbf{1}_{r=r'}+\frac{1}{n}\mathbf{1}_{c=c'}-\frac{1}{n^{2}},
\]
so at most two $j$ are ``special'' (same row and/or same column),
and one gets
\[
\sum_{j}|(P_{\parallel})_{e_{i},f_{j}}|\le\frac{3}{n}.
\]
Also $|\lambda-\mu|\le C\alpha^{2}$ for $\alpha\le1/4$. Therefore
\[
\sum_{j:(i_{t},j_{s})\in E^{\ast}}|w_{(i_{t},j_{s})}|\le\frac{\alpha n}{\kappa^{2}}\cdot(2\alpha^{2})\cdot\frac{3}{n}\le24\alpha^{3}.
\]
It follows that
\[
\sum_{e\ni i_{t}}|w_{e}|\le16\alpha^{2}+24\alpha^{3}\le40\alpha^{2}.
\]
By symmetry the same bound holds for every $i_{s}$. Hence for every
reduced graph, $\delta_{(S)}\le40\alpha^{2}.$

\textbf{Step 3: bound the linear sum $\sum_{e\in E^{(S)}}w_{e}$.}
Let $s:=|S|$. Decompose:
\[
\sum_{e\in E^{(S)}}w_{e}=\Sigma_{tt+ss}^{(S)}+\Sigma_{ts}^{(S)}.
\]
There are $n-s$ $t$-vertices and $n-s$ $s$-vertices left, so
\[
\Sigma_{tt+ss}^{(S)}=-2\binom{n-s}{2}u=-(n-s)(n-s-1)u.
\]
Thus
$\Sigma_{tt+ss}^{(S)}-\Sigma_{tt+ss}(n)=u\big[n(n-1)-(n-s)(n-s-1)\big]=us(2n-s-1)$.
Since $u\le C\alpha^{2}/n$, we get
\[
\Sigma_{tt+ss}^{(S)}\le\Sigma_{tt+ss}(n)+C\alpha^{2}s.
\]
For the $ts$ part, we have
\[
\Sigma_{ts}^{(S)}\le|\Sigma_{ts}^{(S)}|\le\frac{\nu|z|^{2}}{\kappa^{2}}|\lambda-\mu|\sum_{(i,j)\in I_{S}}|(P_{\parallel})_{e_{i},f_{j}}|,
\]
where $I_{S}$ is the set of remaining allowed $ts$ pairs. Removing
vertices/edges only decreases the sum, so
\[
\sum_{(i,j)\in I_{S}}|(P_{\parallel})_{e_{i},f_{j}}|\le\sum_{i=1}^{n}\sum_{j=1}^{n}|(P_{\parallel})_{e_{i},f_{j}}|\le n\cdot\frac{3}{n}=3.
\]
Hence
$|\Sigma_{ts}^{(S)}|\le\frac{\alpha n}{\kappa^{2}}\cdot(2\alpha^{2})\cdot3\le C\alpha^{3}$.
Therefore
\[
\sum_{e\in E^{(S)}}w_{e}=\Sigma_{tt+ss}^{(S)}+\Sigma_{ts}^{(S)}\le\Sigma_{tt+ss}(n)+C\alpha^{2}|S|+C\alpha^{3}n.
\]
This proves the lemma.
\end{proof}
\begin{lem}
There exist absolute $c_{0},C>0$ such that if $\alpha\leq c_{0}$,
then
\[
\log\mathbb{E}_{\sigma,\tau}[Z_{\sigma,\tau}]\leq\Sigma_{tt+ss}(n)+w_{\mathrm{diag}}+C\alpha^{3}n.
\]
\end{lem}

\begin{proof}
Fix $(\sigma,\tau)$ and $S\subset F(\sigma,\tau)$. By Lemma \ref{lem:reduced-graph-local},
$\delta_{(S)}\leq C\alpha^{2}$. For $\alpha\leq c_{0}$ small, $\delta_{(S)}\leq1/(4e)$,
so Lemma \ref{lem:cluster} applies to $Z_{\sigma,\tau}^{(S)}$. Then\footnote{The following inequalities are between real-valued quantities, since the weights in the partition function are real. The cluster expansion with the KP condition ensures $Z_{\sigma,\tau}^{(S)} > 0$.}
\[
\log Z_{\sigma,\tau}^{(S)}\leq\sum_{e\in E^{(S)}}w_{e}+Cn\delta_{(S)}^{2}\leq\sum_{e\in E^{(S)}}w_{e}+Cn\alpha^{4}.
\]
Using Lemma \ref{lem:reduced-graph-local}, 
\[
\log Z_{\sigma,\tau}^{(S)}\leq\Sigma_{tt+ss}(n)+C\alpha^{2}|S|+C\alpha^{3}n.
\]
Exponentiating and summing via Lemma \ref{lem:diagonal-decomposition}
gives
\begin{align*}
Z_{\sigma,\tau} & \leq\exp\left(\Sigma_{tt+ss}(n)+C\alpha^{3}n\right)\sum_{S\subset F(\sigma,\tau)}(w_{\mathrm{diag}}e^{C\alpha^{2}})^{|S|}\\
 & =\exp\left(\Sigma_{tt+ss}(n)+C\alpha^{3}n\right)(1+w_{\mathrm{diag}}e^{C\alpha^{2}})^{m}.
\end{align*}
Using $m=\operatorname{fix}(\tau^{-1}\sigma)$ and averaging in $(\sigma,\tau)$,
we get by Lemma~\ref{lem:fixedpts-gen} that
\[
\mathbb{E}_{\sigma,\tau}(1+w_{\mathrm{diag}}e^{C\alpha^{2}})^{m}\leq\exp(w_{\mathrm{diag}}e^{C\alpha^{2}}).
\]
Hence
\[
\log\mathbb{E}_{\sigma,\tau}[Z_{\sigma,\tau}]\leq\Sigma_{tt+ss}(n)+w_{\mathrm{diag}}e^{C\alpha^{2}}+C\alpha^{3}n.
\]
Because $w_{\mathrm{diag}}\le C\alpha n$ and $e^{C\alpha^{2}}=1+O(\alpha^{2})$,
\[
w_{\mathrm{diag}}e^{C\alpha^{2}}=w_{\mathrm{diag}}+O(\alpha^{2}w_{\mathrm{diag}})=w_{\mathrm{diag}}+O(\alpha^{3}n).
\]
The conclusion then follows. 
\end{proof}
\begin{thm}
There exist absolute constants $c,C>0$ such that if $\alpha:=\nu|z|^{2}/n\leq c$,
then
\[
\log\mathbb{E}|X_{W}^{(2)}(z)|^{2}\leq C\frac{\nu^{3}|z|^{6}}{n^{2}}.
\]
\end{thm}

\begin{proof}
We have $\mathbb{E}|X_{W}^{(2)}(z)|^{2}=M(-za,-\overline{z}a)\kappa^{2n}\mathbb{E}_{\sigma,\tau}[Z_{\sigma,\tau}]$,
so
\[
\log\mathbb{E}|X_{W}^{(2)}(z)|^{2}\leq-\frac{1}{2}\log\det(S)+(n\beta+2n\log\kappa+w_{\mathrm{diag}})+\Sigma_{tt+ss}(n)+C\alpha^{3}n.
\]
We can now show how terms at order $\alpha^2 n$ cancel out, giving the desired $O(\alpha^3 n)$ bound.
Using
\[
\det(S)=(1-\alpha^{2})^{2n-1}\left(1-\frac{\alpha^{2}}{(n-1)^{2}}\right)^{(n-1)^{2}},
\]
one gets
$-\frac{1}{2}\log\det(S)=\alpha^{2}n+O(\alpha^{4}n)$.
By Lemma \ref{lem:reduced-graph-local}, 
\[
\Sigma_{tt+ss}(n)=-\frac{\alpha^{2}}{\kappa^{2}}((n-1)\lambda+\mu).
\]
Since $\lambda=1+O(\alpha^{2})$, $\mu=1+O(\alpha^{2}/(n-1)^{2})$,
and $1/\kappa^{2}=1+2\alpha+O(\alpha^{2})$, we find that
$\Sigma_{tt+ss}(n)=-\alpha^{2}n+O(\alpha^{3}n)$.
Hence by cancellation we obtain
\[
-\frac{1}{2}\log\det(S)+\Sigma_{tt+ss}(n)=O(\alpha^{3}n).
\]
Also, for any coordinate $e$,
\[
\Gamma_{e,e}=\nu\left(\mu+(\lambda-\mu)\left(\frac{2}{n}-\frac{1}{n^{2}}\right)\right)=\nu(1+O(\alpha^{2}/n)),
\]
so 
\[
w_{\mathrm{diag}}=\frac{|z|^{2}\Gamma_{e,e}}{\kappa^{2}}=\frac{\alpha n}{\kappa^{2}}+O(\alpha^{3}).
\]
Recalling the definition of $\beta$ and $\kappa$ from \eqref{eqn:beta-kappa}, we let 
\[
\psi(\alpha):=\frac{\alpha}{1-\alpha}+2\log\frac{1-2\alpha}{1-\alpha}+\frac{\alpha}{\left(\frac{1-2\alpha}{1-\alpha}\right)^{2}}.
\]
and observe that
\[
n\beta+2n\log\kappa+\frac{\alpha n}{\kappa^{2}}=n\psi(\alpha)=O(\alpha^{3}n)
\]
because, by explicit Taylor series expansion, $\psi(\alpha)=\frac{4}{3}\alpha^{3}+O(\alpha^{4})$ near 0.
Thus, 
\[
\log\mathbb{E}|X_{W}^{(2)}(z)|^{2} = O(\alpha^{3}n).
\]
\end{proof}

\subsection{Zero-free region and other applications}
\begin{cor}
Fix $\beta\in(0,1/3)$ and set $R_{n}:=\frac{n^{1/3-\beta}}{\sqrt{\nu}}.$
Then for every fixed $\varepsilon\in(0,1)$ and all sufficiently large
$n$,
\[
\mathbb{P}[G_{W}\text{ has a zero in }\mathbb{D}(0,(1-\varepsilon)R_{n})]=O_{\varepsilon}(n^{-6\beta}).
\]
On this event, for every $\rho \in\mathbb{C}$ with $|\rho|^{-1}<(1-\varepsilon)R_{n}$,
Barvinok interpolation outputs a deterministic $e^{\pm\eta}$-multiplicative
approximation to $\per(\rho J+W)$ in time $n^{O(m)}$, where
\[
m=O\left(\frac{1}{\varepsilon}\log\left(\frac{n}{\eta\varepsilon}\right)\right).
\]
That is, the mean-shift scale improves to $|\rho|\geq(1-\varepsilon)^{-1}\sqrt{\nu}n^{-1/3+\beta}$. 
\end{cor}
\begin{proof}
Plugging in $z = Re^{i \theta}$ and averaging over $\theta$, the previous theorem gives
\[ \mathbb{E}_{\theta} \log\mathbb{E}|X_{W}^{(2)}(R e^{i\theta})|^{2}\leq C\frac{\nu^{3}R^{6}}{n^{2}}. \]
So the expected number of zeros can be bounded by Jensen's formula \eqref{eq:jensen-variant}, using that
\[ \mathbb{E}_{\theta} \mathbb{E}\log|X_{W}^{(2)}(R e^{i\theta})|^{2} \le \mathbb{E}_{\theta} \log \mathbb{E}|X_{W}^{(2)}(R e^{i\theta})|^{2} \]
by Jensen's inequality. Finally, the high-probability statement follows from Markov's inequality. 
\end{proof}

\subsubsection{Second-order expansion}
From the above analysis, we directly have a simple approximate formula for the log-permanent in the regime $|z| \ll n^{1/3}$. This formula is somewhat analogous to the TAP approximation in spin glass theory \cite{thouless1977solution,aizenman1987some}, which was also originally derived from diagrammatic expansions.
\begin{cor}[Second-order log-permanent expansion]
\label{cor:weak-lln-permanent}
Let $W$ be as defined above, and suppose $z \in \mathbb{C}$ depends on $n$ such that $\nu^3|z|^6/n^2 \to 0$ as $n \to \infty$. Then $X_W^{(2)}(z)$ converges in probability to $1$:
\[
X_W^{(2)}(z) \xrightarrow{\mathbb{P}} 1.
\]
Consequently, the log-permanent admits the asymptotic expansion
\[
\log \per(J+zW) = \log(n!) + zD_1(W) + \frac{z^2}{2}D_2(W) + O_\mathbb{P}\left( \frac{\nu^{3/2}|z|^3}{n} \right),
\]
where the implicit constant in the $O_\mathbb{P}$ term depends only on the absolute constants from the second moment bound.
\end{cor}

\begin{proof}
By construction, $\mathbb{E}[X_W^{(2)}(z)] = 1$. The preceding theorem establishes that
\[
\log \mathbb{E}\big|X_W^{(2)}(z)\big|^2 \le C \frac{\nu^3|z|^6}{n^2}
\]
for some absolute constant $C > 0$. Since $\nu^3|z|^6/n^2 = o(1)$ by assumption, we can exponentiate this bound to find the variance:
\[
\operatorname{Var}\big(X_W^{(2)}(z)\big) = \mathbb{E}\big|X_W^{(2)}(z)\big|^2 - \big|\mathbb{E}[X_W^{(2)}(z)]\big|^2 \le \exp\left(C \frac{\nu^3|z|^6}{n^2}\right) - 1 = O\left( \frac{\nu^3|z|^6}{n^2} \right).
\]
Because the variance vanishes as $n \to \infty$, Chebyshev's inequality immediately yields $X_W^{(2)}(z) \xrightarrow{\mathbb{P}} 1$. More precisely, it gives the fluctuation bound
\[
X_W^{(2)}(z) = 1 + O_\mathbb{P}\left( \frac{\nu^{3/2}|z|^3}{n} \right).
\]
Expanding the definition of $X_W^{(2)}(z)$, we have
\[
X_W^{(2)}(z) = \exp\left(-zD_1(W) - \frac{z^2}{2}D_2(W)\right) G_W(z),
\]
where $G_W(z) = \frac{1}{n!}\per(J+zW)$. Taking the principal branch of the logarithm on both sides (which is well-defined with probability approaching $1$ since $X_W^{(2)}(z) \xrightarrow{\mathbb{P}} 1$), and using the Taylor expansion $\log(1+x) = x + O(x^2)$ for small $x$, we obtain
\[
\log X_W^{(2)}(z) = O_\mathbb{P}\left( \frac{\nu^{3/2}|z|^3}{n} \right).
\]
Substituting the definition of $X_W^{(2)}(z)$ into the logarithm yields
\[
\log G_W(z) - zD_1(W) - \frac{z^2}{2}D_2(W) = O_\mathbb{P}\left( \frac{\nu^{3/2}|z|^3}{n} \right).
\]
Rearranging and replacing $\log G_W(z)$ with $\log \per(J+zW) - \log(n!)$ completes the proof.
\end{proof}


\begin{cor}[Central limit theorem]
Provided that $z_n = o(n^{1/3})$, we have the following convergence in distribution as $n \to \infty$:
\[
\frac{\log \per(J + z_n W) - \log(n!)}{z_n} \xrightarrow{d} \mathcal{N}_{\mathbb C}(0,\nu).
\]
\end{cor}

\begin{proof}
By Corollary~\ref{cor:weak-lln-permanent}, we have the asymptotic expansion
\[
\log \per(J+z_n W) - \log(n!) = z_n D_1(W) + \frac{z_n^2}{2}D_2(W) + O_\mathbb{P}\left(\frac{|z_n|^3}{n}\right).
\]
Dividing by $z_n$, we obtain
\[
\frac{\log \per(J+z_n W) - \log(n!)}{z_n} = D_1(W) + \frac{z_n}{2}D_2(W) + O_\mathbb{P}\left(\frac{|z_n|^2}{n}\right).
\]
By definition, $D_1(W) = \frac{1}{n} \sum_{i,j} W_{ij}$. Because the entries $W_{ij}$ are i.i.d. circularly symmetric $\mathcal{N}_{\mathbb C}(0,\nu)$ variables, their scaled sum satisfies $D_1(W) \sim \mathcal{N}_{\mathbb C}(0,\nu)$ for all $n$. 

For the second-order term, $D_2(W) = w^T B w$ is a holomorphic quadratic form. Since $W$ is circularly symmetric, $\mathbb{E}[W_{ij}^2] = 0$, giving $\mathbb{E}[D_2(W)] = 0$. Its variance is given by $\nu^2 \operatorname{Tr}(B^2) = 4\nu^2/n$, meaning $D_2(W) = O_\mathbb{P}(n^{-1/2})$. 

Therefore, for $z_n = o(n^{1/3})$, the relative contribution of the second-order term is
\[
\frac{z_n}{2}D_2(W) = O_\mathbb{P}(|z_n|n^{-1/2}) = o_\mathbb{P}(n^{-1/6}) = o_\mathbb{P}(1).
\]
Similarly, the remainder term is
\[
O_\mathbb{P}\left(\frac{|z_n|^2}{n}\right) = o_\mathbb{P}(n^{-1/3}) = o_\mathbb{P}(1).
\]
Since all higher-order terms vanish in probability, the asymptotic distribution is purely determined by the leading term $D_1(W) \sim \mathcal{N}_{\mathbb C}(0,\nu)$.
\end{proof}

\subsubsection{Approximation algorithm}
\begin{prop}[Coefficient bound, Proposition~\ref{prop:true-log-perm-coeff} of Appendix~\ref{apdx:coeff-bounds}]
Let
\[
P_n(z):=\frac{1}{n!}\per(J+zW),
\qquad
\log P_n(z)=\sum_{k\ge1} b_{k,n} z^k,
\]
where \(W=(W_{ij})_{1\le i,j\le n}\) has i.i.d.\ standard complex Gaussian entries,
$W_{ij}\sim\CN(0,1)$.
Then there exists an absolute constant \(C>0\) such that for every integer \(k\ge1\), 
letting
\[
C_k = \exp\!\bigl(C\,k\,2^k\log k\bigr)
\]
we have that
\[
\mathbb E |b_{k,n}|^2 \le C_k\, n^{1-k}
\qquad\text{for all }n\ge 1.
\]
\end{prop}

\begin{lem}[Quantitative fixed-degree truncation for \(d\ge 2\)]
\label{lem:fixed-degree-truncation-quant}
Let
\[
P_n(z):=\frac{1}{n!}\per(J+zW),
\qquad
\log P_n(z)=\sum_{k\ge1} b_{k,n} z^k,
\]
where \(W\in\C^{n\times n}\) has i.i.d.\ entries
$W_{ij}\sim\mathcal N_{\C}(0,1)$.
Fix \(d\ge2\) and \(\beta\in(0,1/3)\), and set
$L_n:=n^{1/3-\beta}$.
Let \(M\) be any fixed integer such that
\[
M>\max\!\left\{d,\frac{2}{\beta}-1\right\}.
\]
Then there exists a constant \(A_\beta>0\) such that for all sufficiently
large \(n\) and all \(t>0\),
\[
\mathbb P\!\left[
\sup_{|z|\le L_n}
\left|
\log P_n(z)-\sum_{k=1}^{d} b_{k,n}z^k
\right|
>
t+\frac{n(2n^{-\beta/2})^{M+1}}{(M+1)(1-2n^{-\beta/2})}
\right]
\]
\[
\le
A_\beta n^{-3\beta}
+
\frac{M-d}{t^{2}}
\sum_{k=d+1}^{M} C_k\, n^{\,1-k/3-2\beta k},
\]
where 
$C_k=\exp\!\bigl(C\,k\,2^k\log k\bigr)$
are the constants from Proposition~\ref{prop:true-log-perm-coeff}.
\end{lem}
\begin{proof}
Set
\[
T_{d,n}(z):=\sum_{k=1}^{d} b_{k,n} z^k.
\]
Apply the zero-free corollary with \(\beta/2\) in place of \(\beta\) and
\(\varepsilon=1/2\). Then there exists a constant \(A_\beta>0\) such that,
for all sufficiently large \(n\),
$\mathbb P(E_n^c)\le A_\beta n^{-3\beta}$,
where \(E_n\) denotes the event that \(P_n\) has no zeros in
$\mathbb D(0,r_n)$
with 
$r_n:=\frac12\,n^{1/3-\beta/2}$.

On \(E_n\), we may factor
\[
P_n(z)=\prod_{j=1}^{n}\left(1-\frac{z}{\rho_j}\right),
\qquad |\rho_j|\ge r_n,
\]
and hence for \(|z|<r_n\),
\[
\log P_n(z)
=
-\sum_{j=1}^{n}\sum_{k\ge1}\frac{1}{k}\left(\frac{z}{\rho_j}\right)^k.
\]
Therefore, for every \(k\ge1\),
\[
b_{k,n}=-\frac1k\sum_{j=1}^{n}\rho_j^{-k},
\qquad
|b_{k,n}|
\le
\frac{n}{k\,r_n^k}.
\]
Now define
\[ \rho_n:=\frac{L_n}{r_n}=2n^{-\beta/2}. \]
Since \(M>\frac{2}{\beta}-1\), we have
\[
n\rho_n^{M+1}
=
2^{M+1}n^{\,1-\frac{\beta}{2}(M+1)}
\to 0.
\]

It follows that, on \(E_n\), uniformly for \(|z|\le L_n\),
\begin{align*}
\left|\sum_{k>M} b_{k,n} z^k\right|
&\le
\sum_{k>M}\frac{n}{k}\left(\frac{|z|}{r_n}\right)^k \\
&\le
\sum_{k>M}\frac{n}{k}\rho_n^k 
\le
\frac{n\,\rho_n^{M+1}}{(M+1)(1-\rho_n)} =
\frac{n(2n^{-\beta/2})^{M+1}}{(M+1)(1-2n^{-\beta/2})}.
\end{align*}

Thus it remains to control the finite block
\[
\sum_{k=d+1}^{M} b_{k,n} z^k.
\]
By the triangle inequality,
\[
\sup_{|z|\le L_n}\left|\sum_{k=d+1}^{M} b_{k,n} z^k\right|
\le
\sum_{k=d+1}^{M}|b_{k,n}|\,L_n^k.
\]
Hence, by Cauchy--Schwarz and Markov's inequality,
\begin{align*}
&\mathbb P\!\left[
\sup_{|z|\le L_n}\left|\sum_{k=d+1}^{M} b_{k,n} z^k\right|>t
\right] \\
&\qquad\le
\frac{1}{t^2}
\mathbb E\!\left(\sum_{k=d+1}^{M}|b_{k,n}|L_n^k\right)^2 \le
\frac{M-d}{t^2}
\sum_{k=d+1}^{M}\mathbb E|b_{k,n}|^2\,L_n^{2k}.
\end{align*}
By Proposition~\ref{prop:true-log-perm-coeff},
$\mathbb E|b_{k,n}|^2\le C_k n^{1-k}$,
where
$C_k=\exp\!\bigl(C\,k\,2^k\log k\bigr)$.
Since \(L_n=n^{1/3-\beta}\), we obtain
\[
\mathbb E|b_{k,n}|^2\,L_n^{2k}
\le
C_k n^{1-k} n^{2k/3-2\beta k}
=
C_k\,n^{\,1-k/3-2\beta k}.
\]
Therefore
\[
\mathbb P\!\left[
\sup_{|z|\le L_n}\left|\sum_{k=d+1}^{M} b_{k,n} z^k\right|>t
\right]
\le
\frac{M-d}{t^2}
\sum_{k=d+1}^{M} C_k\, n^{\,1-k/3-2\beta k}.
\]

Combining this with the deterministic tail bound on \(E_n\), and then
adding the probability of \(E_n^c\), yields
\[
\mathbb P\!\left[
\sup_{|z|\le L_n}
\left|
\log P_n(z)-\sum_{k=1}^{d} b_{k,n}z^k
\right|
>
t+\frac{n(2n^{-\beta/2})^{M+1}}{(M+1)(1-2n^{-\beta/2})}
\right]
\]
\[
\le
A_\beta n^{-3\beta}
+
\frac{M-d}{t^{2}}
\sum_{k=d+1}^{M} C_k\, n^{\,1-k/3-2\beta k},
\]
which proves the lemma.

\end{proof}
\begin{thm}[Degree--accuracy tradeoff for the truncated log expansion]
\label{thm:degree-accuracy-tradeoff}
Let
\[
P_n(z):=\frac{1}{n!}\per(J+zW),
\qquad
\log P_n(z)=\sum_{k\ge1} b_{k,n} z^k,
\]
where \(W\in\C^{n\times n}\) has i.i.d.\ entries
$W_{ij}\sim \mathcal N_{\C}(0,1)$.
Fix \(\beta\in(0,1/3)\). For each integer \(d\ge2\), define
\[
\Gamma_{d,\beta}:=\frac{d+1}{6}+\beta(d+1)-\frac12.
\]

Then for every \(d\ge2\), every \(\gamma<\Gamma_{d,\beta}\), and every
\(\delta\in(0,1)\), there exists \(N=N(\beta,d,\gamma,\delta)\) such that
for all \(n\ge N\), with probability at least \(1-\delta\), the branch of
\(\log P_n\) determined by \(\log P_n(0)=0\) is well-defined on
$\mathbb D\!\bigl(0,n^{1/3-\beta}\bigr)$,
and satisfies
\[
\sup_{|z|\le n^{1/3-\beta}}
\left|
\log P_n(z)-\sum_{k=1}^{d} b_{k,n} z^k
\right|
\le n^{-\gamma}.
\]

In particular, to obtain an additive \(n^{-\gamma}\)-approximation to
\(\log P_n(z)\) uniformly for \(|z|\le n^{1/3-\beta}\), taking $d = \lceil 6 \gamma + 3 \rceil$, the corresponding
degree-\(d\) truncation gives the desired additive approximation in time
\(n^{O(\gamma)}\).
\end{thm}

\begin{proof}
Apply Lemma~\ref{lem:fixed-degree-truncation-quant}. Choose \(M\) sufficiently large so that the
deterministic tail term in that lemma is \(O(n^{-\gamma})\). Since
\(\gamma<\Gamma_{d,\beta}\), the probabilistic contribution coming from the
block \(k=d+1,\dots,M\) is also \(o(1)\) when evaluated at threshold
\(t=n^{-\gamma}/2\). Hence, after increasing \(N\) if necessary, the lemma
implies that with probability at least \(1-\delta\),
\[
\sup_{|z|\le n^{1/3-\beta}}
\left|
\log P_n(z)-\sum_{k=1}^{d} b_{k,n} z^k
\right|
\le n^{-\gamma}.
\]
This proves the approximation statement.

The runtime claim is immediate from the fact that \(b_{1,n},\dots,b_{d,n}\)
are recovered from the first \(d\) coefficients of \(P_n\), and these can be
computed exactly in time \(n^{O(d)}\) by summing permanents of
\(k\times k\) submatrices for \(k\le d\).
\end{proof}
\subsection{Second-order monomer-dimer analysis}\label{subsec:second-monomer-dimer}
\paragraph{Hardcore model on the line graph of \(K_{n,n}\).}
Let \(G=(V,E)\) be a finite graph. Recall that the \emph{hardcore partition function} on \(G\) with
vertex fugacities \(x=(x_v)_{v\in V}\in\C^V\) is
\[
Z_G(x):=\sum_{I\in \mathcal I(G)} \prod_{v\in I} x_v,
\]
where \(\mathcal I(G)\) denotes the collection of independent sets of \(G\).

In the present section we take
\[
G_n:=L(K_{n,n}),
\]
the line graph of the complete bipartite graph \(K_{n,n}\). In other words, we consider the monomer-dimer model on $K_{n,n}$.
We identify the vertex set
of \(G_n\) with the edge set of \(K_{n,n}\), namely
\[
V(G_n)=[n]\times[n],
\]
and two vertices \((i,j)\) and \((i',j')\) are adjacent in \(G_n\) iff they share a
row or a column:
\[
(i,j)\sim(i',j')
\qquad\Longleftrightarrow\qquad
i=i' \text{ or } j=j'.
\]
Under this identification, independent sets in \(G_n\) are exactly matchings in
\(K_{n,n}\). Therefore the hardcore partition function on \(G_n\) is precisely the
matching generating polynomial of \(K_{n,n}\):
\[
Z_{G_n}(x)
=
\sum_{M\ \mathrm{matching\ in}\ K_{n,n}} \prod_{e\in M} x_e.
\]

We study this partition function under random complex fugacities \(x=zW\), where
\(W=(W_{ij})_{1\le i,j\le n}\) has i.i.d.\ entries \(W_{ij}\sim\CN(0,1)\), together
with its second-order reweighting
\[
F_n^{(2)}(x):=Z_{G_n}(x)\exp\!\bigl(-L_1(x)-L_2(x)\bigr).
\]

\begin{thm}[Second-order bound on \(L(K_{n,n})\), Theorem~\ref{thm:linegraph-unconditional}]
There exist absolute constants \(c,C>0\) such that the following holds.

Let
\[
G_n:=L(K_{n,n}),
\qquad
F_n^{(2)}(x):=Z_{G_n}(x)\exp\!\bigl(-L_1(x)-L_2(x)\bigr),
\]
where
\[
L_1(x):=\sum_{(i,j)\in[n]\times[n]} x_{ij},
\qquad
L_2(x):=-\frac12\sum_{(i,j)}x_{ij}^2-\sum_{(i,j)\sim(i',j')}x_{ij}x_{i'j'}.
\]
Let \(W=(W_{ij})\) have i.i.d.\ entries \(W_{ij}\sim \CN(0,1)\), and define
\[
X_n^{(2)}(z):=F_n^{(2)}(zW).
\]
If
$n|z|^2\le c$,
then
\[
\log \E\bigl[|X_n^{(2)}(z)|^2\bigr]
\le
C\,n^4\,|z|^6.
\]
\end{thm}

We defer the proof of this result to Appendix~\ref{apdx:monomer-dimer-second-order}.

From this result the zero-free region follows by the same application of Jensen's formula. It is also straightforward to derive analogous corollaries in terms of algorithmic guarantees, etc.\ as in the permanent case --- we omit the details.

\section{Universality of first-order zero-free region}\label{sec:permanent-universality}
In this section, we show how to prove a zero-free bound for random matrices with general sub-exponential entries. This requires a new argument, because some of the approximations for the log-permanent obtained earlier are \emph{not} universal --- see the discussion in Section~\ref{sec:nonuniversal}.

\paragraph{Setup.} Let $W$ be an $n\times n$ complex random matrix with i.i.d. entries
such that $\mathbb{E}[W_{ij}]=0$, $\mathbb{E}|W_{ij}|^{2}=\nu$,
and we have the sub-exponential condition
\[
\sup_{\phi\in[0,2\pi)}\Vert\operatorname{Re}(e^{-i\phi}W_{ij})\Vert_{\psi_{1}}\leq K\sqrt{\nu}.
\]
In particular, this covers complex Gaussians and all real sub-exponential
distributions. Let $S_{W}:=\frac{1}{n}\sum_{i,j}W_{ij}$ and define
\[
G_{W}(z):=\frac{1}{n!}\per(J+zW),\quad X_{W}(z):=e^{-zS_{W}}G_{W}(z).
\]
Notice that the zeros of $G_{W}$ and $X_{W}$ coincide. Define $M(u):=\mathbb{E}[e^{uW_{ij}+\overline{u}\overline{W_{ij}}}]$
and consider the tilted expectation 
\[
\mathbb{E}_{u}[F(W)]:=\frac{\mathbb{E}[F(W)e^{uW_{ij}+\overline{u}\overline{W_{ij}}}]}{M(u)}.
\]
Then set
\[
\mu(u):=\mathbb{E}_{u}[W],\quad\sigma^{2}(u):=\mathbb{E}_{u}[|W-\mu(u)|^{2}].
\]
Note that $\mu(u)=\partial_{u}\Psi(u)$ and $\sigma^{2}(u)=\partial_{u}\partial_{\overline{u}}\Psi(u)$
where $\Psi(u)=\log M(u)$. 
\subsection{Exact formulas}

\begin{prop}
It holds that
\[
\mathbb{E}|X_{W}(z)|^{2}=M(-z/n)^{n^{2}}|1+z\mu(-z/n)|^{2n}\sum_{k=0}^{n}\frac{1}{k!}\left(\frac{|z|^{2}\sigma^{2}(-z/n)}{|1+z\mu(-z/n)|^{2}}\right)^{k}.
\]
\end{prop}

\begin{proof}
Write
\[
G_{W}(z)=\frac{1}{n!}\sum_{\sigma\in S_{n}}\prod_{i=1}^{n}(1+zW_{i,\sigma(i)}).
\]
 Let $u:=-z/n$. Then 
\[
\mathbb{E}|X_{W}(z)|^{2}=\frac{1}{(n!)^{2}}\sum_{\sigma,\tau\in S_{n}}\mathbb{E}\left[\prod_{a,b=1}^{n}\exp\left(uW_{a,b}+\overline{u}\overline{W_{a,b}}\right)\prod_{i=1}^{n}(1+zW_{i,\sigma(i)})(1+\overline{z}\overline{W_{i,\tau(i)}})\right].
\]
Fix $\sigma$ and $\tau$. Let $k:=\#\{i\in[n]:\sigma(i)=\tau(i)\}$.
Each pair $(a,b)$ which is used by neither permutation (there are
$n^{2}-2n+k$ such pairs) contributes 
\[
\mathbb{E}[\exp(uW_{a,b}+\overline{u}\overline{W_{a,b}})]=M(u)
\]
to the product. Each pair $(a,b)$ which is used by both permutation
(i.e., $b=\sigma(a)=\tau(a)$) (there are $k$ such pairs) contributes
\begin{align*}
\mathbb{E}[\exp(uW_{a,b}+\overline{u}\overline{W_{a,b}})|1+zW_{a,b}|^{2}] & =M(u)\mathbb{E}_{u}[|1+zW_{a,b}|^{2}]\\
 & =M(u)\mathbb{E}_{u}[|1+z\mu(u)+z(W_{a,b}-\mu(u))|^{2}]\\
 & =M(u)(|1+z\mu(u)|^{2}+|z|^{2}\sigma^{2}(u)).
\end{align*}
Each pair $(a,b)$ which is used by $\sigma$ but not $\tau$ (there
are $n-k$ such pairs) contributes
\begin{align*}
\mathbb{E}[\exp(uW_{a,b}+\overline{u}\overline{W_{a,b}})(1+zW)] & =M(u)\mathbb{E}_{u}[1+zW]\\
 & =M(u)(1+z\mu(u)).
\end{align*}
Similarly, each pair $(a,b)$ which is used by $\tau$ but not $\sigma$
(there are $n-k$ such pairs) contributes $M(u)(1+\overline{z}\overline{\mu(u)}).$
Thus, the expectation term for a fixed pair of permutations $(\sigma,\tau)$
is 
\[
M(u)^{n^{2}}|1+z\mu(u)|^{2(n-k)}\left(|1+z\mu(u)|^{2}+|z|^{2}\sigma^{2}(u)\right)^{k}=M(u)^{n^{2}}|1+z\mu(u)|^{2n}t^{k},
\]
where $t:=1+\frac{|z|^{2}\sigma^{2}}{|1+z\mu(u)|^{2}}.$ Note that
$k=\fix(\pi)$ where $\pi=\sigma^{-1}\tau\sim\mathsf{Unif}(S_{n})$,
so using the probability generating function for $\fix(\pi)$ from Lemma~\ref{lem:fixedpts-gen} gives
\[
\mathbb{E}|X_{W}(z)|^{2}=M(u)^{n^{2}}|1+z\mu(u)|^{2n}\mathbb{E}_{\pi}[t^{\fix(\pi)}]
\]
which yields the desired formula. 
\end{proof}
\begin{lem}
There exists $r_{0}(K)>0$ such that $F(a,b):=\log\mathbb{E}e^{aW_{ij}+b\overline{W_{ij}}}$
    is holomorphic on the bidisc $D:=\{(a,b)\in\mathbb{C}^{2}:|a|\leq r_{0}(K)/\sqrt{\nu},|b|\leq r_{0}(K)/\sqrt{\nu}\}$
and has an absolutely convergent power series
\[
F(a,b)=\sum_{p+q\geq2}f_{pq}a^{p}b^{q}
\]
where $|f_{pq}|\leq\frac{1}{2}r_{0}(K)^{-(p+q)}\nu^{(p+q)/2}$. 
\end{lem}

\begin{proof}
From the directional sub-exponential condition with $\phi=0$ and
$\phi=\pi/2$, we get
\[
\Vert\operatorname{Re}(W_{ij})\Vert_{\psi_{1}}\leq K\sqrt{\nu},\quad\Vert\operatorname{Im}(W_{ij})\Vert_{\psi_{1}}\leq K\sqrt{\nu}.
\]
Since $|W|\leq|\operatorname{Re}(W_{ij})|+|\operatorname{Im}(W_{ij})|$,
we get $\Vert|W_{ij}|\Vert_{\psi_{1}}\leq C_{1}K\sqrt{\nu}$ for some
absolute constant $C_{1}$. By a standard fact about moments of sub-exponential
variables, we get
\[
\mathbb{E}|W_{ij}|^{m}\leq(C_{3}K\sqrt{\nu}m)^{m}\quad\text{for all }m\ge1.
\]
Fix $r>0$ and assume $a,b\in\mathbb{C}$ satisfy $|a|,|b|\leq r/\sqrt{\nu}$.
Then
\[
|aW_{ij}+b\overline{W_{ij}}|\leq(|a|+|b|)|W_{ij}|\leq\frac{2r}{\sqrt{\nu}}|W_{ij}|.
\]
Using $m!\geq(m/e)^{m}$, we get
\[
\frac{\mathbb{E}|aW_{ij}+b\overline{W_{ij}}|^{m}}{m!}\leq\frac{(2r/\sqrt{\nu})^{m}(C_{3}K\sqrt{\nu}m)^{m}}{m!}\leq(2eC_{3}Kr)^{m}.
\]
Choosing $r_{0}(K):=(8eC_{3}K)^{-1}$, we have for $|a|,|b|\leq r_{0}(K)/\sqrt{\nu}$,
\[
\sum_{m\geq0}\frac{\mathbb{E}|aW_{ij}+b\overline{W_{ij}}|^{m}}{m!}\leq\sum_{m\geq0}2^{-m}<\infty,
\]
so $M(a,b):=\mathbb{E}e^{aW_{ij}+b\overline{W_{ij}}}=\sum_{m\geq0}\frac{\mathbb{E}(aW_{ij}+b\overline{W_{ij}})^{m}}{m!}$
converges absolutely and locally uniformly, so $M$ is holomorphic
on $D$. Because $\mathbb{E}W_{ij}=0$, 
\[
M(a,b)-1=\sum_{m\geq2}\frac{\mathbb{E}(aW_{ij}+b\overline{W_{ij}})^{m}}{m!},
\]
so 
$|M(a,b)-1|\leq\sum_{m\geq2}2^{-m}=\frac{1}{2}$.
Shrinking $r_{0}(K)$ by an absolute factor if needed, we may ensure
$|M(a,b)-1|\leq1/4$ on $D$, so that $M$ stays in a simply connected
neighborhood of 1. This ensures that $F(a,b)=\log M(a,b)$ is holomorphic
on $D$ and hence $F$ admits a power series
\[
F(a,b)=\sum_{p,q\geq0}f_{pq}a^{p}b^{q}.
\]
Also, $f_{00}=F(0,0)=0$, and
\[
f_{10}=\partial_{a}F(0,0)=\frac{\mathbb{E}W_{ij}}{M(0,0)}=0,\quad f_{01}=\partial_{b}F(0,0)=\frac{\mathbb{E}\overline{W_{ij}}}{M(0,0)}=0,
\]
so the sum goes over $p+q\geq2$. Applying Cauchy coefficient estimates
on $D$ gives
\[
|f_{pq}|\leq\sup_{a,b\in D}|F(a,b)|\left(\frac{\sqrt{\nu}}{r_{0}(K)}\right)^{p+q}\leq\frac{1}{2}\left(\frac{\sqrt{\nu}}{r_{0}(K)}\right)^{p+q},
\]
which finishes the proof.
\end{proof}
\subsection{Averaged second-moment estimate}
\begin{thm}
There exist constants $c_{\ast}(K),C_{\ast}(K)>0$ such that for every
$R>0$ with $\nu R^{2}/n\leq c_{\ast}(K)$, 
\[
\mathbb{E}_{\theta}\log\mathbb{E}|X_{W}(Re^{i\theta})|^{2}\leq\frac{C_{\ast}(K)\nu^{2}R^{4}}{n}.
\]
\end{thm}

\begin{proof}
Set 
\[
z_{\theta}:=Re^{i\theta},\quad u_{\theta}:=-\frac{R}{n}e^{i\theta},\quad\alpha:=\frac{\nu R^{2}}{n}.
\]
Using the formula for $\mathbb{E}|X_{W}(z)|^{2}$ and $\sum_{k=0}^{n}\frac{t^{k}}{k!}\leq\exp(t)$,
we have
\[
\log\mathbb{E}|X_{W}(z_{\theta})|^{2}\leq n^{2}\Psi(u_{\theta})+2n\log|1+z_{\theta}\mu(u_{\theta})|+\frac{|z_{\theta}|^{2}\sigma^{2}(u_{\theta})}{|1+z\mu(u_{\theta})|^{2}}.
\]
By the previous lemma, there exist $r_{0}(K)>0$ such that
\[
\Psi(u)=\sum_{p+q\ge2}c_{pq}u^{p}\overline{u}^{q}
\]
and the series converges absolutely for $|u|\leq r_{0}(K)/\sqrt{\nu}$
and the coefficients satisfy
\begin{equation}
|c_{pq}|\leq\frac{1}{2}r_{0}(K)^{-(p+q)}\nu^{(p+q)/2}.\label{eq:coef}
\end{equation}
Also, note that $c_{11}=\nu$ and $c_{20}=\kappa/2$ where $\kappa=|\mathbb{E}W^{2}|\leq\nu$. 

\textbf{Step 1: the $n^{2}\Psi$ term. }Averaging over $\theta$,
only the $p=q$ terms survive, so 
\[
\mathbb{E}_{\theta}[n^{2}\Psi(u_{\theta})]=\nu R^{2}+n^{2}\sum_{m\geq2}c_{mm}\left(\frac{R}{n}\right)^{2m}.
\]
Choose $c_{\ast}(K)\leq r_{0}(K)^{2}/4$. By (\ref{eq:coef}), for
all $\alpha\leq c_{\ast}(K)$, we have
\begin{align*}
\mathbb{E}_{\theta}[n^{2}\Psi(u_{\theta})] & \leq\nu R^{2}+\frac{1}{2}n^{2}\sum_{m\geq2}\left(\frac{\sqrt{\nu}}{r_{0}(K)}\right)^{2m}\left(\frac{R}{n}\right)^{2m}\\
 & =\nu R^{2}+\frac{1}{2}n^{2}\sum_{m\geq2}\left(\frac{\alpha}{nr_{0}(K)^{2}}\right)^{m}
  =\nu R^{2}+O_{K}\left(\frac{\nu^{2}R^{4}}{n^{2}}\right)
\end{align*}
since $\frac{\alpha}{nr_{0}(K)^{2}}\leq\frac{c_{\ast}(K)}{nr_{0}(K)^{2}}\leq\frac{1}{4n}<1$. 

\textbf{Step 2: the $2n\log|1+z\mu(u_{\theta})|$ term. }Since $\mu(u)=\partial_{u}\Psi(u)$,
we have
\[
\mu(u)=\sum_{p+q\ge1}m_{pq}u^{p}\overline{u}^{q},\quad m_{pq}:=(p+1)c_{p+1,q}.
\]
From (\ref{eq:coef}) and $p+1\leq2^{p+q+1},$ we have
\[
|m_{pq}|\leq B_{1}(K)A_{1}(K)^{p+q}\nu^{(p+q+1)/2},\quad p+q\geq1.
\]
Set $\eta:=\sqrt{\nu}R/n$. We have
\[
w_{\theta}:=z_{\theta}\mu(u_{\theta})=\sum_{p+q\geq1}a_{pq}e^{i(p-q+1)\theta},\quad a_{pq}:=(-1)^{p+q}m_{pq}\frac{R^{p+q+1}}{n^{p+q}}.
\]
By the bound on $|m_{pq}|$, we have
\[
|a_{pq}|\leq C_{K}\alpha(C_{K}\eta)^{p+q-1}.
\]
The $(p,q)=(0,1)$ term equals $a_{01}=-m_{01}\frac{R^{2}}{n}=-c_{11}\frac{R^{2}}{n}=-\alpha$,
so write
\[
w_{\theta}=-\alpha+\widetilde{w}_{\theta},\quad\widetilde{w}_{\theta}:=\sum_{(p,q)\neq(0,1)}a_{pq}e^{i(p-q+1)\theta}.
\]
Choose $c_{\ast}(K)$ small enough so that $C_{K}\sqrt{c_{\ast}(K)}\leq1/4$.
Then $C_{K}\eta=C_{K}\sqrt{\alpha/n}\leq C_{K}\sqrt{\alpha}\leq1/4$,
so 
\[
\sup_{\theta}|\widetilde{w}_{\theta}|\leq\sum_{(p,q)\neq(0,1)}|a_{pq}|\leq C_{K}\alpha.
\]
Also,
\[
\mathbb{E}_{\theta}\widetilde{w}_{\theta}=\sum_{(p,q)\neq(0,1),p-q+1=0}a_{pq}=\sum_{p\geq1}a_{p,p+1},
\]
so 
\begin{align*}
|\mathbb{E}_{\theta}\widetilde{w}_{\theta}| \leq\sum_{p\geq1}|a_{p,p+1}|
 & \leq C_{K}\alpha\sum_{p\geq1}(C_{K}\eta)^{2p}\\
 & =O_{K}(\alpha\eta^{2})
 =O_{K}(\frac{\alpha^{2}}{n})
 =O_{K}(\alpha^{2}).
\end{align*}
Therefore, 
\[
\mathbb{E}_{\theta}\widetilde{w}_{\theta}=O_{K}(\alpha^{2}),\quad\sup_{\theta}|\widetilde{w}_{\theta}|\leq C_{K}\alpha.
\]
Since $\alpha\leq c_{\ast}(K)$, after shrinking $c_{\ast}(K)$ again
if needed, we can ensure that $\sup_{\theta}|\widetilde{w}_{\theta}|\leq1/2$.
For $|w|\leq1/2$, $\log|1+w|=\operatorname{Re}(w)+O(|w|^{2})$. Hence
\[
\mathbb{E}_{\theta}\log|1+w_{\theta}|=\operatorname{Re}(\mathbb{E}_{\theta}w_{\theta})+O(\mathbb{E}_{\theta}|w_{\theta}|^{2})=-\alpha+O_{K}(\alpha^{2}).
\]
It follows that 
\[
\mathbb{E}_{\theta}[2n\log|1+z_{\theta}\mu(u_{\theta})|]=-2\nu R^{2}+O_{K}\left(\frac{\nu^{2}R^{4}}{n}\right).
\]

\textbf{Step 3: the rational term. }Define
\[
T_{3}:=\mathbb{E}_{\theta}\left[\frac{|z_{\theta}|^{2}\sigma^{2}(u_{\theta})}{|1+z\mu(u_{\theta})|^{2}}\right]=R^{2}\mathbb{E}_{\theta}\left[\frac{\sigma^{2}(u_{\theta})}{|1+w_{\theta}|^{2}}\right].
\]
Set $\sigma^{2}(u)=\nu+\widetilde{\sigma}(u)$, and use 
\[
\sigma^{2}(u)=\partial_{u}\partial_{\overline{u}}\Psi(u)=\nu+\sum_{p+q\geq1}s_{pq}u^{p}\overline{u}^{q},\quad s_{pq}:=(p+1)(q+1)c_{p+1,q+1}.
\]
From (\ref{eq:coef}) and $(p+1)(q+1)\leq4^{p+q+1},$ we have
\[
|s_{pq}|\leq B_{2}(K)A_{2}(K)^{p+q}\nu^{(p+q+2)/2},\quad p+q\geq1.
\]
Also, 
\[
\widetilde{\sigma}(u_{\theta})=\sum_{p+q\geq1}b_{pq}e^{i(p-q)\theta},\quad b_{pq}:=(-1)^{p+q}s_{pq}\left(\frac{R}{n}\right)^{p+q},
\]
and $|b_{pq}|\leq C_{K}\nu(C_{K}\eta)^{p+q}.$ By shrinking $c_{\ast}(K)$
if needed, assume that $C_{K}\eta\leq1/4$. Then
\begin{align*}
\sup_{\theta}|\widetilde{\sigma}(u_{\theta})| & \leq\sum_{p+q\geq1}|b_{pq}|\\
 & \leq C_{K}\nu\sum_{p+q\geq1}(C_{K}\eta)^{p+q}\\
 & \leq C_{K}\nu\sum_{m\geq1}\sum_{p+q=m}(C_{K}\eta)^{m}\\
 & =C_{K}\nu\sum_{m\geq1}(m+1)(C_{K}\eta)^{m}
  =C_{K}\nu\left(\frac{1}{(1-C_{K}\eta)^{2}}-1\right) \leq C'_{K}\nu\eta,
\end{align*}
and 
\[
\mathbb{E}_{\theta}\widetilde{\sigma}(u_{\theta})=\sum_{p\geq1}b_{pp}=O_{K}(\nu\eta^{2})=O_{K}\left(\nu\frac{\alpha}{n}\right).
\]
For $|w|\leq1/2$, Taylor expansion of $D(w):=|1+w|^{-2}$ gives
\[
D(w)=1-(w+\overline{w})+R_{D}(w),\quad|R_{D}(w)|\leq C_{D}|w|^{2}.
\]
So
\[
\frac{T_{3}}{R^{2}}=A_{0}+A_{1}+A_{2},
\]
where $A_{0}:=\mathbb{E}_{\theta}[\nu+\widetilde{\sigma}(u_{\theta})]$,
$A_{1}:=-\mathbb{E}_{\theta}[(\nu+\widetilde{\sigma}(u_{\theta}))(w_{\theta}+\overline{w_{\theta}})]$,
$A_{2}:=\mathbb{E}_{\theta}[(\nu+\widetilde{\sigma}(u_{\theta}))R_{D}(w_{\theta})]$.
For $A_{0}$, we have
\[
A_{0}=\nu+O_{K}\left(\nu\frac{\alpha}{n}\right).
\]
For $A_{1}$, we write $w_{\theta}=-\alpha+\widetilde{w}_{\theta}$.
We showed before that $\mathbb{E}_{\theta}\widetilde{w}_{\theta}=O_{K}(\alpha^{2})$
and $\sup_{\theta}|\widetilde{w}_{\theta}|\leq C_{K}\alpha$, so
\[
-\nu\mathbb{E}_{\theta}[w_{\theta}+\overline{w_{\theta}}]=2\nu\alpha+O_{K}(\nu\alpha^{2}).
\]
Also, using $\mathbb{E}_{\theta}|w_{\theta}+\overline{w_{\theta}}|\leq2\sup_{\theta}|w_{\theta}|\leq2\alpha+2\sup_{\theta}|\widetilde{w}_{\theta}|=O_{K}(\alpha)$,
we get
\[
|\mathbb{E}_{\theta}[\widetilde{\sigma}(u_{\theta})(w_{\theta}+\overline{w_{\theta}})]|\leq\sup_{\theta}|\widetilde{\sigma}(u_{\theta})|\,\mathbb{E}_{\theta}|w_{\theta}+\overline{w_{\theta}}|=O_{K}\left(\nu\alpha\sqrt{\frac{\alpha}{n}}\right).
\]
Thus
\[
A_{1}=2\nu\alpha+O_{K}(\nu\alpha^{2})+O_{K}\left(\nu\alpha\sqrt{\frac{\alpha}{n}}\right).
\]
For $A_{2}$, we use $|R_{D}(w_{\theta})|\leq C_{D}|w_{\theta}|^{2}$;
then using $\mathbb{E}_{\theta}|w_{\theta}|^{2}\leq\sup_{\theta}|w_{\theta}|^{2}=O_{K}(\alpha^{2})$,
we have
\begin{align*}
|A_{2}| & \leq\left(\nu+\sup_{\theta}|\widetilde{\sigma}(u_{\theta})|\right)C_{D}\mathbb{E}_{\theta}|w_{\theta}|^{2}\\
 & \leq\nu(1+C_{K}\eta)O_{K}(\alpha^{2})\\
 & =O_{K}(\nu\alpha^{2}).
\end{align*}
Combining, we get
\[
\frac{T_{3}}{R^{2}}=\nu+2\nu\alpha+O_{K}\left(\nu\alpha^{2}+\nu\alpha\sqrt{\frac{\alpha}{n}}+\nu\frac{\alpha}{n}\right).
\]
By choosing $c_{\ast}(K)$ to be smaller than some absolute constant
smaller than 1, $\alpha\leq c_{\ast}(K)$ implies the bracketed terms
are $O_{K}(\nu\alpha)$. Thus, 
\[
\frac{T_{3}}{R^{2}}=\nu+O_{K}\left(\nu\alpha\right).
\]
It follows that $T_{3}=\nu R^{2}+O_{K}(\nu\alpha R^{2})=\nu R^{2}+O_{K}(\nu^{2}R^{4}/n)$. 

\textbf{Step 4: combine. }It follows that for all $\alpha\leq c_{\ast}(K)$,
we have
\begin{align*}
\mathbb{E}_{\theta}\log\mathbb{E}|X_{W}(Re^{i\theta})|^{2} & =(\nu R^{2})+\left(-2\nu R^{2}+O_{K}\left(\frac{\nu^{2}R^{4}}{n}\right)\right)+\left(\nu R^{2}+O_{K}\left(\frac{\nu^{2}R^{4}}{n}\right)\right)\\
 & =O_{K}\left(\frac{\nu^{2}R^{4}}{n}\right).
\end{align*}
This is the desired conclusion. 
\end{proof}
This gives a zero-free disk with high probability, which we discuss next.
\subsection{Zero-free region and algorithm}
\begin{thm}
Let $0<r<R$ and assume $\nu R^{2}/n\leq c_{\ast}(K)$. Then 
\[
\mathbb{P}[N_{G_{W}}(r)\geq1]\leq\frac{C_{\ast}(K)\nu^{2}R^{4}}{2n\log(R/r)}.
\]
\end{thm}

\begin{proof}
We have $N_{G_{W}}(r)=N_{X_{W}}(r)$, so Markov's inequality gives
\[
\mathbb{P}[N_{G_{W}}(r)\geq1]=\mathbb{P}[N_{X_{W}}(r)\geq1]\leq\mathbb{E}[N_{X_{W}}(r)].
\]
Because $X_W(0) = \frac{n!}{n!} = 1$, using Lemma~\ref{lem:core-jensen-app} (the consequence of Jensen's formula) gives
\[
\mathbb{E}[N_{X_{W}}(r)]\leq\frac{1}{2\log(R/r)}\mathbb{E}_{\theta}\log\mathbb{E}|X_{W}(Re^{i\theta})|^{2}.
\]
By the previous theorem, we have
\[
\mathbb{E}[N_{X_{W}}(r)]\leq\frac{C_{\ast}(K)\nu^{2}R^{4}}{2n\log(R/r)}.
\]
\end{proof}
\begin{thm}
Fix $\beta\in(0,1/4)$ and set $R_{n}:=\frac{n^{1/4-\beta}}{\sqrt{\nu}}.$
Then for every fixed $\varepsilon\in(0,1)$ and all sufficiently large
$n$,
\[
\mathbb{P}[G_{W}\text{ has a zero in }\mathbb{D}(0,(1-\varepsilon)R_{n})]=O_{\varepsilon}(n^{-4\beta}).
\]
\end{thm}

\begin{proof}
Set $R:=R_{n}$ and $r:=(1-\varepsilon/2)R$. By the previous theorem,
\[
\mathbb{P}[G_{W}\text{ has a zero in }\mathbb{D}(0,r)]\leq\frac{C_{\ast}(K)\nu^{2}R^{4}}{2n\log(R/r)}=O_{\varepsilon}(n^{-4\beta}),
\]
since $\nu^{2}R^{4}=n^{1-4\beta}$ and $\log(R/r)=\log((1-\varepsilon/2)^{-1})=\Theta_{\varepsilon}(1)$.
\end{proof}
\begin{rem}
The result applies with a sequence $\beta_{n}\in(0,1/4)$ in place
of fixed $\beta$. Hence if $\beta_{n}\downarrow0$ and $\beta_{n}\log n\to\infty$,
then the algorithm works for $|\lambda|\apprge\sqrt{\nu}n^{-1/4+\beta_{n}}=\sqrt{\nu}n^{-1/4+o(1)}$.
The success probability is still $1-o(1)$ but no longer $1-1/\operatorname{poly}(n)$.
\end{rem}

\begin{cor}[Approximation algorithm]
\label{cor:universality-algorithmic}
Fix \(\beta\in(0,1/4)\) and \(\varepsilon\in(0,1)\).
Then for every fixed \(\gamma>0\), with probability
$1-O_{\varepsilon}(n^{-4\beta})$,
the following holds simultaneously for every \(z\in\C\) such that
\[
|z|\le (1-\varepsilon)\frac{n^{1/4-2\beta}}{\sqrt{\nu}}.
\]
There is a deterministic algorithm which outputs an
\(n^{-\gamma}\)-additive approximation to the branch of
\[
\log G_W(z)=\log\!\left(\frac{1}{n!}\per(J+zW)\right)
\]
determined by $\log G_W(0)=0$,
in time \(n^{O(m)}\), where
\[
m=\left\lceil \frac{\gamma+1}{\beta}\right\rceil.
\]
This algorithm is based on the degree-\(m\) truncated Taylor series of
\(\log G_W\).
\end{cor}

\begin{proof}
Set
\[
R_n:=\frac{n^{1/4-\beta}}{\sqrt{\nu}},
\qquad
r_n:=(1-\varepsilon/2)R_n,
\qquad
L_n:=(1-\varepsilon) n^{-\beta} R_n
=\frac{(1-\varepsilon)n^{1/4-2\beta}}{\sqrt{\nu}}.
\]

By the previous theorem, for all sufficiently large \(n\),
\[
\mathbb P\!\left[G_W \text{ has a zero in } \mathbb D(0,r_n)\right]
=O_{\varepsilon}(n^{-4\beta}).
\]
Condition on the complementary event. Then \(G_W\) is zero-free on
\(\mathbb D(0,r_n)\), so there is a holomorphic branch
\[
f(z):=\log G_W(z)=\sum_{k\ge1} a_k z^k
\]
on \(\mathbb D(0,r_n)\), normalized by \(f(0)=0\).

Now fix \(z\) with \(|z|\le L_n\). Then
\[
\frac{|z|}{r_n}
\le
\frac{(1-\varepsilon)n^{-\beta}R_n}{(1-\varepsilon/2)R_n}
=
\rho_{\varepsilon}\, n^{-\beta},
\qquad
\rho_{\varepsilon}:=\frac{1-\varepsilon}{1-\varepsilon/2}<1.
\]
Let
\[
T_m(z):=\sum_{k=1}^{m} a_k z^k
\]
be the degree-\(m\) truncation of the Taylor series of \(\log G_W\).
By Barvinok interpolation \cite{barvinok2016computing},
\[
|f(z)-T_m(z)|
\le
\frac{n}{m+1}\,
\frac{(\rho_{\varepsilon}n^{-\beta})^{m+1}}
{1-\rho_{\varepsilon}n^{-\beta}}.
\]
Choose
\[
m=\left\lceil \frac{\gamma+1}{\beta}\right\rceil.
\]
Then \(\beta(m+1)\ge \gamma+1+\beta\), so for all sufficiently large \(n\),
\[
\frac{n}{m+1}\,
\frac{(\rho_{\varepsilon}n^{-\beta})^{m+1}}
{1-\rho_{\varepsilon}n^{-\beta}}
\le
C_{\varepsilon,\beta,\gamma}\, n^{1-\beta(m+1)}
\le
n^{-\gamma}.
\]
Hence
\[
|\log G_W(z)-T_m(z)|\le n^{-\gamma}.
\]

Next, write
\[
G_W(z)=\sum_{k=0}^{n} g_k z^k,
\qquad g_0=1.
\]
As before,
\[
g_k=\frac{(n-k)!}{n!}
\sum_{\substack{I,J\subseteq[n]\\ |I|=|J|=k}}
\per(W_{I,J}),
\]
so \(g_0,\dots,g_m\) can be computed exactly in time \(n^{O(m)}\),
because each \(k\times k\) permanent is computable in \(O(k2^k)\) time
via Ryser's formula, and there are \(\binom{n}{k}^2\) choices of
\((I,J)\). The coefficients \(a_1,\dots,a_m\) are then recovered from
\[
k g_k=\sum_{j=1}^{k} j a_j g_{k-j},
\]
that is,
\[
a_k=g_k-\frac1k\sum_{j=1}^{k-1} j a_j g_{k-j},
\]
in additional polynomial time in \(m\).

Therefore \(T_m(z)\) is an \(n^{-\gamma}\)-additive approximation to
\(\log G_W(z)\) on the specified branch. The total running time is
\(n^{O(m)}\), which is \(n^{O(\gamma/\beta)}\).
\end{proof}

\subsection{Discussion: some non-universal formulas}\label{sec:nonuniversal}

When the entries of \(W\) are no longer circular complex Gaussian, there
are distribution-dependent corrections which appear in
\(\log \per(J+zW)\). What happens is that the first-order reweighting
\[
X_W(z):=e^{-zS_W}G_W(z)
\]
is no longer exactly centered at \(1\): already its first moment acquires
a nontrivial quadratic correction. The easiest way to see this is in the
case of real Gaussian entries.

\paragraph{Example: real Gaussian.}
Suppose \(W=(W_{ij})_{1\le i,j\le n}\) has i.i.d.\ entries
\(W_{ij}\sim N(0,1)\), and let \(\xi\sim N(0,1)\) denote a single scalar
entry. Set
\[
u:=-\frac{z}{n}.
\]
Since
\[
G_W(z)=\frac{1}{n!}\sum_{\sigma\in S_n}\prod_{i=1}^n (1+zW_{i,\sigma(i)}),
\]
we have
\[
\E[X_W(z)]
=
\frac{1}{n!}\sum_{\sigma\in S_n}
\E\!\left[
\prod_{a,b=1}^n e^{uW_{ab}}
\prod_{i=1}^n (1+zW_{i,\sigma(i)})
\right].
\]
Fix a permutation \(\sigma\in S_n\). By independence, the expectation
factors over the matrix entries.

Each entry \((a,b)\) not used by \(\sigma\) contributes
\[
\E[e^{u\xi}]=e^{u^2/2}.
\]
Each entry \((a,\sigma(a))\) used by \(\sigma\) contributes
\[
\E[e^{u\xi}(1+z\xi)]
=
\E[e^{u\xi}] + z\,\E[\xi e^{u\xi}]
=
e^{u^2/2}(1+uz),
\]
where we used Stein's lemma in the form
$\E[\xi e^{u\xi}]=u\,\E[e^{u\xi}]$.
There are \(n^2-n\) unused entries and \(n\) used entries, so every
permutation contributes
\[
\bigl(e^{u^2/2}\bigr)^{n^2-n}\bigl(e^{u^2/2}(1+uz)\bigr)^n
=
e^{n^2u^2/2}(1+uz)^n.
\]
Therefore
\[
\E[X_W(z)]
=
e^{n^2u^2/2}(1+uz)^n
=
\exp\!\left(\frac{z^2}{2}\right)\left(1-\frac{z^2}{n}\right)^n \to e^{-z^2/2}
\]
as $n\to\infty$.
So even for Gaussian entries, the first-order reweighting is not exactly
centered unless one is in the circular complex Gaussian setting.

The general mechanism, which we can guess from the cluster expansion (consider the mean of the second-order term), is that the leading discrepancy is governed by the
\emph{holomorphic second moment}
\[
\kappa:=\E[\xi^2].
\]
For a real standard Gaussian, \(\kappa=1\), which produces the quadratic
correction above. For a circular complex Gaussian, \(\kappa=0\), and in
fact \(\E[X_W(z)]=1\) exactly.
The next lemma makes this precise. Note that this lemma is included for expository purposes --- it is not used in the proof of any other results.

\begin{lem}[Quadratic correction to the first moment]
\label{lem:first-moment-kappa}
Let \(\xi\) be a complex-valued random variable with \(\E[\xi]=0\), and
assume that the holomorphic moment generating function
\[
M(u):=\E[e^{u\xi}]
\]
is analytic and nonvanishing in a neighborhood of \(u=0\). Define
\[
K(u):=\log M(u),
\qquad
\kappa:=\E[\xi^{2}],
\]
where the branch of \(\log\) is chosen so that \(K(0)=0\).

Let \(W=(W_{ij})_{1\le i,j\le n}\) be an \(n\times n\) random matrix with
i.i.d.\ entries distributed as \(\xi\), and set
\[
G_W(z):=\frac{1}{n!}\per(J+zW),
\qquad
S_W:=\frac1n\sum_{i,j=1}^n W_{ij},
\qquad
X_W(z):=e^{-zS_W}G_W(z).
\]
Then, for all \(z\) sufficiently close to \(0\),
\[
\E[X_W(z)]
=
\exp\!\Bigl(
n^{2}K(-z/n)+n\log\!\bigl(1+zK'(-z/n)\bigr)
\Bigr).
\]
In particular,
\[
\log \E[X_W(z)]
=
-\frac{\kappa}{2}z^{2}+O(z^{3}).
\]
\end{lem}

\begin{proof}
Write
\[
G_W(z)=\frac1{n!}\sum_{\sigma\in S_n}\prod_{i=1}^n (1+zW_{i,\sigma(i)}).
\]
Let \(u:=-z/n\). Then
\[
\E[X_W(z)]
=
\frac1{n!}\sum_{\sigma\in S_n}
\E\!\left[
\prod_{a,b=1}^n e^{uW_{ab}}
\prod_{i=1}^n (1+zW_{i,\sigma(i)})
\right].
\]
Fix \(\sigma\in S_n\). Each entry \((a,b)\) not used by \(\sigma\) contributes
$\E[e^{u\xi}]=M(u)$.
Each used entry contributes
\[
\E[e^{u\xi}(1+z\xi)]
=
M(u)+z\,\E[\xi e^{u\xi}]
=
M(u)\bigl(1+zK'(u)\bigr),
\]
since
\[
K'(u)=\frac{M'(u)}{M(u)}=\frac{\E[\xi e^{u\xi}]}{\E[e^{u\xi}]}.
\]
Therefore every permutation contributes the same quantity,
\[
M(u)^{n^{2}}\bigl(1+zK'(u)\bigr)^n,
\]
and so
\[
\E[X_W(z)]
=
M(u)^{n^{2}}\bigl(1+zK'(u)\bigr)^n.
\]
Taking logarithms gives
\[
\log \E[X_W(z)]
=
n^{2}K(-z/n)+n\log\!\bigl(1+zK'(-z/n)\bigr).
\]

Now \(K(0)=0\), \(K'(0)=\E[\xi]=0\), and
$K''(0)=\E[\xi^2]=\kappa$.
Expanding near \(0\),
\[
K(u)=\frac{\kappa}{2}u^2+O(u^3),
\qquad
K'(u)=\kappa u+O(u^2).
\]
Substituting \(u=-z/n\), we get
$n^2K(-z/n)=\frac{\kappa}{2}z^2+O(z^3/n)$
and
$zK'(-z/n)=-\frac{\kappa}{n}z^2+O(z^3/n^2)$.
Hence
\[
n\log(1+zK'(-z/n))
=
-\kappa z^2+O(z^3/n),
\]
and combining the two terms yields
\[
\log \E[X_W(z)]
=
-\frac{\kappa}{2}z^2+O(z^3).
\]
This proves the claim.
\end{proof}

\section{Hardcore model with random fugacities}\label{sec:hardcore-first}
In the next two sections, we move on from the permanent and consider hardcore models on general graph. As a reminder (see the Overview), the permanent is roughly analogous to the hardcore model on the line graph of the complete bipartite graph $K_{n,n}$, which has maximum degree $n$ and ${n \choose 2}$ many vertices. So while $K_{n,n}$ is very dense, its line graph is somewhat sparser --- it has degree proportional to square-root of the number of vertices. On general graphs, this relative sparsity will play a crucial role in the analysis. 

\subsection{Model on a general graph}
Let $G = (V, E)$ be a finite graph with maximum degree $\Delta \ge 2$, with vertex fugacities $\lambda_v \in \mathbb C$.

Recall from earlier that the hardcore partition function $Z$ is the sum of the weights of all valid independent sets:
\begin{equation}
    Z = \sum_{I \in \mathcal{I}(G)} \prod_{v \in I} \lambda_v
\end{equation}
and the corresponding cluster expansion of the log partition function is given by:
\begin{equation}
    \log Z = \sum_{\mathbf{m} \neq \mathbf{0}} \frac{1}{\mathbf{m}!} \phi(H[\mathbf{m}]) \prod_{v \in V} \lambda_v^{m_v}
\end{equation}
where $\mathbf{m}! = \prod_{v \in V} m_v!$, and $\phi(H)$ is the Ursell function. For any graph $H = (W, E_H)$, the Ursell function evaluates the sum over all its connected spanning subgraphs:
\begin{equation}
    \phi(H) = \sum_{\substack{A \subseteq E_H \\ (W, A) \text{ is connected}}} (-1)^{|A|}
\end{equation}

\paragraph{Random fugacities.} We now focus on the case where $\lambda_v = \lambda W_v$ and each $W_v$ is an i.i.d. mean zero complex-valued random variable. The cluster expansion becomes:
\[ \sum_{k = 1}^{\infty} \lambda^k \sum_{|\mathbf{m}| = k} \frac{1}{\mathbf{m}!} \phi(H[\mathbf{m}]) \prod_{v \in V} W_v^{m_v} = \sum_{k = 1}^{\infty} a_k \lambda^k \]

\subsection{First-order reweighting}
Let $G = (V, E)$ be a finite graph with maximum degree $\Delta \ge 1$. We consider the hardcore model partition function with complex random fugacities $\lambda W_v$, where $W_v \sim \CN(0,1)$ are i.i.d. standard complex Gaussians. We define the reweighted partition function as:
\begin{equation}
    X(\lambda) = Z(\lambda W) \exp\left( -\lambda \sum_{v \in V} W_v \right)
\end{equation}

\begin{thm}
\label{thm:second_moment_clean}
Let $G=(V,E)$ be a finite graph with maximum degree $\Delta \ge 1$, and let
\[
Z(\lambda W)=\sum_{I\in\mathcal I(G)} \lambda^{|I|}\prod_{v\in I} W_v
\]
be the hardcore partition function with random fugacities $\lambda W_v$, where
$W_v \sim \CN(0,1)$ are i.i.d. standard complex Gaussians. Define
\[
X(\lambda):= Z(\lambda W)\exp\!\Bigl(-\lambda\sum_{v\in V} W_v\Bigr).
\]
Let $(G_n)_{n\ge1}$ be a sequence of graphs with vertex sets $V_n$ and maximum degrees $\Delta_n$, and let $\lambda_n\in\mathbb C$ be a sequence such that
\[
|\lambda_n|^4 \Delta_n |V_n| = o(1).
\]
Then
\[
\mathbb E\bigl[\,|X(\lambda_n)|^2\,\bigr]=1+o(1).
\]
\end{thm}

\begin{proof}
We first expand $X(\lambda)$ as a polynomial in the Gaussian variables $(W_v)_{v\in V}$.

\medskip
\noindent
\textbf{Step 1: Coefficient formula.}
Write
\[
Z(\lambda W)=\sum_{I\in\mathcal I(G)} \lambda^{|I|} W^I,
\qquad
W^I:=\prod_{v\in I} W_v,
\]
and
\[
\exp\!\Bigl(-\lambda\sum_{v\in V} W_v\Bigr)
=\prod_{v\in V}\Bigl(\sum_{r_v\ge0}\frac{(-\lambda W_v)^{r_v}}{r_v!}\Bigr).
\]
For a multi-index $m=(m_v)_{v\in V}\in(\mathbb Z_{\ge0})^V$, let
\[
|m|:=\sum_{v\in V} m_v,
\qquad
m!:=\prod_{v\in V} m_v!,
\qquad
\operatorname{supp}(m):=\{v\in V:\ m_v>0\}.
\]
Collecting the coefficient of $W^m:=\prod_v W_v^{m_v}$ gives
\[
X(\lambda)
=
\sum_{m\in(\mathbb Z_{\ge0})^V}
\lambda^{|m|}
\frac{(-1)^{|m|}}{m!}
\Biggl(
\sum_{\substack{I\in\mathcal I(G)\\ I\subseteq \operatorname{supp}(m)}}
\prod_{v\in I}(-m_v)
\Biggr)
W^m.
\]
The inner sum is exactly the hardcore partition function on the induced subgraph
$G[\operatorname{supp}(m)]$ with vertex fugacities $(-m_v)_{v\in\operatorname{supp}(m)}$, namely
\[
Z_{G[\operatorname{supp}(m)]}(-m)
:=
\sum_{I\in\mathcal I(G[\operatorname{supp}(m)])}
\prod_{v\in I} (-m_v).
\]
Therefore
\begin{equation}
\label{eq:X-expansion-clean}
X(\lambda)
=
\sum_{m\in(\mathbb Z_{\ge0})^V}
\lambda^{|m|}
\frac{(-1)^{|m|}}{m!}
Z_{G[\operatorname{supp}(m)]}(-m)\,W^m.
\end{equation}

\medskip
\noindent
\textbf{Step 2: Exact second-moment formula.}
Since the $W_v$ are independent standard complex Gaussians,
\[
\mathbb E\bigl[W^m \overline{W}^{\,m'}\bigr]=0
\qquad\text{unless }m=m',
\]
and
\[
\mathbb E\bigl[|W^m|^2\bigr]
=
\prod_{v\in V}\mathbb E\bigl[|W_v|^{2m_v}\bigr]
=
\prod_{v\in V} m_v!
=
m!.
\]
Thus, taking the expectation of the absolute square of \eqref{eq:X-expansion-clean},
all off-diagonal terms vanish and we obtain
\begin{equation}
\label{eq:second-moment-exact-clean}
\mathbb E\bigl[|X(\lambda)|^2\bigr]
=
\sum_{m\in(\mathbb Z_{\ge0})^V}
|\lambda|^{2|m|}
\frac{1}{m!}
\bigl|Z_{G[\operatorname{supp}(m)]}(-m)\bigr|^2.
\end{equation}
Every summand is nonnegative, and the $m=0$ term is equal to $1$.

\medskip
\noindent
\textbf{Step 3: Factorization over connected components.}
Fix $m$. Let $S:=\operatorname{supp}(m)$, and let
\[
G[S]=C_1 \sqcup \cdots \sqcup C_r
\]
be its decomposition into connected components. Since the hardcore partition function
factorizes over disjoint unions,
\[
Z_{G[S]}(-m)=\prod_{i=1}^r Z_{G[C_i]}(-m_{C_i}),
\]
where $m_{C_i}$ denotes the restriction of $m$ to $C_i$.
Also,
\[
|\lambda|^{2|m|}\frac{1}{m!}
=
\prod_{i=1}^r
\left(
|\lambda|^{2|m_{C_i}|}\frac{1}{m_{C_i}!}
\right).
\]
Hence the summand in \eqref{eq:second-moment-exact-clean} factors over connected components.

For a connected vertex set $C\subseteq V$ and a multiplicity vector
$m_C\in(\mathbb Z_{\ge1})^C$, define the polymer weight
\[
w(C,m_C)
:=
|\lambda|^{2|m_C|}
\frac{1}{m_C!}
\bigl|Z_{G[C]}(-m_C)\bigr|^2.
\]
Then \eqref{eq:second-moment-exact-clean} is exactly the partition function of a hard-core polymer gas:
\[
\mathbb E\bigl[|X(\lambda)|^2\bigr]
=
\sum_{\Gamma\ \text{compatible}}
\ \prod_{(C,m_C)\in\Gamma} w(C,m_C),
\]
where “compatible” means that the vertex sets of the polymers are pairwise disjoint.

Since all polymer weights are nonnegative, forgetting the compatibility constraint only enlarges the sum, and therefore
\begin{equation}
\label{eq:polymer-exp-upper-clean}
\mathbb E\bigl[|X(\lambda)|^2\bigr]
\le
\prod_{\substack{C\subseteq V\\ C\ \text{connected}}}
\ \prod_{m_C\in(\mathbb Z_{\ge1})^C}
\bigl(1+w(C,m_C)\bigr)
\le
\exp\!\Biggl(
\sum_{\substack{C\subseteq V\\ C\ \text{connected}}}
\ \sum_{m_C\in(\mathbb Z_{\ge1})^C}
w(C,m_C)
\Biggr).
\end{equation}
Let
\[
\mathcal W
:=
\sum_{\substack{C\subseteq V\\ C\ \text{connected}}}
\ \sum_{m_C\in(\mathbb Z_{\ge1})^C}
w(C,m_C).
\]
It remains to show that $\mathcal W=o(1)$.

\medskip
\noindent
\textbf{Step 4: Contributions of connected sets of size \(1\).}
If $C=\{v\}$, then
$Z_{G[\{v\}]}(-m_v)=1-m_v$.
Hence the contribution of a single vertex is
\[
\sum_{m\ge1} |\lambda|^{2m}\frac{|1-m|^2}{m!}
=
\sum_{m\ge2} |\lambda|^{2m}\frac{(m-1)^2}{m!}.
\]
For $|\lambda|\le1$, this is bounded by
\[
|\lambda|^4 \sum_{m\ge2}\frac{(m-1)^2}{m!}
\le c_1 |\lambda|^4
\]
for some absolute constant $c_1>0$. Summing over $v\in V$ gives
\[
\mathcal W_1 \le c_1 |V|\,|\lambda|^4.
\]

\medskip
\noindent
\textbf{Step 5: Contributions of connected sets of size \(2\).}
If $C=\{u,v\}$ is connected, then necessarily $u\sim v$, and
$Z_{G[C]}(-m_u,-m_v)=1-m_u-m_v$.
Therefore
\[
\sum_{m_u,m_v\ge1}
|\lambda|^{2(m_u+m_v)}
\frac{|1-m_u-m_v|^2}{m_u!m_v!}
\le c_2 |\lambda|^4
\]
for some absolute constant $c_2>0$ and all sufficiently small $|\lambda|$.
Summing over edges,
\[
\mathcal W_2
\le
c_2 |E|\,|\lambda|^4
\le
\frac{c_2}{2}|V|\Delta\,|\lambda|^4.
\]

\medskip
\noindent
\textbf{Step 6: Contributions of connected sets of size at least \(3\).}
Let $C\subseteq V$ be connected with $|C|=k\ge3$. By the triangle inequality,
\[
|Z_{G[C]}(-m_C)|
\le
\sum_{I\in\mathcal I(G[C])}\prod_{v\in I} m_v
\le
\sum_{S\subseteq C}\prod_{v\in S} m_v
=
\prod_{v\in C}(1+m_v).
\]
Hence
\[
w(C,m_C)
\le
\prod_{v\in C}
\left(
|\lambda|^{2m_v}\frac{(1+m_v)^2}{m_v!}
\right).
\]
Summing over $m_v\ge1$ independently,
\[
\sum_{m_C\in(\mathbb Z_{\ge1})^C} w(C,m_C)
\le
\left(
\sum_{m\ge1} |\lambda|^{2m}\frac{(1+m)^2}{m!}
\right)^k.
\]
For $|\lambda|\le1$, the one-variable sum is bounded by
\[
\sum_{m\ge1} |\lambda|^{2m}\frac{(1+m)^2}{m!}
\le c_3 |\lambda|^2
\]
for some absolute constant $c_3>0$. Therefore
\[
\sum_{m_C\in(\mathbb Z_{\ge1})^C} w(C,m_C)
\le
(c_3|\lambda|^2)^k.
\]

Now fix $k\ge3$. The number of connected induced subgraphs of $G$ with $k$ vertices
containing a given vertex is at most $(e\Delta)^{k-1}$, so the total number of connected
\(k\)-vertex subsets of \(V\) is at most
\[
\frac{|V|}{k}(e\Delta)^{k-1}
\le
|V|(e\Delta)^{k-1}.
\]
Hence
\[
\mathcal W_{\ge3}
\le
|V|\sum_{k\ge3}(e\Delta)^{k-1}(c_3|\lambda|^2)^k
=
\frac{|V|}{e\Delta}
\sum_{k\ge3}(c_3 e\Delta |\lambda|^2)^k.
\]

We claim that $\Delta |\lambda|^2\to0$ under the hypothesis
\(
|\lambda|^4 \Delta |V|=o(1).
\)
Indeed, since every graph satisfies \(|V|\ge \Delta+1\), we have
$|\lambda|^4 \Delta^2 \le |\lambda|^4 \Delta |V| = o(1)$,
and therefore \(\Delta |\lambda|^2=o(1)\).
So for all sufficiently large \(n\),
$c_3 e\Delta |\lambda|^2 \le \frac12$.
Thus the geometric series is bounded by a constant multiple of its first term:
\[
\mathcal W_{\ge3}
\le
c_4 |V| \Delta^2 |\lambda|^6
\]
for some absolute constant \(c_4>0\).

\medskip
\noindent
\textbf{Step 7: Conclusion.}
Combining the three estimates,
\[
\mathcal W
\le
c_1 |V||\lambda|^4
+
\frac{c_2}{2}|V|\Delta |\lambda|^4
+
c_4 |V|\Delta^2 |\lambda|^6.
\]
By the hypothesis,
$|V|\Delta |\lambda|^4=o(1)$,
and also
\[
|V|\Delta^2 |\lambda|^6
=
\bigl(|V|\Delta |\lambda|^4\bigr)\bigl(\Delta |\lambda|^2\bigr)
=o(1).
\]
Hence \(\mathcal W=o(1)\). Returning to \eqref{eq:polymer-exp-upper-clean},
\[
\mathbb E\bigl[|X(\lambda)|^2\bigr]
\le
\exp(\mathcal W)
=
1+o(1).
\]
On the other hand, the \(m=0\) term in \eqref{eq:second-moment-exact-clean} is \(1\), and all
other terms are nonnegative, so
$\mathbb E\bigl[|X(\lambda)|^2\bigr]\ge1$.
Therefore
\[
\mathbb E\bigl[|X(\lambda_n)|^2\bigr]=1+o(1),
\]
as claimed.
\end{proof}

\vspace{0.5cm}

To illustrate how this theorem is applied, we show how to recover a weaker version of the result from Theorem~\ref{thm:linegraph-unconditional}.

\begin{cor}[Weaker version of Theorem~\ref{thm:linegraph-unconditional}]
Let $G_n = L(K_{n,n})$ be the line graph of the complete bipartite graph, and let $W_{ij} \sim \CN(0,1)$.  Let $X_n(\lambda) = P_n(\lambda) \exp(-\lambda \sum_{i,j} W_{ij})$. If $|\lambda_n| = o(n^{-3/4})$, then $\mathbb{E}[|X_n(\lambda_n)|^2] = 1 + o(1)$.
\end{cor}

\begin{proof}
The independent sets of $L(K_{n,n})$ correspond exactly to the matchings of $K_{n,n}$.

For the line graph $G_n = L(K_{n,n})$, the number of vertices is $|V_n| = n^2$. The maximum degree is the number of edges adjacent to a given edge in $K_{n,n}$, which is $\Delta_n = 2n - 2$. 

Applying Theorem \ref{thm:second_moment_clean}, the sufficient condition for the relative variance to vanish is $|\lambda_n|^4 (2n - 2) (n^2) = o(1)$, which is equivalent to $|\lambda_n|^4 n^3 = o(1)$, or $|\lambda_n| = o(n^{-3/4})$. 
\end{proof}



\section{Universality of the zero-free region for the hardcore model}\label{sec:hardcore-universality}
We now prove an analogue of Theorem~\ref{thm:second_moment_clean} for general i.i.d.\ complex vertex weights, without assuming rotational symmetry. The key input is cancellation after averaging over the phase \(\theta\), together with a cluster expansion for a suitable \emph{pair model} representing the second moment. Part of the technical estimate for the auxiliary cluster expansion is deferred to Lemma~\ref{thm:linegraph-unconditional} at the end.

\begin{thm}\label{thm:hardcore-universality}
Let \(G_n=(V_n,E_n)\) be a sequence of finite graphs with maximum degree \(\Delta_n\). Let \(W_v\) be i.i.d.\ complex random variables with
\[
\mathbb E[W_v]=0,
\qquad
\mathbb E|W_v|^2=\nu,
\]
and assume that for some \(K>0\),
\[
\sup_{\phi\in[0,2\pi)}
\bigl\|\operatorname{Re}(e^{-i\phi}W_v)\bigr\|_{\psi_1}
\le K\sqrt{\nu}.
\]
For \(\lambda\in\mathbb C\), define
\[
Z_G(\lambda W):=\sum_{I\in\mathcal I(G)} \lambda^{|I|}\prod_{v\in I} W_v,
\qquad
X_W(\lambda):=Z_G(\lambda W)\exp\!\Bigl(-\lambda\sum_{v\in V}W_v\Bigr).
\]
If \(R_n>0\) satisfies
\[
\nu^2 R_n^4 (1+\Delta_n)|V_n|=o(1),
\]
then
\[
\mathbb E_\theta \log \mathbb E_W \bigl|X_W(R_n e^{i\theta})\bigr|^2 = o(1),
\]
where \(\theta\sim\mathrm{Unif}[0,2\pi)\).
\end{thm}

\begin{lem}\label{lem:local-cumulant}
There exist constants \(r_0=r_0(K)>0\) and \(A_K,B_K>0\), depending only on \(K\), such that
\[
F(a,b):=\log \mathbb E e^{aW+b\overline W}
\]
is holomorphic on the bidisc
\[
D:=\left\{(a,b)\in\mathbb C^2:\ |a|,|b|\le \frac{r_0}{\sqrt{\nu}}\right\},
\]
and has an absolutely convergent expansion
\[
F(a,b)=\sum_{p+q\ge 2} c_{pq} a^p b^q
\]
with
\[
|c_{pq}|\le A_K\, B_K^{\,p+q}\,\nu^{(p+q)/2}.
\]
Moreover,
$c_{11}=\nu$.
\end{lem}

\begin{proof}
Write \(W=X+iY\). Taking \(\phi=0\) and \(\phi=\pi/2\) in the hypothesis gives
\[
\|X\|_{\psi_1}\le K\sqrt{\nu},
\qquad
\|Y\|_{\psi_1}\le K\sqrt{\nu}.
\]
Since \(|W|\le |X|+|Y|\), we obtain
\[
\|W\|_{\psi_1}\le C_K\sqrt{\nu},
\]
for some constant \(C_K\) depending only on \(K\). Hence the standard
sub-exponential moment bound yields
\[
\mathbb E|W|^m \le (C_K \sqrt{\nu}\, m)^m,
\qquad m\ge 1.
\]

Set
\[
M(a,b):=\mathbb E e^{aW+b\overline W}.
\]
For \(s:=|a|+|b|\), we have
\[
\sum_{p,q\ge 0}
\frac{|a|^p |b|^q}{p!\,q!}\,\mathbb E|W|^{p+q}
=
\sum_{m\ge 0}\frac{s^m}{m!}\,\mathbb E|W|^m
\le
1+\sum_{m\ge 1}\frac{(C_K\sqrt{\nu}\, m s)^m}{m!}.
\]
Using \(m!\ge (m/e)^m\), we obtain
\[
\frac{(C_K\sqrt{\nu}\, m s)^m}{m!}
\le
(eC_K s\sqrt{\nu})^m.
\]
Therefore the series converges absolutely whenever \(s\le c_K/\sqrt{\nu}\) for
\(c_K>0\) sufficiently small. In particular, \(M(a,b)\) is holomorphic on a
neighborhood of a bidisc
\[
\bigl\{|a|,|b|\le 2r_0/\sqrt{\nu}\bigr\}
\]
for some \(r_0=r_0(K)>0\).

Shrinking \(r_0\) if necessary, the same estimate gives
\[
|M(a,b)-1|
\le
\sum_{m\ge 1}(eC_K(|a|+|b|)\sqrt{\nu})^m
\le \frac12
\]
throughout \(|a|,|b|\le 2r_0/\sqrt{\nu}\). Thus \(M\) has no zeros there, and
the principal branch
\[
F(a,b):=\log M(a,b)
\]
is holomorphic on that bidisc. Moreover \(|F(a,b)|\le C_K\) there.

Now apply Cauchy's estimate on the smaller bidisc
\[
D=\left\{|a|,|b|\le \frac{r_0}{\sqrt{\nu}}\right\}.
\]
If
\[
F(a,b)=\sum_{p,q\ge 0} c_{pq} a^p b^q,
\]
then
\[
|c_{pq}|
\le
C_K\Bigl(\frac{\sqrt{\nu}}{r_0}\Bigr)^{p+q}
=
A_K\,B_K^{\,p+q}\,\nu^{(p+q)/2}.
\]
Since \(M(0,0)=1\), we have \(F(0,0)=0\). Since \(\mathbb EW=0\),
\[
\partial_a F(0,0)=\mathbb EW=0,
\qquad
\partial_b F(0,0)=\mathbb E\overline W=0,
\]
so there are no linear terms. Finally,
\[
c_{11}
=
\partial_a\partial_b F(0,0)
=
\partial_a\partial_b \log M(0,0)
=
\mathbb E|W|^2-(\mathbb EW)(\mathbb E\overline W)
=
\nu.
\]
This proves the lemma.
\end{proof}

\begin{proof}[Proof of Theorem~\ref{thm:hardcore-universality}]
Write \(N:=|V_n|\), \(\Delta:=\Delta_n\), and \(R:=R_n\). Since
\[
\nu^2R_n^4(1+\Delta_n)|V_n|=o(1),
\]
we in particular have \(\nu R_n^2=o(1)\). Hence for all sufficiently large \(n\),
\[
R_n\le \frac{r_0}{\sqrt{\nu}},
\]
so Lemma~\ref{lem:local-cumulant} is applicable. It therefore suffices to prove
the claimed estimate for all large \(n\).

Fix such an \(n\), and set
\[
G=(V,E),\qquad N:=|V|,\qquad \lambda=Re^{i\theta},\qquad u:=-\lambda.
\]
Define
\[
M(a,b):=\mathbb E e^{aW+b\overline W},
\qquad
F(a,b):=\log M(a,b),
\]
and
\[
\mu(a,b):=\partial_a F(a,b),
\qquad
\tau(a,b):=\partial_a\partial_b F(a,b)+\partial_aF(a,b)\partial_bF(a,b).
\]
Equivalently,
\[
\mu(a,b)=\frac{\mathbb E\!\left[W e^{aW+b\overline W}\right]}{M(a,b)},
\qquad
\tau(a,b)=\frac{\mathbb E\!\left[|W|^2 e^{aW+b\overline W}\right]}{M(a,b)}.
\]
Along the slice \(b=\overline a\), abbreviate
\[
\Psi(u):=F(u,\overline u),\qquad
\mu(u):=\mu(u,\overline u),\qquad
\tau(u):=\tau(u,\overline u).
\]

\medskip
\noindent
\textbf{Step 1: exact pair-model representation.}
By definition,
\[
X_W(\lambda)=Z_G(\lambda W)\exp\!\Bigl(-\lambda\sum_{v\in V}W_v\Bigr),
\]
hence
\[
|X_W(\lambda)|^2
=
\sum_{I,J\in\mathcal I(G)}
\lambda^{|I|}\overline\lambda^{\,|J|}
\Bigl(\prod_{v\in I}W_v\Bigr)
\Bigl(\prod_{v\in J}\overline{W_v}\Bigr)
\exp\!\Bigl(
u\sum_{v\in V}W_v+\overline u\sum_{v\in V}\overline{W_v}
\Bigr).
\]
Taking expectation and using independence across vertices,
\begin{align*}
\mathbb E_W |X_W(\lambda)|^2
&=
\sum_{I,J\in\mathcal I(G)}
\prod_{v\in V}
\mathbb E\Bigl[
e^{uW_v+\overline u\,\overline W_v}
(\lambda W_v)^{\mathbf 1}_{v\in I}
(\overline\lambda\,\overline W_v)^{\mathbf 1}_{v\in J}
\Bigr] \\
&=
M(u,\overline u)^N
\sum_{I,J\in\mathcal I(G)}
a_L(u,\lambda)^{|I\setminus J|}
a_R(u,\lambda)^{|J\setminus I|}
a_B(u,\lambda)^{|I\cap J|},
\end{align*}
where
\[
a_L(u,\lambda):=\lambda\,\mu(u),\qquad
a_R(u,\lambda):=\overline\lambda\,\overline{\mu(u)},\qquad
a_B(u,\lambda):=|\lambda|^2\,\tau(u).
\]
Thus
\begin{equation}\label{eq:pair-factorization-clean-fixed}
\mathbb E_W |X_W(\lambda)|^2
=
e^{N\Psi(u)}\,
Z_{\mathrm{pair}}\bigl(a_L(u,\lambda),a_R(u,\lambda),a_B(u,\lambda)\bigr),
\end{equation}
where \(Z_{\mathrm{pair}}\) is the pair-model partition function from
Lemma~\ref{lem:pair-model-kp}.

\medskip
\noindent
\textbf{Step 2: local expansions and size bounds.}
By Lemma~\ref{lem:local-cumulant},
\[
F(a,b)=\sum_{p+q\ge 2} c_{pq}a^p b^q,
\qquad
|c_{pq}|\le A_K B_K^{p+q}\nu^{(p+q)/2},
\qquad
c_{11}=\nu.
\]
Hence
\[
\Psi(u)=\sum_{p+q\ge 2} c_{pq}u^p\overline u^{\,q}.
\]

Differentiating term-by-term,
\[
\mu(a,b)=\partial_aF(a,b)=\sum_{p+q\ge 1} m_{pq}a^p b^q,
\qquad
m_{pq}=(p+1)c_{p+1,q}.
\]
Using \((p+1)\le 2^{p+1}\), we may absorb the factor \(p+1\) into the geometric
constant and obtain
\[
|m_{pq}|\le A'_K (B'_K)^{p+q+1}\nu^{(p+q+1)/2}.
\]
In particular,
\[
m_{01}=c_{11}=\nu.
\]

Similarly,
\[
\partial_a\partial_b F(a,b)=\nu+\sum_{p+q\ge 1} d_{pq}a^p b^q,
\qquad
|d_{pq}|\le A'_K (B'_K)^{p+q+2}\nu^{(p+q+2)/2}.
\]
Since \(\partial_aF\) and \(\partial_bF\) have no constant term, their product
also has an absolutely convergent expansion whose coefficients obey the same
type of geometric bound after enlarging constants if necessary. Thus
\[
\tau(a,b)=\nu+\sum_{p+q\ge 1} t_{pq}a^p b^q,
\qquad
|t_{pq}|\le A''_K (B''_K)^{p+q+2}\nu^{(p+q+2)/2}.
\]

Therefore, after shrinking the allowed radius once more if necessary, there is a
constant \(c_K>0\) such that whenever \(|u|\le c_K/\sqrt{\nu}\),
\begin{equation}\label{eq:activity-size-clean-fixed}
|a_L(u,\lambda)|+|a_R(u,\lambda)|+|a_B(u,\lambda)|
\le C_K \nu R^2.
\end{equation}
Indeed, for \(|u|\le c_K/\sqrt{\nu}\) the above series are dominated by geometric
series in \(B_K\sqrt{\nu}|u|\).

\medskip
\noindent
\textbf{Step 3: angular averages and quadratic cancellation.}

\smallskip
\noindent
\emph{The \(e^{N\Psi(u)}\) term.}
Since \(u=-Re^{i\theta}\),
\[
\Psi(-Re^{i\theta})
=
\sum_{p+q\ge 2} c_{pq}(-1)^{p+q}R^{p+q}e^{i(p-q)\theta}.
\]
Averaging over \(\theta\), only the diagonal terms \(p=q\) survive:
\[
\mathbb E_\theta \Psi(-Re^{i\theta})
=
\sum_{p\ge 1} c_{pp}R^{2p}.
\]
Using \(c_{11}=\nu\) and the geometric bound on \(c_{pp}\), we get
\begin{equation}\label{eq:Psi-average-clean-fixed}
\mathbb E_\theta \Psi(-Re^{i\theta})
=
\nu R^2+O_K(\nu^2R^4).
\end{equation}

\smallskip
\noindent
\emph{The \(a_L\) term.}
We have
\[
a_L(u,\lambda)=\lambda\mu(u),
\]
so
\[
a_L(-Re^{i\theta},Re^{i\theta})
=
\sum_{p+q\ge 1}
(-1)^{p+q}m_{pq}R^{p+q+1}e^{i(p-q+1)\theta}.
\]
After averaging over \(\theta\), only the resonant modes \(p-q+1=0\) survive.
The leading one is \((p,q)=(0,1)\), which contributes
\[
(-1)^1 m_{01}R^2=-\nu R^2.
\]
All other resonant terms satisfy \(p+q+1\ge 4\), hence
\begin{equation}\label{eq:aL-average-clean-fixed}
\mathbb E_\theta a_L(-Re^{i\theta},Re^{i\theta})
=
-\nu R^2+O_K(\nu^2R^4).
\end{equation}
By complex conjugation,
\begin{equation}\label{eq:aR-average-clean-fixed}
\mathbb E_\theta a_R(-Re^{i\theta},Re^{i\theta})
=
-\nu R^2+O_K(\nu^2R^4).
\end{equation}

\smallskip
\noindent
\emph{The \(a_B\) term.}
Since \(a_B(u,\lambda)=R^2\tau(u)\), we have
\[
\tau(-Re^{i\theta})
=
\nu+\sum_{p+q\ge 1} t_{pq}(-1)^{p+q}R^{p+q}e^{i(p-q)\theta}.
\]
Averaging over \(\theta\) kills all off-diagonal terms, giving
\[
\mathbb E_\theta \tau(-Re^{i\theta})
=
\nu+\sum_{p\ge 1} t_{pp}R^{2p}
=
\nu+O_K(\nu^2R^2).
\]
Therefore
\begin{equation}\label{eq:aB-average-clean-fixed}
\mathbb E_\theta a_B(-Re^{i\theta},Re^{i\theta})
=
\nu R^2+O_K(\nu^2R^4).
\end{equation}

\medskip
\noindent
\textbf{Step 4: application of the KP lemma.}
By \eqref{eq:activity-size-clean-fixed},
\[
\rho(\theta):=\max\{|a_L|,|a_R|,|a_B|\}\le C_K \nu R^2.
\]
We can apply our cluster expansion estimate (Lemma~\ref{lem:pair-model-kp} below), provided that
\[
\nu R^2 \Delta=o(1).
\]
This follows from the hypothesis, since \(\Delta\le N\) for every finite graph,
and therefore
\[
(\nu R^2\Delta)^2
=
\nu^2R^4\Delta^2
\le
\nu^2R^4\,\Delta N
\le
\nu^2R^4(1+\Delta)N
=
o(1).
\]
Hence for all sufficiently large \(n\),
\[
\rho(\theta)\le \frac{c_{\mathrm{KP}}}{\Delta}
\]
uniformly in \(\theta\), and Lemma~\ref{lem:pair-model-kp} gives
\[
\log Z_{\mathrm{pair}}(a_L,a_R,a_B)
=
N(a_L+a_R+a_B)+\mathcal R(\theta),
\]
with
\begin{equation}\label{eq:Rtheta-clean-fixed}
|\mathcal R(\theta)|
\le
C_K N\Delta (\nu R^2)^2
=
C_K N\Delta \nu^2R^4.
\end{equation}
Combining this with \eqref{eq:pair-factorization-clean-fixed}, we obtain
\[
\log \mathbb E_W |X_W(\lambda)|^2
=
N\Psi(u)+N(a_L+a_R+a_B)+\mathcal R(\theta).
\]

\medskip
\noindent
\textbf{Step 5: averaging in \(\theta\).}
Taking \(\mathbb E_\theta\) and using
\eqref{eq:Psi-average-clean-fixed},
\eqref{eq:aL-average-clean-fixed},
\eqref{eq:aR-average-clean-fixed},
\eqref{eq:aB-average-clean-fixed},
and \eqref{eq:Rtheta-clean-fixed}, we get
\begin{align*}
\mathbb E_\theta \log \mathbb E_W |X_W(Re^{i\theta})|^2
&=
N\bigl(\nu R^2+O_K(\nu^2R^4)\bigr) \\
&\quad
+N\bigl(-\nu R^2-\nu R^2+\nu R^2+O_K(\nu^2R^4)\bigr)
+O_K(N\Delta \nu^2R^4).
\end{align*}
The quadratic terms cancel, leaving
\[
\mathbb E_\theta \log \mathbb E_W |X_W(Re^{i\theta})|^2
=
O_K(N\nu^2R^4)+O_K(N\Delta \nu^2R^4)
=
O_K\bigl(N(1+\Delta)\nu^2R^4\bigr).
\]
Applying this to \(G=G_n\), \(N=|V_n|\), \(\Delta=\Delta_n\), and \(R=R_n\), the
assumption
\[
\nu^2R_n^4(1+\Delta_n)|V_n|=o(1)
\]
implies
\[
\mathbb E_\theta \log \mathbb E_W \bigl|X_W(R_ne^{i\theta})\bigr|^2=o(1).
\]
This proves the theorem.
\end{proof}
\subsection{Cluster expansion of pair model}
We explain in detail the cluster-expansion estimate for the pair model used above.

\begin{lem}\label{lem:pair-model-kp}
There exist absolute constants \(c_{\mathrm{KP}},C_{\mathrm{KP}}>0\) such that the following holds.
Let \(G=(V,E)\) be a finite graph with maximum degree \(\Delta\), and let
\[
a_L,a_R,a_B\in\mathbb C.
\]
Consider the finite-state spin model on \(V\) with local states \(\{0,L,R,B\}\), vertex weights
\[
w(0)=1,\qquad w(L)=a_L,\qquad w(R)=a_R,\qquad w(B)=a_B,
\]
and hard-core constraints that
\begin{itemize}
    \item adjacent vertices may not both lie in $\{L,B\}$,
    \item and adjacent vertices may not both lie in $\{R,B\}.$
\end{itemize}
Let \(Z_{\mathrm{pair}}(a_L,a_R,a_B)\) denote its partition function given by the weighted
sum over all valid configurations in $\{0,L,R,B\}^V$. If
\[
\rho:=\max\{|a_L|,|a_R|,|a_B|\}\le \frac{c_{\mathrm{KP}}}{\Delta},
\]
then
\[
\log Z_{\mathrm{pair}}(a_L,a_R,a_B)
=
|V|(a_L+a_R+a_B)+\mathcal R,
\]
with
\[
|\mathcal R|\le C_{\mathrm{KP}}\,|V|\,\Delta\,\rho^2.
\]
\end{lem}

\begin{proof}
We rewrite the model as an abstract polymer gas and then apply the
Koteck\'y--Preiss (KP) criterion. See \cite{friedli2017statistical,perkins2023five}
for background on the terminology used.

\medskip
\noindent
\textbf{1. Polymer representation.}
Let us call a site \emph{occupied} if its spin lies in \(\{L,R,B\}\), and
\emph{vacant} if its spin is \(0\).
For two nonzero labels \(s,t\in\{L,R,B\}\), write \(s\sim_{\mathrm{adm}} t\) if
the pair \((s,t)\) is allowed across an edge.
By the stated constraints, the only allowed adjacent nonzero pairs are
\[
(L,R)\quad\text{and}\quad (R,L).
\]
Indeed, \(L\)-\(L\), \(R\)-\(R\), \(L\)-\(B\), \(B\)-\(L\), \(R\)-\(B\), \(B\)-\(R\),
and \(B\)-\(B\) are all forbidden.

A \emph{polymer} is a pair \(\gamma=(S,\sigma)\) where:

\begin{itemize}
\item \(S\subseteq V\) is nonempty and connected in \(G\);
\item \(\sigma:S\to\{L,R,B\}\);
\item \(\sigma\) is internally admissible on \(S\), meaning that for every edge
\(xy\in E(G[S])\), the pair \((\sigma(x),\sigma(y))\) is allowed.
\end{itemize}

The activity of \(\gamma=(S,\sigma)\) is
\[
z(\gamma):=\prod_{v\in S} a_{\sigma(v)}.
\]
In particular,
\[
|z(\gamma)|\le \rho^{|S|}.
\]

Two polymers \(\gamma=(S,\sigma)\) and \(\gamma'=(T,\tau)\) are said to be
\emph{compatible}, written \(\gamma\sim\gamma'\), if
\[
\operatorname{dist}_G(S,T)\ge 2.
\]
Equivalently, they are compatible iff \(S\cap T=\varnothing\) and there is no edge
of \(G\) joining a vertex of \(S\) to a vertex of \(T\).

We claim that
\[
Z_{\mathrm{pair}}(a_L,a_R,a_B)
=
\sum_{\Gamma\ \mathrm{compatible}}
\prod_{\gamma\in\Gamma} z(\gamma),
\]
where the sum runs over all finite compatible families of polymers.

To see this, start from a spin configuration \(\eta:V\to\{0,L,R,B\}\) satisfying
the edge constraints, and let
\[
U:=\{v\in V:\eta(v)\neq 0\}
\]
be its occupied set.
Decompose \(U\) into connected components in the graph \(G\):
\[
U=S_1\sqcup \cdots \sqcup S_m.
\]
For each component \(S_i\), the restricted labeling \(\eta|_{S_i}\) is internally
admissible, so \(\gamma_i:=(S_i,\eta|_{S_i})\) is a polymer.
Since \(S_i,S_j\) are distinct connected components of \(U\), there is no edge
between them, hence \(\operatorname{dist}_G(S_i,S_j)\ge 2\); thus the polymers
\(\gamma_1,\dots,\gamma_m\) are pairwise compatible.
Conversely, any compatible family of polymers determines a unique admissible spin
configuration by placing the given nonzero labels on the supports of the polymers
and putting \(0\) elsewhere.
The weight factorizes multiplicatively, which proves the identity.

\medskip
\noindent
\textbf{2. A counting bound.}
Fix a vertex \(x\in V\). For \(k\ge 1\), the number of connected \(k\)-vertex
subsets \(S\subseteq V\) containing \(x\) is at most \((e\Delta)^{k-1}\).
For each such support \(S\), there are at most \(3^k\) choices of the labeling
\(\sigma:S\to\{L,R,B\}\), and hence at most \(3^k\) polymers with support \(S\).
Therefore
\[
\sum_{\substack{\gamma\ni x\\ |\gamma|=k}} |z(\gamma)| e^{|\gamma|}
\le
(e\Delta)^{k-1}\,3^k\,\rho^k\,e^k
=
3e\rho \,(3e^2\Delta\rho)^{k-1}.
\]
Summing over \(k\ge 1\), we obtain
\begin{equation}\label{eq:Sx-bound}
\sum_{\gamma\ni x} |z(\gamma)| e^{|\gamma|}
\le
3e\rho \sum_{k\ge 1}(3e^2\Delta\rho)^{k-1}.
\end{equation}
Hence, provided
$3e^2\Delta\rho\le \frac12$,
the geometric series converges and yields
\begin{equation}\label{eq:Sx-final}
\sum_{\gamma\ni x} |z(\gamma)| e^{|\gamma|}
\le C_0\,\rho
\end{equation}
for an absolute constant \(C_0\).

Now let \(\gamma=(S,\sigma)\) be a fixed polymer.
If \(\eta=(T,\tau)\) is incompatible with \(\gamma\), then
\(\operatorname{dist}_G(S,T)\le 1\), so in particular \(T\) contains some vertex
of the closed neighborhood
\[
N[S]:=S\cup\{y\in V:\exists x\in S,\ xy\in E\}.
\]
Thus, using \eqref{eq:Sx-final},
\begin{align}
\sum_{\eta\not\sim \gamma} |z(\eta)| e^{|\eta|}
&\le
\sum_{x\in N[S]} \sum_{\eta\ni x} |z(\eta)| e^{|\eta|}
\nonumber\\
&\le
|N[S]|\, C_0\,\rho
\nonumber\\
&\le
(\Delta+1)|S|\, C_0\,\rho.
\label{eq:kp-left}
\end{align}

\medskip
\noindent
\textbf{3. Verification of the KP condition.}
Set
\[
a(\gamma):=|\gamma|=|S|.
\]
Choose \(c_{\mathrm{KP}}>0\) small enough so that whenever
\[
\rho\le \frac{c_{\mathrm{KP}}}{\Delta},
\]
both
$3e^2\Delta\rho\le \frac12$
and
$(\Delta+1)C_0\rho \le 1$
hold.
Then \eqref{eq:kp-left} gives
\[
\sum_{\eta\not\sim\gamma} |z(\eta)| e^{a(\eta)}
\le a(\gamma)
\qquad\text{for every polymer }\gamma.
\]
This is precisely the Koteck\'y--Preiss criterion for the abstract polymer gas.

Therefore the cluster expansion for \(\log Z_{\mathrm{pair}}\) converges absolutely:
\[
\log Z_{\mathrm{pair}}
=
\sum_{\mathcal X} \phi^T(\mathcal X)\prod_{\gamma} z(\gamma)^{\mathcal X(\gamma)},
\]
where the sum is over all finitely supported polymer multisets \(\mathcal X\), and
\(\phi^T(\mathcal X)\) denotes the usual truncated cluster coefficient.

\medskip
\noindent
\textbf{4. Extraction of the linear term.}
The singleton clusters are exactly the one-polymer clusters \(\mathcal X=\{\gamma\}\),
and their contribution is \(\sum_\gamma z(\gamma)\).
Among these, the polymers of size \(1\) are exactly
\[
(\{v\},L),\qquad (\{v\},R),\qquad (\{v\},B),
\qquad v\in V,
\]
with activities \(a_L,a_R,a_B\).
Hence
\[
\sum_{|\gamma|=1} z(\gamma)=|V|(a_L+a_R+a_B).
\]

So it remains to show that all other contributions are
\(O(|V|\Delta\rho^2)\).
Write
\[
\mathcal R
:=
\log Z_{\mathrm{pair}}-|V|(a_L+a_R+a_B)
=
R_{\ge 2}^{(1)}+R_{\mathrm{clust}},
\]
where
\[
R_{\ge 2}^{(1)}:=\sum_{|\gamma|\ge 2} z(\gamma)
\]
is the contribution of one-polymer clusters of size at least \(2\), and
\(R_{\mathrm{clust}}\) is the sum of all clusters involving at least two polymers.

\medskip
\noindent
\emph{Bound on \(R_{\ge 2}^{(1)}\).}
Using the same counting estimate as above,
\begin{align*}
\sum_{|\gamma|\ge 2} |z(\gamma)|
&\le
\sum_{x\in V}\ \sum_{\substack{\gamma\ni x\\ |\gamma|\ge 2}} |z(\gamma)| \\
&\le
|V| \sum_{k\ge 2} (e\Delta)^{k-1} 3^k \rho^k \\
&=
|V|\, 3\rho \sum_{k\ge 2} (3e\Delta\rho)^{k-1}.
\end{align*}
If \(c_{\mathrm{KP}}\) is chosen smaller if necessary, then \(3e\Delta\rho\le 1/2\),
and therefore
\[
|R_{\ge 2}^{(1)}|
\le
C_1 |V|\,\Delta\,\rho^2
\]
for an absolute constant \(C_1\).

\medskip
\noindent
\emph{Bound on \(R_{\mathrm{clust}}\).}
A standard corollary of the KP theorem \cite{perkins2023five} gives the estimate
\begin{equation}\label{eq:nontrivial-clusters-bound}
|R_{\mathrm{clust}}|
\le
C_2 \sum_{\gamma} |z(\gamma)| e^{a(\gamma)}
\sum_{\eta\not\sim \gamma} |z(\eta)| e^{a(\eta)},
\end{equation}
for an absolute constant \(C_2\).
Using \eqref{eq:kp-left},
\[
\sum_{\eta\not\sim \gamma} |z(\eta)| e^{a(\eta)}
\le C_0(\Delta+1)|\gamma|\rho.
\]
Therefore
\[
|R_{\mathrm{clust}}|
\le
C_3 \Delta\rho \sum_{\gamma} |\gamma|\, |z(\gamma)| e^{|\gamma|}.
\]
Now
\[
\sum_{\gamma} |\gamma|\, |z(\gamma)| e^{|\gamma|}
=
\sum_{x\in V}\sum_{\gamma\ni x}|z(\gamma)|e^{|\gamma|}
\le
|V|\, C_0\, \rho
\]
by \eqref{eq:Sx-final}. Hence
\[
|R_{\mathrm{clust}}|
\le
C_4 |V|\,\Delta\,\rho^2.
\]

Combining the bounds for \(R_{\ge 2}^{(1)}\) and \(R_{\mathrm{clust}}\), we obtain
\[
|\mathcal R|
\le
C_{\mathrm{KP}}\, |V|\,\Delta\,\rho^2
\]
for a suitable absolute constant \(C_{\mathrm{KP}}\), as claimed.
\end{proof}

\subsection{Zero-free region and algorithm}

We now record the zero-free consequence of Theorem~\ref{thm:hardcore-universality}.
For a holomorphic function \(f\not\equiv 0\), write \(N_f(r)\) for the number of
zeros of \(f\) in \(\mathbb D(0,r)\), counted with multiplicity.

For a fixed weighted instance \(W=(W_v)_{v\in V}\), define
\[
G_W(\lambda):=Z_G(\lambda W)
=
\sum_{I\in\mathcal I(G)} \lambda^{|I|}\prod_{v\in I} W_v.
\]
Since
\[
X_W(\lambda)=G_W(\lambda)\exp\!\Bigl(-\lambda\sum_{v\in V}W_v\Bigr),
\]
the functions \(G_W\) and \(X_W\) have exactly the same zeros.

\begin{thm}[Quantitative zero-count bound]
\label{thm:hardcore-zero-count}
There exist constants \(c_{\ast}(K),C_{\ast}(K)>0\), depending only on \(K\), such that
the following holds.

Let \(G=(V,E)\) be a finite graph with maximum degree \(\Delta\), and let
\(W_v\) be i.i.d.\ complex random variables satisfying the assumptions of
Theorem~\ref{thm:hardcore-universality}. Let \(0<r<R\), and assume that
\[
R\le \frac{c_{\ast}(K)}{\sqrt{\nu}},
\qquad
\nu R^{2}\Delta \le c_{\ast}(K).
\]
Then
\[
\mathbb E\,N_{G_W}(r)
\le
\frac{C_{\ast}(K)\,|V|\,(1+\Delta)\,\nu^{2}R^{4}}{2\log(R/r)}.
\]
In particular,
\[
\mathbb P\!\left[G_W \text{ has a zero in } \mathbb D(0,r)\right]
\le
\frac{C_{\ast}(K)\,|V|\,(1+\Delta)\,\nu^{2}R^{4}}{2\log(R/r)}.
\]
\end{thm}

\begin{proof}
By Lemma~\ref{lem:core-jensen-app}, the application of Jensen's formula, we have that
\[ \mathbb E\,N_{G_W}(r)
=
\mathbb E\,N_{X_W}(r)
\le
\frac{1}{2\log(R/r)}
\mathbb E_\theta \log \mathbb E_W |X_W(Re^{i\theta})|^2.
\]

Now inspect the proof of Theorem~\ref{thm:hardcore-universality}. Under the present
smallness assumptions \(R\le c_{\ast}(K)/\sqrt{\nu}\) and \(\nu R^2\Delta\le c_{\ast}(K)\),
the same argument gives the quantitative bound
\[
\mathbb E_\theta \log \mathbb E_W |X_W(Re^{i\theta})|^2
\le
C_{\ast}(K)\,|V|\,(1+\Delta)\,\nu^2 R^4.
\]
Substituting this into the previous display yields
\[
\mathbb E\,N_{G_W}(r)
\le
\frac{C_{\ast}(K)\,|V|\,(1+\Delta)\,\nu^{2}R^{4}}{2\log(R/r)}.
\]
Finally, Markov's inequality gives
\[
\mathbb P\!\left[N_{G_W}(r)\ge 1\right]
\le
\mathbb E\,N_{G_W}(r),
\]
which proves the theorem.
\end{proof}

\begin{thm}
\label{thm:hardcore-zero-free-sequence}
Let \(G_n=(V_n,E_n)\) be a sequence of finite graphs with maximum degree
\(\Delta_n\), and set
\[
M_n:=(1+\Delta_n)|V_n|.
\]
Fix \(\beta\in(0,1/4)\), and define
\[
R_n:=\frac{1}{\sqrt{\nu}}\,M_n^{-1/4-\beta}.
\]
Then for every fixed \(\varepsilon\in(0,1)\), and all sufficiently large \(n\),
\[
\mathbb P\!\left[
Z_{G_n}(\lambda W)\text{ has a zero in }\mathbb D\!\bigl(0,(1-\varepsilon)R_n\bigr)
\right]
=
O_{\varepsilon,K}\!\left(M_n^{-4\beta}\right).
\]
\end{thm}

\begin{proof}
Set
\[
R:=R_n,
\qquad
r:=(1-\varepsilon/2)R.
\]
Since
\[
\nu^2 R^4 M_n = M_n^{-4\beta},
\]
it suffices to apply Theorem~\ref{thm:hardcore-zero-count}, provided the smallness
assumptions there hold.

First,
\[
R=\frac{1}{\sqrt{\nu}}M_n^{-1/4-\beta}
\le \frac{c_{\ast}(K)}{\sqrt{\nu}}
\]
for all sufficiently large \(n\). Second,
\[
\nu R^2\Delta_n
=
M_n^{-1/2-2\beta}\Delta_n.
\]
Since \(\Delta_n\le |V_n|-1\), we have \(M_n=(1+\Delta_n)|V_n|\ge \Delta_n^2\), and hence
\[
\nu R^2\Delta_n
\le
M_n^{-1/2-2\beta}\,M_n^{1/2}
=
M_n^{-2\beta}
=o(1).
\]
Thus Theorem~\ref{thm:hardcore-zero-count} applies and gives
\[
\mathbb P\!\left[
Z_{G_n}(\lambda W)\text{ has a zero in }\mathbb D(0,r)
\right]
\le
\frac{C_{\ast}(K)\,M_n^{-4\beta}}{2\log(R/r)}.
\]
Since
\[
\log(R/r)=\log\bigl((1-\varepsilon/2)^{-1}\bigr)=\Theta_\varepsilon(1),
\]
we conclude that
\[
\mathbb P\!\left[
Z_{G_n}(\lambda W)\text{ has a zero in }\mathbb D(0,r)
\right]
=
O_{\varepsilon,K}(M_n^{-4\beta}).
\]
Because \((1-\varepsilon)R_n<r\) for all sufficiently large \(n\), the same bound
holds for \(\mathbb D(0,(1-\varepsilon)R_n)\).
\end{proof}

\begin{cor}[Algorithmic consequence]
\label{cor:hardcore-universality-algorithmic}
Fix \(\beta\in(0,1/4)\), \(\varepsilon\in(0,1)\), and \(\gamma>0\).
Let \(G_n=(V_n,E_n)\) be a sequence of finite graphs with maximum degree \(\Delta_n\),
and assume
\[
\mathbb P\!\left[
Z_{G_n}(\lambda W)\text{ has a zero in }\mathbb D\!\bigl(0,(1-\varepsilon/2)R_n\bigr)
\right]
=
O_{\varepsilon,K}\!\left(M_n^{-4\beta}\right),
\]
where
\[
M_n:=(1+\Delta_n)|V_n|,
\qquad
R_n:=\frac{1}{\sqrt{\nu}}\,M_n^{-1/4-\beta}.
\]
Then with probability
\[
1-O_{\varepsilon,K}\!\left(M_n^{-4\beta}\right),
\]
the following holds simultaneously for every \(\lambda\in\C\) such that
\[
|\lambda|
\le
(1-\varepsilon)\frac{1}{\sqrt{\nu}}\,M_n^{-1/4-2\beta}.
\]
There is a deterministic algorithm which outputs an \( |V_n|^{-\gamma}\)-additive
approximation to the branch of
\[
\log Z_{G_n}(\lambda W)
\]
determined by \(\log Z_{G_n}(0)=0\), in time
\[
\min\!\Bigl\{
\,|V_n|^{\,O(m)},
\;
|V_n|\,\exp\!\bigl(O_{\Delta_n}(m)\bigr)
\Bigr\},
\qquad
m=\Bigl\lceil \frac{\gamma+1}{\beta}\Bigr\rceil.
\]
Equivalently, since \(X_W(\lambda)\) differs from \(Z_{G_n}(\lambda W)\) by a
known exponential factor, the same holds for \(\log X_W(\lambda)\).
\end{cor}

\begin{proof}
On the event that \(Z_{G_n}(\lambda W)\) is zero-free on
\(\mathbb D(0,(1-\varepsilon/2)R_n)\), there is a holomorphic branch
\[
f(\lambda):=\log Z_{G_n}(\lambda W)=\sum_{k\ge 1} a_k \lambda^k
\]
on that disk, normalized by \(f(0)=0\).

Now fix \(\lambda\) with
\[
|\lambda|
\le
L_n:=(1-\varepsilon)\frac{1}{\sqrt{\nu}}\,M_n^{-1/4-2\beta}.
\]
Then
\[
\frac{|\lambda|}{(1-\varepsilon/2)R_n}
=
\frac{1-\varepsilon}{1-\varepsilon/2}\,M_n^{-\beta}
=
\rho_\varepsilon M_n^{-\beta},
\qquad
\rho_\varepsilon:=\frac{1-\varepsilon}{1-\varepsilon/2}<1.
\]
Let
\[
T_m(\lambda):=\sum_{k=1}^m a_k \lambda^k
\]
be the degree-\(m\) truncation of the Taylor series of \(f\).
By the standard Barvinok interpolation estimate,
\[
|f(\lambda)-T_m(\lambda)|
\le
\frac{|V_n|}{m+1}\,
\frac{(\rho_\varepsilon M_n^{-\beta})^{m+1}}
{1-\rho_\varepsilon M_n^{-\beta}}.
\]
Choose
\[
m=\left\lceil \frac{\gamma+1}{\beta}\right\rceil.
\]
Since \(M_n\ge |V_n|\), for all sufficiently large \(n\),
\[
\frac{|V_n|}{m+1}\,
\frac{(\rho_\varepsilon M_n^{-\beta})^{m+1}}
{1-\rho_\varepsilon M_n^{-\beta}}
\le
C_{\varepsilon,\beta,\gamma}\,
|V_n|\,M_n^{-\beta(m+1)}
\le
|V_n|^{-\gamma}.
\]
Hence
\[
|\log Z_{G_n}(\lambda W)-T_m(\lambda)|\le |V_n|^{-\gamma}.
\]

It remains to compute \(T_m\), equivalently the first \(m\) Taylor coefficients
of \(\log Z_{G_n}(\lambda W)\).

For the general graph \(G_n\), write
\[
Z_{G_n}(\lambda W)=\sum_{k=0}^{|V_n|} b_k \lambda^k,
\qquad
b_k=\sum_{\substack{I\in\mathcal I(G_n)\\ |I|=k}}\prod_{v\in I}W_v.
\]
Thus \(b_0,\dots,b_m\) can be computed exactly by enumerating all subsets of size
at most \(m\) and checking independence, which takes time \(|V_n|^{O(m)}\).
The coefficients \(a_1,\dots,a_m\) are then recovered recursively from
\[
k\,b_k=\sum_{j=1}^k j\,a_j\,b_{k-j},
\qquad k\ge 1,
\]
with \(b_0=1\).

On the other hand, for bounded-degree graphs one may instead use the standard
bounded-degree coefficient algorithm for independence polynomials, which computes
the first \(m\) coefficients in time
\[
|V_n|\,\exp\!\bigl(O_{\Delta_n}(m)\bigr);
\]
see \cite{barvinok2016computing,patelregts}.

Therefore the first \(m\) coefficients, and hence \(T_m(\lambda)\), can be
computed deterministically in time
\[
\min\!\Bigl\{
\,|V_n|^{\,O(m)},
\;
|V_n|\,\exp\!\bigl(O_{\Delta_n}(m)\bigr)
\Bigr\}.
\]
This gives the claimed approximation algorithm.
\end{proof}

\section{Stability, hardness, and anticoncentration}\label{sec:permanent-stability}
\subsection{Optimized permanent hardness reduction from \cite{eldar2018approximating}}
In Theorem 22 of \cite{eldar2018approximating}, the authors formally showed that being able to approximate the permanent of a random matrix with mean of order $o(1/n)$ implies a good approximation for the permanent of a random matrix with zero mean. By slightly optimizing their analysis, we can see that the same result holds if the mean is of order $o(1/\sqrt{n})$. For completeness, we formally state and prove the optimized version. 
\begin{thm}\label{thm:stable-bound}
Let $\Delta = \per(W + J\mu) - \per(W)$, where $W$ is a matrix with i.i.d. entries of mean zero and variance one, and suppose that
\[ |\mu| = \frac{c}{\sqrt{n}} \]
for some $c < 1$. Then
\[ \mathbb E |\Delta|^2 \le (n!) \frac{c^2}{1 - c^2} \]
so by Markov's inequality, for any $t > 0$
\[ \Pr\left(|\Delta|^2 > (n!) \frac{c^2 t}{1 - c^2}\right) \le 1/t. \]
\end{thm}
\begin{proof}
By direct computation (see Lemma 23 of \cite{eldar2018approximating}), one obtains
\[ \mathbb E |\Delta|^2 = (n!)^2 \sum_{k = 1}^{n} \frac{|\mu|^{2k}}{(n - k)!}, \]
and then we conclude
\[ \mathbb E |\Delta|^2 \le (n!)^2 \sum_{k = 1}^n \frac{|\mu|^{2k}}{(n - k)!} = (n!) \sum_{k = 1}^n |\mu|^{2k} n(n - 1) \cdots (n - k + 1) \le (n!) \sum_{k = 1}^n (|\mu|^2 n)^k \le (n!) \frac{c^2}{1 - c^2} \]
by summing the geometric series.
\end{proof}
Since $\mathbb E |\per(W)|^2 = n!$, if we can take $c = o(1)$ in this theorem we directly get an additive approximation to the permanent of $W$ which is $o(1)$ standard deviations away from the truth. See \cite{aaronson2011computational,eldar2018approximating} for discussion of further reductions under conjectured anticoncentration of the permanent.
\subsection{Anticoncentration conjecture and zeros}
This section serves as a kind of ``warm-up'' to the next one, where we prove a related unconditional result.

Here, we show that under the standard anti-concentration conjecture for the permanent \cite{aaronson2011computational}, only $O_{\epsilon}(\log n)$ many of the zeros of the permanent $\per(W + zJ)$ are within a ball of radius $(1 - \epsilon)\sqrt{1/n}$ around the origin. Combined with our previous (unconditional) analysis, this tells us that under the PACC the bulk of the zeros, i.e. $n - o(n)$ many of them, are at scale $\sqrt{1/n}$. Because the PACC is specific to the case where $W$ is an $n \times n$ matrix with $N(0,1)_{\mathbb C}$ entries, in this section we focus on the complex Gaussian case. 

\begin{conj}[Permanent Anti-Concentration Conjecture (PACC) \cite{aaronson2011computational}]
There exists a polynomial $p$ such that for all $n$ and $\delta > 0$,
\[ \Pr(|\per(W)| < \frac{\sqrt{n!}}{p(n,1/\delta)}) < \delta. \]
where $W$ is an $n \times n$ matrix with i.i.d. $N(0,1)_{\mathbb C}$ entries.
\end{conj}
It follows from the work of Tao and Vu \cite{tao2009permanent,Tao2010MO} that with probability at least $1 - o(1)$,
\[ \log |\per(W)| = (n/2)\log(n) \pm o(n \log(n)). \]
The difficulty in analyzing the scale of $\log |\per(W)|$ is the lower-tail bound rather than the upper-tail. For the upper tail, we have $\mathbb E |\per(W)|^2 = n!$, so by  Markov's inequality $\Pr(|per(W)|^2 \ge n!/\delta) \le \delta$). So, under the PACC we have the much more precise asymptotic
\begin{equation}\label{eqn:PACC-asymptotic}
\log |\per(W)| = (1/2)\log(n!) \pm O(\log(n/\delta)) = (1/2)(n\log n - n) \pm O(\log(n/\delta)) 
\end{equation}
with probability at least $1 - \delta$. Here the second equality is Stirling's approximation. If we assume this asymptotic holds, we can prove a strong anticoncentration result for the zeros of the permanent:
\begin{thm}
Suppose that the PACC conjecture holds. 
Let $n(r)$ be the number of zeros of the function $\per(W + Jz/\sqrt{n})$ with $|z| < r$. Then for any $1 > R > r$, it holds with probability $1 - \delta$ that
\[ n(r) \le  \frac{-\log(1 - R^2)}{2 \log(R/r)} + O(\log(n/\delta)/\log(R/r)). \]
In particular, by taking $R = (1 + r)/2$, this yields that $n(r) = O_{r}(\log(n/\delta))$ for any $r < 1$. 
\end{thm}
\begin{proof}
By Jensen's formula,  we have for any $R > r$ that
\[ n(r) \log(R/r) \le \frac{1}{2\pi} \int_{0}^{2\pi} \log |\per(W + JRe^{i\theta}/\sqrt{n})| d\theta - \log |\per(W)|. \]
By Jensen's inequality,
\begin{equation} n(r) \log(R/r) \le \log \frac{1}{2\pi} \int_{0}^{2\pi} |\per(W + JRe^{i\theta}/\sqrt{n})| d\theta - \log |\per(W)|
\end{equation}
By Markov's inequality,
\[ \Pr\left( \frac{1}{2\pi} \int_{0}^{2\pi} |\per(W + JRe^{i\theta}/\sqrt{n})| d\theta  > \frac{t}{2\pi} \int_{0}^{2\pi} \mathbb{E} |\per(W + JRe^{i\theta}/\sqrt{n})| d\theta\right)\le  1/t \]
where the randomness is over the matrix $W$. Letting $\Delta_{\theta} =  \per(W + JRe^{i\theta}/\sqrt{n}) - \per(W)$, we can compute that under the requirement $1 > R > r$,
\begin{align}
\frac{1}{2\pi} \int_{0}^{2\pi} \mathbb{E} |\per(W) + \Delta_{\theta}| d\theta  
&\le \sqrt{\frac{1}{2\pi} \int_{0}^{2\pi} \mathbb{E} |\per(W) + \Delta_{\theta}|^2 d\theta} \\
&= \sqrt{\frac{1}{2\pi} \int_{0}^{2\pi} [\mathbb{E} |\per(W)|^2 + \mathbb E |\Delta_{\theta}|^2] d\theta} \\
&\le \sqrt{n![1 + R^2/(1 - R^2)]} = \sqrt{n!/(1 - R^2)} \label{eqn:per-delta-estimate}
\end{align}
where we used that $\frac{1}{2\pi} \int \Delta_{\theta} = 0$ to show the crossterm is zero, and then we used that $\mathbb{E} |\per(W)|^2 = n!$ and Theorem~\ref{thm:stable-bound}. So setting $t = 2/\delta$, we have with probability at least $1 - \delta/2$ that
\[ \frac{1}{2\pi} \int_{0}^{2\pi} |\per(W + JRe^{i\theta}/\sqrt{n})| d\theta \le \frac{2}{\delta} \sqrt{n!/(1 - R^2)}. \]
Under this event, we therefore have
\[  n(r) \log(R/r) \le \log(2/\delta) + \frac{\log(n!) - \log(1 - R^2)}{2} - \log |\per(W)|. \]
Now assuming the PACC, we can use the asymptotic \eqref{eqn:PACC-asymptotic} to get
\[ n(r) \log(R/r) \le \frac{-\log(1 - R^2)}{2} + O(\log(n/\delta)) \]
as desired. 
\end{proof}
Recall that by the fundamental theorem of algebra, the polynomial $\per(W + Jz/\sqrt{n})$ must have $n$ zeros (counted with multiplicity). 
\begin{cor}
Under the PACC, with probability at least $1 - \delta$ the bulk of the zeros (i.e., $1 - o(1)$ portion of the $n$ total zeros) of  $\per(W + Jz/\sqrt{n})$ lie in the annulus
\[ \{ z : |z| \in (1 - o(1), O_{\delta}(1)) \} \]
\end{cor}
\begin{proof}
Combine the previous theorem with Markov's inequality applied to the basic estimate discussed in Section~\ref{sec:overview} (which follows from our results, but for which the results of \cite{eldar2018approximating} are also sufficient).
\end{proof}
\subsection{Unconditional anticoncentration for zeros}

We can prove an unconditional variant of the anticoncentration of zeros based on the following inductive argument.

\begin{lem}\label{lem:gaussian-per-log-lower}
Let \(W_n\) be an \(n\times n\) matrix with i.i.d.\ complex Gaussian entries \(W_{ij}\sim \mathcal{CN}(0,1)\), and define
\[
Z_n:=\per(W_n).
\]
Then for every \(n\ge 1\),
\[
\mathbb E \log |Z_n|^2 \ge \log(n!) - \gamma n,
\]
where \(\gamma \approx 0.577\) denotes Euler's constant.
\end{lem}

\begin{proof}
We argue by recursion. Expanding the permanent along the last row,
\[
Z_n=\sum_{j=1}^n W_{nj} C_j,
\]
where \(C_j\) is the permanent of the \((n-1)\times(n-1)\) minor obtained by deleting row \(n\) and column \(j\).

Conditional on the first \(n-1\) rows, the cofactors \(C_1,\dots,C_n\) are deterministic complex numbers, while the last-row entries \(W_{n1},\dots,W_{nn}\) remain independent \(\mathcal{CN}(0,1)\) random variables. Since
\[
Z_n=\sum_{j=1}^n W_{nj}C_j,
\]
it follows that \(Z_n\mid \mathcal F_{n-1}\) is a complex Gaussian random variable. Indeed, a linear combination of independent centered complex Gaussians is again centered complex Gaussian, and
\[
\mathbb E\!\left[ Z_n \mid \mathcal F_{n-1}\right]=0,
\]
while
\[
\mathbb E\!\left[ |Z_n|^2 \mid \mathcal F_{n-1}\right]
=
\sum_{j,k=1}^n C_j\overline{C_k}\,
\mathbb E\!\left[W_{nj}\overline{W_{nk}}\right]
=
\sum_{j=1}^n |C_j|^2
=:S_n,
\]
because \(\mathbb E[W_{nj}\overline{W_{nk}}]=\delta_{jk}\). Thus
\[
Z_n \,\big|\, \mathcal F_{n-1} \sim \mathcal{CN}(0,S_n).
\]

Now if \(G\sim \mathcal{CN}(0,1)\), then \(Z_n\mid \mathcal F_{n-1}\) has the same law as \(\sqrt{S_n}\,G\). Therefore
\[
\log |Z_n|^2 = \log S_n + \log |G|^2
\qquad\text{in distribution conditional on }\mathcal F_{n-1}.
\]
Taking conditional expectation gives
\[
\mathbb E\!\left[\log |Z_n|^2 \mid \mathcal F_{n-1}\right]
=
\log S_n + \mathbb E\log |G|^2.
\]
Finally, since \(G\sim \mathcal{CN}(0,1)\), the random variable \(|G|^2\) is exponential with mean \(1\), so
\[
\mathbb E\log |G|^2
=
\int_0^\infty (\log x)e^{-x}\,dx
=
-\gamma.
\]
Hence
\[
\mathbb E\!\left[\log |Z_n|^2 \mid \mathcal F_{n-1}\right]
=
\log S_n-\gamma.
\]
Taking expectation gives
\[
\mathbb E \log |Z_n|^2 = -\gamma + \mathbb E \log S_n.
\]
Now by concavity of \(\log\),
\[
\log S_n
=
\log\!\left(\sum_{j=1}^n |C_j|^2\right)
\ge
\frac1n\sum_{j=1}^n \log\bigl(n|C_j|^2\bigr)
=
\log n + \frac1n\sum_{j=1}^n \log |C_j|^2.
\]
Taking expectation and using that each \(C_j\) has the same distribution as \(\per(W_{n-1})\), we obtain
\[
\mathbb E \log |Z_n|^2
\ge
-\gamma + \log n + \mathbb E \log |Z_{n-1}|^2.
\]
Iterating from \(1\) to \(n\), and using \(Z_1=W_{11}\) with
\[
\mathbb E\log |Z_1|^2 = \mathbb E\log |W_{11}|^2 = -\gamma,
\]
yields
\[
\mathbb E \log |Z_n|^2
\ge
\sum_{k=1}^n (\log k - \gamma)
=
\log(n!) - \gamma n.
\]
\end{proof}

\begin{thm}\label{thm:few-small-zeros-clean}
Let
\[
f_n(z):=\per\!\left(W+\frac{z}{\sqrt n}J\right),
\]
where \(W\) is an \(n\times n\) matrix with i.i.d.\ complex Gaussian entries \(W_{ij}\sim \mathcal{CN}(0,1)\), and \(J\) is the all-ones matrix. Let \(n(r)\) denote the number of zeros of \(f_n\) in the disk \(\{z:|z|<r\}\), counted with multiplicity.

Then for every \(\varepsilon,\eta\in(0,1)\), there exists \(r_0=r_0(\varepsilon,\eta)>0\) such that for every \(r\in(0,r_0)\), there exists \(N=N(\varepsilon,\eta,r)\) with the following property: for all \(n\ge N\),
\[
\Pr\!\bigl(n(r)\le \varepsilon n\bigr)\ge 1-\eta.
\]
\end{thm}

\begin{proof}
Fix \(\varepsilon,\eta\in(0,1)\), and set \(R:=\frac12\). Let \(0<r<R\). 
Recall from Lemma~\ref{lem:core-jensen-app} that
\begin{equation}\label{eq:expected-zero-bound-clean}
\mathbb E\,n(r)
\le
\frac{
\frac{1}{4\pi}\int_0^{2\pi} \mathbb E|f_n(Re^{i\theta})|^2\,d\theta
-\mathbb E\log|f_n(0)|
}{
\log(R/r)
}.
\end{equation}
By the previously established boundary second-moment estimate \eqref{eqn:per-delta-estimate},
\[
\frac{1}{2\pi}\int_0^{2\pi}\mathbb E|f_n(Re^{i\theta})|^2\,d\theta
\le
\frac{n!}{1-R^2},
\]
so using that \(R=\frac12\), 
we find that
\begin{equation}\label{eq:boundary-upper-clean}
\mathbb E\!\left[\frac{1}{2\pi}\int_0^{2\pi}\log|f_n(Re^{i\theta})|\,d\theta\right]
\le
\frac12\log(n!) + \frac12\log\frac43.
\end{equation}

For the next term, since \(f_n(0)=\per(W)\), Lemma~\ref{lem:gaussian-per-log-lower} implies
$\mathbb E\log|f_n(0)|^2
\ge
\log(n!) - \gamma n$.
Therefore
\begin{equation}\label{eq:center-lower-clean}
\mathbb E\log|f_n(0)|
\ge
\frac12\log(n!) - \frac{\gamma}{2}n.
\end{equation}

Substituting \eqref{eq:boundary-upper-clean} and \eqref{eq:center-lower-clean} into \eqref{eq:expected-zero-bound-clean}, we obtain
\[
\mathbb E\,n(r)
\le
\frac{
\frac{\gamma}{2}n+\frac12\log\frac43
}{
\log(1/(2r))
}.
\]
Choose \(r_0=r_0(\varepsilon,\eta)\in(0,\frac12)\) so small that
\[
\log\frac{1}{2r_0}>\frac{\gamma}{2\varepsilon\eta}.
\]
Then for every \(r\in(0,r_0)\),
\[
\frac{\gamma/2}{\log(1/(2r))}<\varepsilon\eta.
\]
Since \(\frac12\log\frac43\) is a constant, for each such \(r\) there exists \(N=N(\varepsilon,\eta,r)\) such that for all \(n\ge N\),
\[
\mathbb E\,n(r)\le \varepsilon\eta\, n.
\]
The conclusion follows from Markov's inequality.
\end{proof}


\paragraph{Acknowledgements.} Some simulations were performed on the University of Chicago’s Data Science Institute cluster.

\bibliographystyle{plain} 
\bibliography{bib}

@article{anand2025simulating,
  title={Simulating Gaussian boson sampling on graphs in polynomial time},
  author={Anand, Konrad and Chen, Zongchen and Cryan, Mary and Freifeld, Graham and Goldberg, Leslie Ann and Guo, Heng and Zhang, Xinyuan},
  journal={arXiv preprint arXiv:2511.16558},
  year={2025}
}

@inproceedings{eldar2018approximating,
  title={Approximating the permanent of a random matrix with vanishing mean},
  author={Eldar, Lior and Mehraban, Saeed},
  booktitle={2018 IEEE 59th Annual Symposium on Foundations of Computer Science (FOCS)},
  pages={23--34},
  year={2018},
  organization={IEEE}
}

@inproceedings{bouland2025exponential,
  title={Exponential improvements to the average-case hardness of BosonSampling},
  author={Bouland, Adam and Datta, Ishaun and Fefferman, Bill and Hern{\'a}ndez, Felipe},
  booktitle={2025 IEEE 66th Annual Symposium on Foundations of Computer Science (FOCS)},
  pages={912--933},
  year={2025},
  organization={IEEE}
}

@article{bencs2025zeros,
  title={On zeros and algorithms for disordered systems: mean-field spin glasses},
  author={Bencs, Ferenc and Huang, Brice and Lee, Daniel Z and Liu, Kuikui and Regts, Guus},
  journal={arXiv preprint arXiv:2507.15616},
  year={2025}
}

@article{mohanty2025eigenvalue,
  title={Eigenvalue Bounds for Random Matrices via Zerofreeness},
  author={Mohanty, Sidhanth and Rajaraman, Amit},
  journal={arXiv preprint arXiv:2509.25471},
  year={2025}
}

@inproceedings{ebrahimnejad2025approximability,
  title={On approximability of the Permanent of PSD matrices},
  author={Ebrahimnejad, Farzam and Nagda, Ansh and Gharan, Shayan Oveis},
  booktitle={Proceedings of the 57th Annual ACM Symposium on Theory of Computing},
  pages={625--630},
  year={2025}
}

@inproceedings{harrow2020classical,
  title={Classical algorithms, correlation decay, and complex zeros of partition functions of quantum many-body systems},
  author={Harrow, Aram W and Mehraban, Saeed and Soleimanifar, Mehdi},
  booktitle={Proceedings of the 52nd Annual ACM SIGACT Symposium on Theory of Computing},
  pages={378--386},
  year={2020}
}

@article{chertkov2013approximating,
  title={Approximating the permanent with fractional belief propagation},
  author={Chertkov, Michael and Yedidia, Adam B},
  journal={The Journal of Machine Learning Research},
  volume={14},
  number={1},
  pages={2029--2066},
  year={2013},
  publisher={JMLR. org}
}

@article{anari2025tight,
  title={A tight analysis of Bethe approximation for permanent},
  author={Anari, Nima and Rezaei, Alireza},
  journal={SIAM Journal on Computing},
  volume={54},
  number={4},
  pages={FOCS19-81--FOCS19-101},
  year={2025},
  publisher={SIAM}
}

@article{barvinok2019computing,
  title={Computing permanents of complex diagonally dominant matrices and tensors},
  author={Barvinok, Alexander},
  journal={Israel Journal of Mathematics},
  volume={232},
  number={2},
  pages={931--945},
  year={2019},
  publisher={Springer}
}

@article{tao2009permanent,
  title={On the permanent of random Bernoulli matrices},
  author={Tao, Terence and Vu, Van},
  journal={Advances in Mathematics},
  volume={220},
  number={3},
  pages={657--669},
  year={2009},
  publisher={Elsevier}
}

@article{kwan2022permanent,
  title={On the permanent of a random symmetric matrix},
  author={Kwan, Matthew and Sauermann, Lisa},
  journal={Selecta Mathematica},
  volume={28},
  number={1},
  pages={15},
  year={2022},
  publisher={Springer}
}

@article{hunter2025exponential,
  title={Exponential anticoncentration of the permanent},
  author={Hunter, Zach and Kwan, Matthew and Sauermann, Lisa},
  journal={arXiv preprint arXiv:2509.22577},
  year={2025}
}

@inproceedings{ji2021approximating,
  title={Approximating permanent of random matrices with vanishing mean: made better and simpler},
  author={Ji, Zhengfeng and Jin, Zhihan and Lu, Pinyan},
  booktitle={Proceedings of the 2021 ACM-SIAM Symposium on Discrete Algorithms (SODA)},
  pages={959--975},
  year={2021},
  organization={SIAM}
}

@inproceedings{aaronson2011computational,
  title={The computational complexity of linear optics},
  author={Aaronson, Scott and Arkhipov, Alex},
  booktitle={Proceedings of the forty-third annual ACM symposium on Theory of computing},
  pages={333--342},
  year={2011}
}

@inproceedings{clifford2018classical,
  title={The classical complexity of boson sampling},
  author={Clifford, Peter and Clifford, Rapha{\"e}l},
  booktitle={Proceedings of the Twenty-Ninth Annual ACM-SIAM Symposium on Discrete Algorithms},
  pages={146--155},
  year={2018},
  organization={SIAM}
}

@article{jerrum2004polynomial,
  title={A polynomial-time approximation algorithm for the permanent of a matrix with nonnegative entries},
  author={Jerrum, Mark and Sinclair, Alistair and Vigoda, Eric},
  journal={Journal of the ACM (JACM)},
  volume={51},
  number={4},
  pages={671--697},
  year={2004},
  publisher={ACM New York, NY, USA}
}

@article{jerrum1989approximating,
  title={Approximating the permanent},
  author={Jerrum, Mark and Sinclair, Alistair},
  journal={SIAM journal on computing},
  volume={18},
  number={6},
  pages={1149--1178},
  year={1989},
  publisher={SIAM}
}

@inproceedings{gurvits2014bounds,
  title={Bounds on the permanent and some applications},
  author={Gurvits, Leonid and Samorodnitsky, Alex},
  booktitle={2014 IEEE 55th Annual Symposium on Foundations of Computer Science},
  pages={90--99},
  year={2014},
  organization={IEEE}
}

@inproceedings{linial1998deterministic,
  title={A deterministic strongly polynomial algorithm for matrix scaling and approximate permanents},
  author={Linial, Nathan and Samorodnitsky, Alex and Wigderson, Avi},
  booktitle={Proceedings of the thirtieth annual ACM symposium on Theory of computing},
  pages={644--652},
  year={1998}
}

@article{zhou2026complex,
  title={Complex-Valued-Matrix Permanents: SPA-based Approximations and Double-Cover Analysis},
  author={Zhou, Junda and Vontobel, Pascal O},
  journal={arXiv preprint arXiv:2601.18232},
  year={2026}
}

@article{lee1952statistical,
  title={Statistical theory of equations of state and phase transitions. II. Lattice gas and Ising model},
  author={Lee, Tsung-Dao and Yang, Chen-Ning},
  journal={Physical Review},
  volume={87},
  number={3},
  pages={410},
  year={1952},
  publisher={APS}
}

@misc{perkins2023five,
  title={Five lectures on statistical physics methods in combinatorics},
  author={Perkins, Will},
  year={2023}
}

@book{ryser1963combinatorial,
  title={Combinatorial mathematics},
  author={Ryser, Herbert John},
  volume={14},
  year={1963},
  publisher={American Mathematical Soc.}
}

@article{barvinok2016computing,
  title={Computing the permanent of (some) complex matrices},
  author={Barvinok, Alexander},
  journal={Foundations of Computational Mathematics},
  volume={16},
  number={2},
  pages={329--342},
  year={2016},
  publisher={Springer}
}

@article{brenner1959relations,
  title={Relations among the minors of a matrix with dominant principal diagonal},
  author={Brenner, JL},
  journal={Duke Math. J.},
  volume={26},
  number={4},
  pages={563--567},
  year={1959}
}

@misc{Tao2010MO,
    author       = {Terence Tao},
    title        = {Answer to ``Anti-concentration bound for permanents of {G}aussian matrices?''},
    howpublished = {Math Overflow},
    year         = {2010},
    note         = {URL: \url{https://mathoverflow.net/questions/45822/anti-concentration-bound-for-permanents-of-gaussian-matrices}},
}

@article{aizenman1987some,
  title={Some rigorous results on the Sherrington-Kirkpatrick spin glass model},
  author={Aizenman, Michael and Lebowitz, Joel L and Ruelle, David},
  journal={Communications in mathematical physics},
  volume={112},
  number={1},
  pages={3--20},
  year={1987},
  publisher={Springer}
}

@article{hughes2008zeros,
  title={The zeros of random polynomials cluster uniformly near the unit circle},
  author={Hughes, Christopher P and Nikeghbali, Ashkan},
  journal={Compositio Mathematica},
  volume={144},
  number={3},
  pages={734--746},
  year={2008},
  publisher={London Mathematical Society}
}

@book{stein2010complex,
  title={Complex analysis},
  author={Stein, Elias M and Shakarchi, Rami},
  volume={2},
  year={2010},
  publisher={Princeton University Press}
}

@article{rempala1999limiting,
  title={Limiting behavior of random permanents},
  author={Rempa{\l}a, Grzegorz and Weso{\l}owski, Jacek},
  journal={Statistics \& probability letters},
  volume={45},
  number={2},
  pages={149--158},
  year={1999},
  publisher={Elsevier}
}

@article{heilmann1972theory,
  title={Theory of monomer-dimer systems},
  author={Heilmann, Ole J and Lieb, Elliott H},
  journal={Communications in mathematical Physics},
  volume={25},
  number={3},
  pages={190--232},
  year={1972},
  publisher={Springer}
}

@article{scott2005repulsive,
  title={The repulsive lattice gas, the independent-set polynomial, and the Lov{\'a}sz local lemma},
  author={Scott, Alexander D and Sokal, Alan D},
  journal={Journal of Statistical Physics},
  volume={118},
  number={5},
  pages={1151--1261},
  year={2005},
  publisher={Springer}
}

@article{bissacot2011improvement,
  title={An improvement of the Lov{\'a}sz local lemma via cluster expansion},
  author={Bissacot, Rodrigo and Fern{\'a}ndez, Roberto and Procacci, Aldo and Scoppola, Benedetto},
  journal={Combinatorics, Probability and Computing},
  volume={20},
  number={5},
  pages={709--719},
  year={2011},
  publisher={Cambridge University Press}
}

@article{chudnovsky2007roots,
  title={The roots of the independence polynomial of a clawfree graph},
  author={Chudnovsky, Maria and Seymour, Paul},
  journal={Journal of Combinatorial Theory, Series B},
  volume={97},
  number={3},
  pages={350--357},
  year={2007},
  publisher={Elsevier}
}

@article{shearer1985problem,
  title={On a problem of Spencer},
  author={Shearer, James B.},
  journal={Combinatorica},
  volume={5},
  number={3},
  pages={241--245},
  year={1985},
  publisher={Springer}
}

@article{patelregts,
author = {Patel, Viresh and Regts, Guus},
title = {Deterministic Polynomial-Time Approximation Algorithms for Partition Functions and Graph Polynomials},
journal = {SIAM Journal on Computing},
volume = {46},
number = {6},
pages = {1893-1919},
year = {2017},
doi = {10.1137/16M1101003},

URL = { 
    
        https://doi.org/10.1137/16M1101003
    
    

},
eprint = { 
    
        https://doi.org/10.1137/16M1101003
    
    

}
,}

@article{jerrum1996markov,
  title={The Markov chain Monte Carlo method: an approach to approximate counting and integration},
  author={Jerrum, Mark and Sinclair, Alistair},
  journal={Approximation Algorithms for NP-hard problems, PWS Publishing},
  year={1996}
}

@article{peters2019conjecture,
  title={On a conjecture of Sokal concerning roots of the independence polynomial},
  author={Peters, Han and Regts, Guus},
  journal={Michigan Mathematical Journal},
  volume={68},
  number={1},
  pages={33--55},
  year={2019},
  publisher={University of Michigan, Department of Mathematics}
}

@article{dyson1952divergence,
  title={Divergence of perturbation theory in quantum electrodynamics},
  author={Dyson, Freeman J},
  journal={Physical Review},
  volume={85},
  number={4},
  pages={631},
  year={1952},
  publisher={APS}
}

@book{friedli2017statistical,
  title={Statistical mechanics of lattice systems: a concrete mathematical introduction},
  author={Friedli, Sacha and Velenik, Yvan},
  year={2017},
  publisher={Cambridge University Press}
}

@article{eldan2022spectral,
  title={A spectral condition for spectral gap: fast mixing in high-temperature Ising models},
  author={Eldan, Ronen and Koehler, Frederic and Zeitouni, Ofer},
  journal={Probability theory and related fields},
  volume={182},
  number={3},
  pages={1035--1051},
  year={2022},
  publisher={Springer}
}

@inproceedings{anari2024trickle,
  title={Trickle-down in localization schemes and applications},
  author={Anari, Nima and Koehler, Frederic and Vuong, Thuy-Duong},
  booktitle={Proceedings of the 56th Annual ACM Symposium on Theory of Computing},
  pages={1094--1105},
  year={2024}
}

@inproceedings{anari2022entropic,
  title={Entropic independence: optimal mixing of down-up random walks},
  author={Anari, Nima and Jain, Vishesh and Koehler, Frederic and Pham, Huy Tuan and Vuong, Thuy-Duong},
  booktitle={Proceedings of the 54th Annual ACM SIGACT Symposium on Theory of Computing},
  pages={1418--1430},
  year={2022}
}

@article{adhikari2024spectral,
  title={Spectral gap estimates for mixed p-spin models at high temperature},
  author={Adhikari, Arka and Brennecke, Christian and Xu, Changji and Yau, Horng-Tzer},
  journal={Probability Theory and Related Fields},
  volume={189},
  number={3},
  pages={879--907},
  year={2024},
  publisher={Springer}
}

@inproceedings{anari2024universality,
  title={Universality of spectral independence with applications to fast mixing in spin glasses},
  author={Anari, Nima and Jain, Vishesh and Koehler, Frederic and Pham, Huy Tuan and Vuong, Thuy-Duong},
  booktitle={Proceedings of the 2024 Annual ACM-SIAM Symposium on Discrete Algorithms (SODA)},
  pages={5029--5056},
  year={2024},
  organization={SIAM}
}

@article{chen2025rapid,
  title={Rapid mixing on random regular graphs beyond uniqueness},
  author={Chen, Xiaoyu and Chen, Zejia and Chen, Zongchen and Yin, Yitong and Zhang, Xinyuan},
  journal={arXiv preprint arXiv:2504.03406},
  year={2025}
}

@inproceedings{sly2010computational,
  title={Computational transition at the uniqueness threshold},
  author={Sly, Allan},
  booktitle={2010 IEEE 51st Annual Symposium on Foundations of Computer Science},
  pages={287--296},
  year={2010},
  organization={IEEE}
}

@inproceedings{buys2022lee,
  title={Lee--Yang zeros and the complexity of the ferromagnetic Ising model on bounded-degree graphs},
  author={Buys, Pjotr and Galanis, Andreas and Patel, Viresh and Regts, Guus},
  booktitle={Forum of Mathematics, Sigma},
  volume={10},
  pages={e7},
  year={2022},
  organization={Cambridge University Press}
}

@inproceedings{de2024zeros,
  title={Zeros, chaotic ratios and the computational complexity of approximating the independence polynomial},
  author={De Boer, David and Buys, Pjotr and Guerini, Lorenzo and Peters, Han and Regts, Guus},
  booktitle={Mathematical Proceedings of the Cambridge Philosophical Society},
  volume={176},
  number={2},
  pages={459--494},
  year={2024},
  organization={Cambridge University Press}
}

@article{galanis2022complexity,
  title={The complexity of approximating the complex-valued Ising model on bounded degree graphs},
  author={Galanis, Andreas and Goldberg, Leslie A and Herrera-Poyatos, Andr{\'e}s},
  journal={SIAM Journal on Discrete Mathematics},
  volume={36},
  number={3},
  pages={2159--2204},
  year={2022},
  publisher={SIAM}
}

@article{bezakova2021complexity,
  title={The complexity of approximating the matching polynomial in the complex plane},
  author={Bez{\'a}kov{\'a}, Ivona and Galanis, Andreas and Goldberg, Leslie Ann and {\v{S}}tefankovi{\v{c}}, Daniel},
  journal={ACM Transactions on Computation Theory (TOCT)},
  volume={13},
  number={2},
  pages={1--37},
  year={2021},
  publisher={ACM New York, NY, USA}
}

@article{bencs2025complex,
  title={On the complex zeros and the computational complexity of approximating the reliability polynomial},
  author={Bencs, Ferenc and Piombi, Chiara and Regts, Guus},
  journal={arXiv preprint arXiv:2512.11504},
  year={2025}
}

@article{valiant1979complexity,
  title={The complexity of computing the permanent},
  author={Valiant, Leslie G},
  journal={Theoretical computer science},
  volume={8},
  number={2},
  pages={189--201},
  year={1979},
  publisher={Elsevier}
}

@inproceedings{kunisky2024optimality,
  title={Optimality of Glauber dynamics for general-purpose Ising model sampling and free energy approximation},
  author={Kunisky, Dmitriy},
  booktitle={Proceedings of the 2024 Annual ACM-SIAM Symposium on Discrete Algorithms (SODA)},
  pages={5013--5028},
  year={2024},
  organization={SIAM}
}

@article{gheissari2019spectral,
  title={On the spectral gap of spherical spin glass dynamics},
  author={Gheissari, Reza and Jagannath, Aukosh},
  journal={Annales de l’Institut Henri Poincar{\'e}-Probabilit{\'e}s et Statistiques},
  volume={55},
  number={2},
  pages={756--776},
  year={2019}
}

@article{stanley2011enumerative,
  title={Enumerative combinatorics},
  volume={1},
  author={Stanley, Richard P},
  journal={Cambridge studies in advanced mathematics},
  year={2011}
}

@book{talagrand_qft,
place={Cambridge}, title={What Is a Quantum Field Theory?}, publisher={Cambridge University Press}, author={Talagrand, Michel}, year={2022}}

@book{stanley2015catalan,
  title={Catalan numbers},
  author={Stanley, Richard P},
  year={2015},
  publisher={Cambridge University Press}
}

@article{liu2019ising,
  title={The Ising Partition Function: Zeros and Deterministic Approximation},
  author={Liu, Jingcheng and Sinclair, Alistair and Srivastava, Piyush},
  journal={Journal of Statistical Physics},
  volume={174},
  number={2},
  pages={287--315},
  year={2019},
  publisher={Springer}
}

@article{thouless1977solution,
  title={Solution of 'solvable model of a spin glass'},
  author={Thouless, David J and Anderson, Philip W and Palmer, Robert G},
  journal={Philosophical Magazine},
  volume={35},
  number={3},
  pages={593--601},
  year={1977},
  publisher={Taylor \& Francis}
}

@book{lando2004graphs,
  title={Graphs on surfaces and their applications},
  author={Lando, Sergei K and Zvonkin, Alexander K},
  volume={141},
  year={2004},
  publisher={Springer}
}

@article{hubbard1959calculation,
  title={Calculation of partition functions},
  author={Hubbard, John},
  journal={Physical Review Letters},
  volume={3},
  number={2},
  pages={77},
  year={1959},
  publisher={APS}
}

@techreport{van2014probability,
  title={Probability in high dimension},
  author={Van Handel, Ramon},
  year={2014}
}

@article{pfister1991large,
  title={Large deviations and phase separation in the two-dimensional Ising model},
  author={Pfister, Charles-Edouard},
  journal={Helvetica Physica Acta},
  volume={64},
  number={7},
  pages={953--1054},
  year={1991},
  publisher={Birkh{\"a}user}
}

@article{penrose1963convergence,
  title={Convergence of fugacity expansions for fluids and lattice gases},
  author={Penrose, Oliver},
  journal={Journal of Mathematical Physics},
  volume={4},
  number={10},
  pages={1312--1320},
  year={1963},
  publisher={American Institute of Physics}
}

\appendix
\section{Numerical simulations}\label{apdx:simulations}
\paragraph{Zeros.} In Figure~\ref{fig:ryser_gaussian}, we located the zeros of the (rescaled) random polynomial
\begin{equation}\label{eqn:experiment-eqn}
    z \mapsto \operatorname{per}(z J/\sqrt{n} + W)
\end{equation}
 by computing the magnitude of the permanent via Ryser's exact formula \cite{ryser1963combinatorial}, for small values of $n$. Here $J$ is the all-ones matrix and $W$ is an $n \times n$ random matrix with standard complex Gaussian entries. For each size $n$, we ran 3 independent simulations which correspond to the rows of the figure. All computations were executed at 64-bit floating point precision.  
 
 As we can see from the figure, in the experiments all of the zeros of \eqref{eqn:experiment-eqn} had real and imaginary parts of magnitude at most $3$. Recall from the discussion in Section~\ref{sec:overview}  that asymptotically as $n \to \infty$, all but $o(1)$ fraction of the zeros of \eqref{eqn:experiment-eqn} must be magnitude $O(1)$, which is the motivation for the $1/\sqrt{n}$ scaling factor.

\begin{figure}
    \centering
    \includegraphics[width=0.95\linewidth]{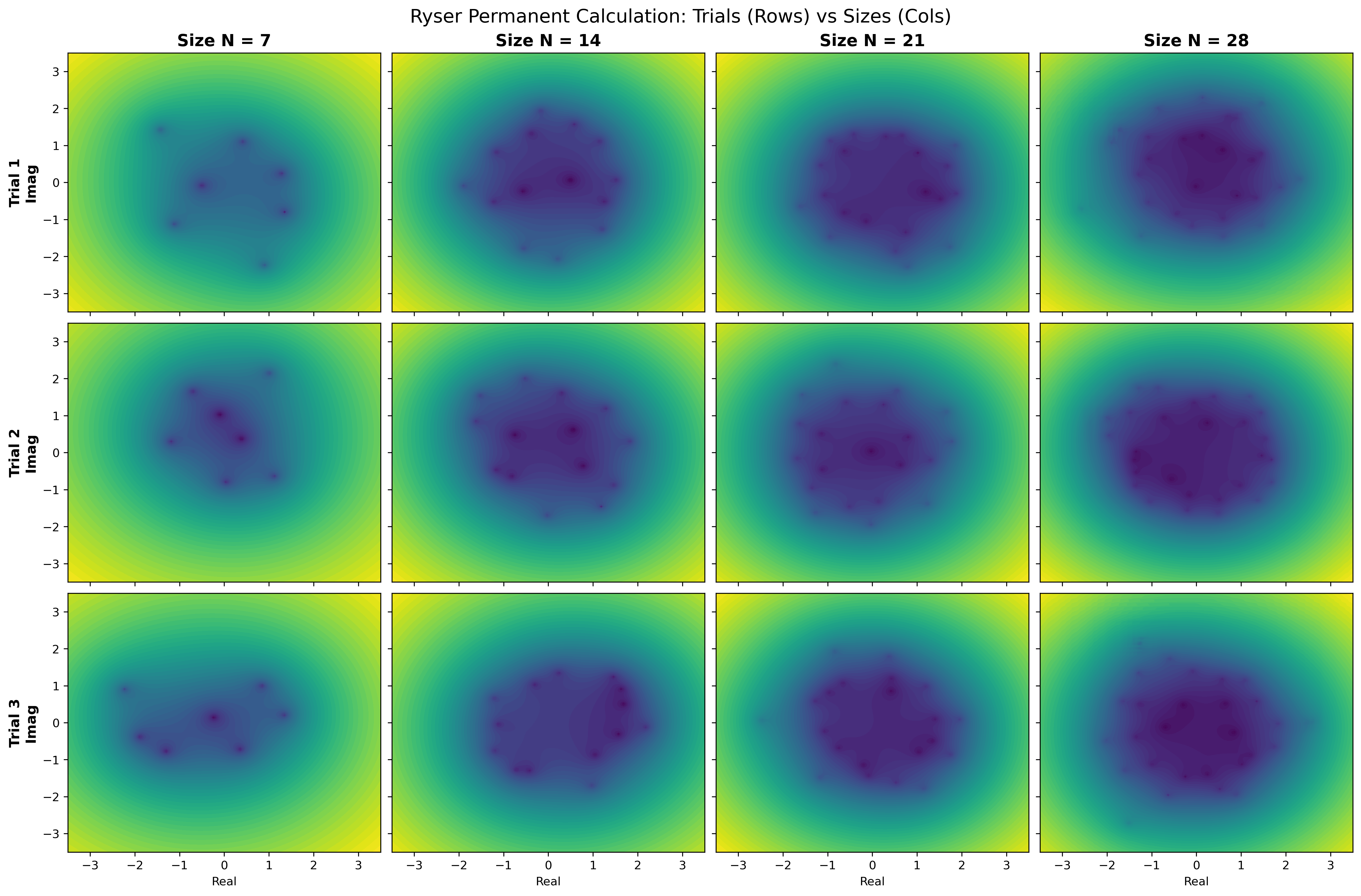}
    \caption{Complex magnitude of the $n \times n$ permanent \eqref{eqn:experiment-eqn} with random complex Gaussian entries. Colors correspond to the log-magnitude of the permanent at a particular point in $\mathbb C$, with dark blue corresponding to smaller permanents (more negative log-permanent) and bright yellow to larger values (more positive log-permanent). In each case, the locations of the zeros can be visually seen as dark-blue dots in the plot.}
    \label{fig:ryser_gaussian}
\end{figure}

We also performed the same experiment with complex Laplace (symmetrical exponential) distributed entries, instead of Gaussian. The results are similar to the Gaussian case. 
In Figure~\ref{fig:ryser_real}, we show the results with real i.i.d. Gaussian entries, which exhibit reflection symmetry across the real-axis but otherwise broadly appear similar. 
\begin{figure}
    \centering
    \includegraphics[width=0.95\linewidth]{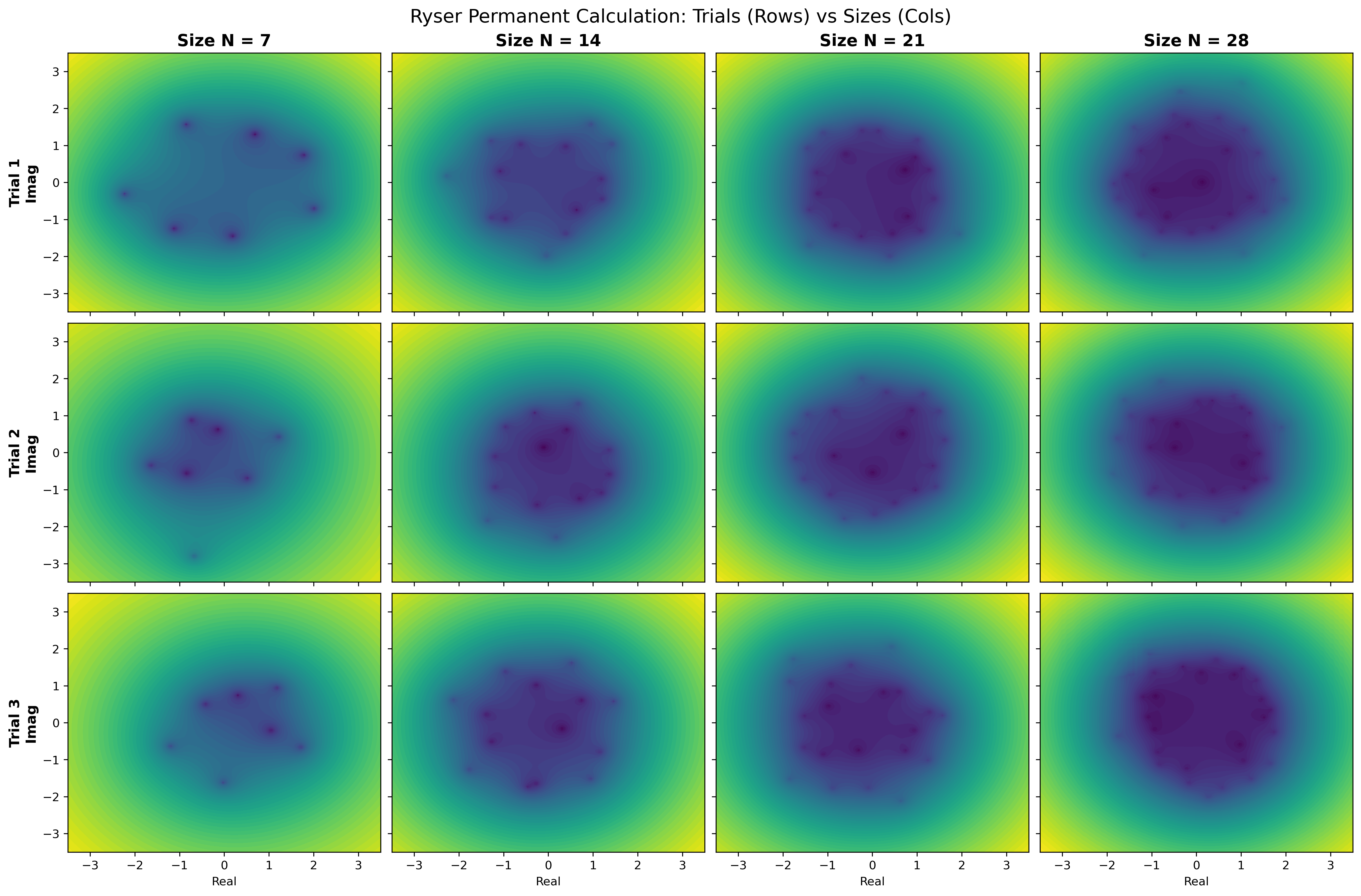}
    \caption{Complex magnitude of the $n \times n$ permanent \eqref{eqn:experiment-eqn} with random complex Laplace entries. As before, the zeros are visible as dark blue dots. The results appear qualitatively similar to the Gaussian case (Figure~\ref{fig:ryser_gaussian}).}
    \label{fig:ryser_laplace}
\end{figure}
\begin{figure}
    \centering
    \includegraphics[width=0.95\linewidth]{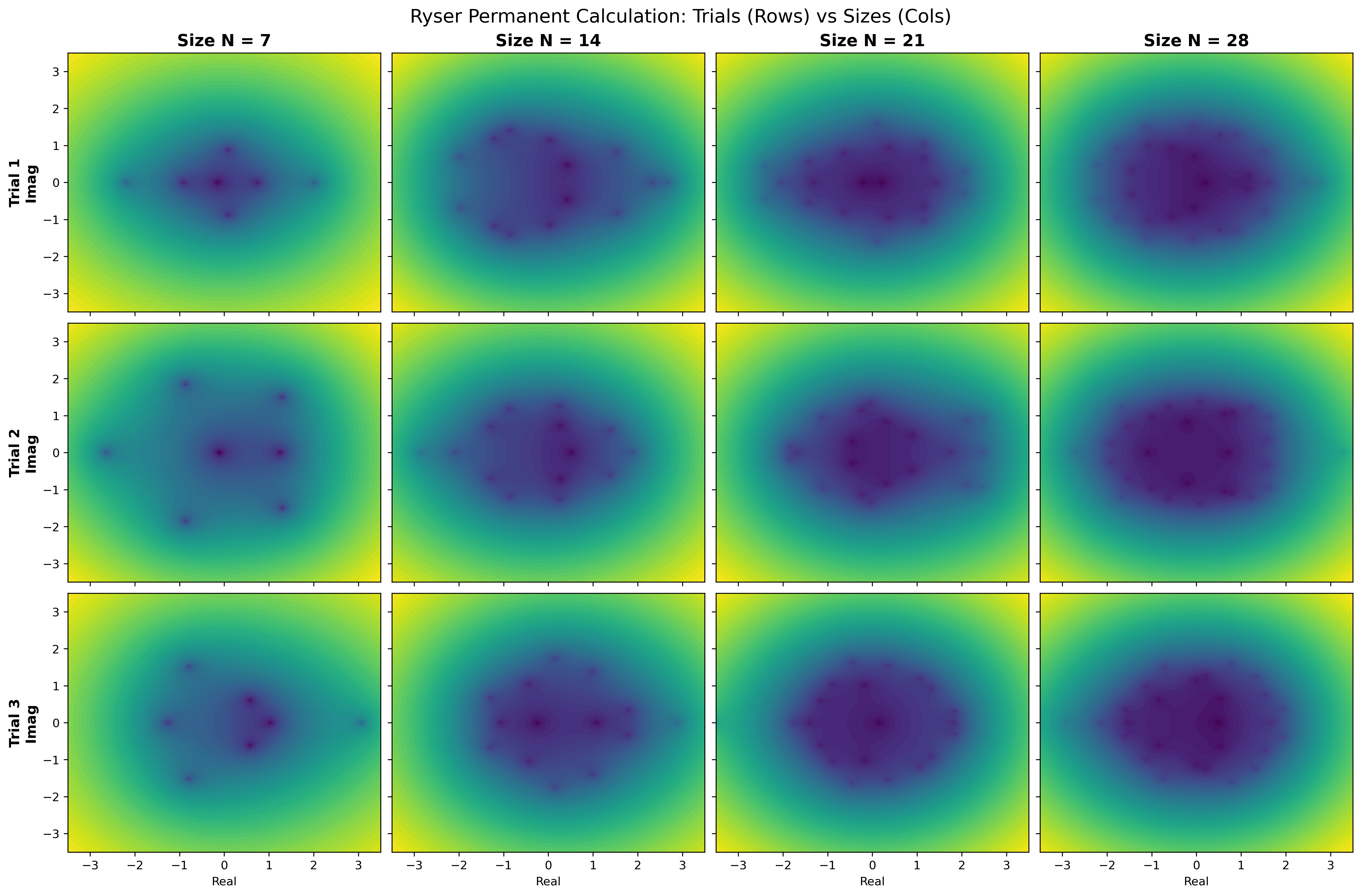}
    \caption{Complex magnitude of the $n \times n$ permanent \eqref{eqn:experiment-eqn} with real Gaussian entries. As before, the zeros are visible as dark blue dots. Because the coefficients are real-valued, and therefore invariant under complex conjugation (which maps $i$ to $-i$), the roots exhibit symmetry across the real axis.}
    \label{fig:ryser_real}
\end{figure}

\begin{figure}
    \centering
    \includegraphics[width=1.02\textwidth,trim={0 0 0 2}, clip]{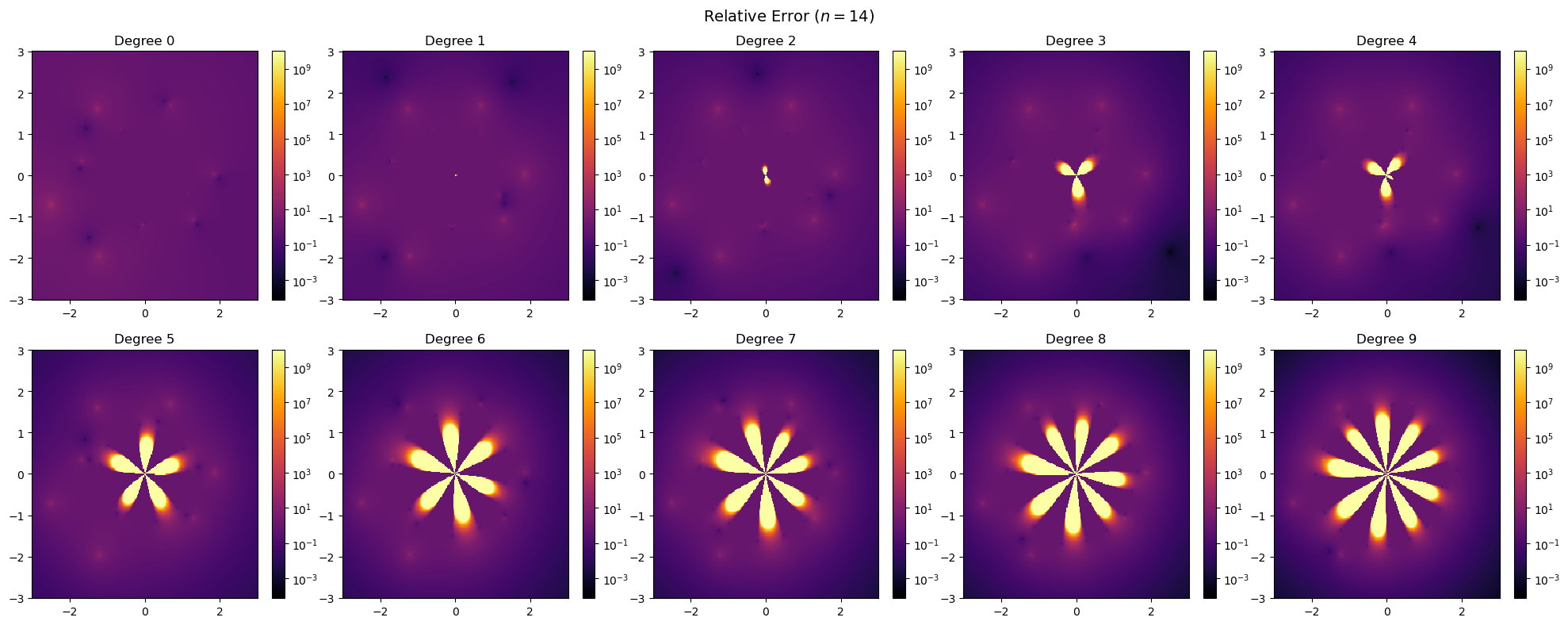} 
    \caption{Relative error of the Taylor series approximation \eqref{eqn:plotted-apx} for $\per (z J/\sqrt{n} + A)$ as a function of the complex parameter $z$ ($n=14$). Each subplot represents an increasing degree of the Taylor series approximation (0 through 9) about $z = \infty$. The color scale is logarithmic, representing the relative deviation $|1 - \per_{\text{approx}}/\per_{\text{exact}}|$.}
    \label{fig:perm_relative_error}
\end{figure}

\begin{rem}
If we consider the analogue of \eqref{eqn:experiment-eqn} with determinant instead of permanent, we will see totally different behavior. By the determinant rank-one update formula,
\[ \det(zJ/\sqrt{n} + W) = \det(W)(1 + z\operatorname{Tr}(J W^{-1})/\sqrt{n}) \]
so the determinantal version is a degree-one polynomial --- it has one zero instead of $n$. 
\end{rem}

\paragraph{Error of Taylor series approximation.}
The approximation to \eqref{eqn:experiment-eqn} is calculated by rewriting the log-permanent in terms of $w = 1/z$:
\begin{equation}\label{eqn:plotted-apx}
    \log \per \left( z \frac{J}{\sqrt{n}} + A \right) = n \log z + \log \per \left( \frac{J}{\sqrt{n}} + w A \right)
\end{equation}
The second term is expanded as a Taylor series in $w$ around the value $0$ (which corresponds to $z = \infty$). Truncating this series at degree $D$ yields:
\begin{equation}
    \log \per \left( \frac{J}{\sqrt{n}} + w A \right) \approx \log \left( \frac{n!}{n^{n/2}} \right) + \sum_{k=1}^{D} q_k w^k
\end{equation}
where the coefficients $q_k$ are computed using the sum of $k \times k$ subpermanents of $A$.

As illustrated in Figure~\ref{fig:perm_relative_error}, the Taylor series approximation exhibits a distinct convergence region that expands as the degree of the expansion increases. Because the approximation is derived as an expansion in powers of $1/z$, the relative error is negligible for large $|z|$, but diverges significantly as $z$ approaches the origin. Specifically, the dark purple regions indicate areas where the approximation achieves high precision, while the bright yellow center denotes the singularity where the $1/z$ terms dominate, leading to a breakdown of the polynomial fit. Notably, by increasing the degree from 0 to 9, the radius of the reliable convergence zone grows, allowing for accurate permanent estimation progressively closer to the center of the complex plane, but this region can never expand beyond the permanental zeros.

\section{Connection to Bethe approximation heuristic}
\label{app:bethe}
This section is mostly of interest to those familiar with Bethe
 approximation.
We explain a loose connection between our second-order
reweighting and the \emph{Bethe approximation} to the permanent,
in that the reweighting is close to a formal second-order Taylor series expansion of the Bethe approximation. We include this since there has been a lot of study of Bethe approximation in the specific context of the permanent, and there are connections between Bethe approximation and diagrammatic expansions \cite{thouless1977solution}.
We otherwise do not use Bethe approximation in any way.

\subsection{A brief introduction to the Bethe approximation}

The permanent of a nonnegative \(n\times n\) matrix \(\theta=(\theta_{ij})\) can
be viewed as the partition function of the perfect matching model on the complete
bipartite graph \(K_{n,n}\):
\[
\per(\theta)=\sum_{\sigma\in S_n}\prod_{i=1}^n \theta_{i,\sigma(i)}.
\]
The Bethe approximation is a standard mean-field approximation to such partition
functions. For the permanent, it can be formulated variationally over the
Birkhoff polytope
\[
\Gamma_n:=\Bigl\{\gamma=(\gamma_{ij})\in\mathbb{R}_{\ge 0}^{n\times n}:
\sum_j \gamma_{ij}=1\ \forall i,\ \sum_i \gamma_{ij}=1\ \forall j\Bigr\},
\]
whose elements may be interpreted as ``soft'' doubly stochastic approximations
to permutation matrices.

For a positive matrix \(\theta\), the Bethe free energy is
\begin{equation}
F_B(\gamma;\theta)
:=
-\sum_{i,j}\gamma_{ij}\log \theta_{ij}
+\sum_{i,j}\gamma_{ij}\log\gamma_{ij}
-\sum_{i,j}(1-\gamma_{ij})\log(1-\gamma_{ij}),
\qquad \gamma\in\Gamma_n.
\label{eq:bethe-free-energy}
\end{equation}
One then defines the \emph{Bethe permanent} by
\[
\per_B(\theta):=\exp\!\Bigl(-\inf_{\gamma\in\Gamma_n}F_B(\gamma;\theta)\Bigr).
\]

\subsection{Quadratic expansion of Bethe approximation}

We now consider a perturbation
\[
A=J+H
\]
with $J$ the all-ones matrix and \(H\) small. For the purpose of the heuristic calculation, we will allow $H$ to be complex-valued, even though this is not so consistent with the variational interpretation. 
If \(H=0\), then \(A=J\) is maximally symmetric, and the
corresponding Bethe variational problem is symmetric under permutations of rows
and columns. The distinguished point in the Birkhoff polytope is therefore the
uniform doubly stochastic matrix
\[
\gamma_0:=\frac1n J.
\]
This is the natural mean-field candidate around which to linearize.

To compare with our exact second derivative, it is convenient to introduce
\[
\Phi_B(\gamma;A):=-F_B(\gamma;A)
=
\sum_{i,j}\gamma_{ij}\log A_{ij}
-
\sum_{i,j}\gamma_{ij}\log\gamma_{ij}
+
\sum_{i,j}(1-\gamma_{ij})\log(1-\gamma_{ij}).
\]
Then
\[
\log \per_B(A)=\sup_{\gamma\in\Gamma_n}\Phi_B(\gamma;A).
\]

A tangent vector to the Birkhoff polytope at \(\gamma_0\) has the form
\[
u\in\mathbb{R}^{n\times n},\qquad
u\mathbf 1=0,\qquad \mathbf 1^T u=0.
\]
Thus the tangent space is exactly the subspace
\[
\mathcal H_m=\{M\in\mathbb{R}^{n\times n}:M\mathbf 1=0,\ \mathbf 1^T M=0\}
\]
that also appears in our exact decomposition. In particular, if
\(\gamma=\gamma_0+u\), then the free variational directions live only in
\(\mathcal H_m\).

\paragraph{Quadratic expansion of the Bethe functional.}

Let
\[
f(x):=-x\log x+(1-x)\log(1-x).
\]
Then the entropy part of \(\Phi_B\) is \(\sum_{i,j}f(\gamma_{ij})\), and
\[
f'(x)=-\log x-\log(1-x)-2,
\qquad
f''(x)=-\frac1x+\frac1{1-x}.
\]
Evaluating at \(x=1/n\) gives
\[
f''(1/n)=-n+\frac{n}{n-1}=-\frac{n(n-2)}{n-1}.
\]
It is therefore convenient to set
\[
q_n:=-f''(1/n)
=\frac1{1/n}-\frac1{1-1/n}
= n-\frac{n}{n-1}
= \frac{n(n-2)}{n-1}.
\]

Now write
\[
L:=\log A
\]
noting that here we are taking the \emph{entrywise logarithm}, so
\[
L = H-\frac12 H^{\circ 2}+O(\|H\|^3).
\]

Expanding \(\Phi_B\) around \(\gamma_0\), we obtain
\[
\Phi_B(\gamma_0+u;A)
=
\Phi_B(\gamma_0;A)
+\langle u,L\rangle
-\frac{q_n}{2}\|u\|_F^2
+O(\|u\|^3).
\]
Since \(u\in\mathcal H_m\), only the \(\mathcal H_m\)-component of \(L\) couples:
\[
\langle u,L\rangle=\langle u,P_mL\rangle.
\]
``Optimizing'' the quadratic approximation over \(u\in\mathcal H_m\) by setting the derivative to zero gives
\[
u_*=\frac1{q_n}P_mL+O(\|L\|^2),
\]
and therefore
\[
\log \per_B(J+H)-\log \per_B(J)
=
\frac1n\sum_{i,j}L_{ij}
+\frac{1}{2q_n}\|P_mL\|_F^2
+O(\|H\|^3).
\]

Substituting \(L=H-\frac12 H^{\circ 2}+O(\|H\|^3)\) and keeping only quadratic
terms yields
\[
\log \per_B(J+H) -\log \per_B(J)
=
\frac1n\sum_{i,j}H_{ij}
-\frac1{2n}\sum_{i,j}H_{ij}^2
+\frac{1}{2q_n}\|P_mH\|_F^2
+O(\|H\|^3).
\]
Since
\[
\|H\|_F^2=\|P_\parallel H\|_F^2+\|P_mH\|_F^2,
\]
we can rewrite this as
\[
\log \per_B(J+H) -\log \per_B(J)
=
\frac1n\sum_{i,j}H_{ij}
+\frac12 \big\langle H,\,
\Bigl(-\frac1n P_\parallel+\Bigl(-\frac1n+\frac1{q_n}\Bigr)P_m\Bigr)H\big\rangle
+O(\|H\|^3).
\]
Using
\[
\frac1{q_n}=\frac{n-1}{n(n-2)},
\qquad
-\frac1n+\frac1{q_n}=\frac1{n(n-2)},
\]
the prediction for the quadratic operator is
\begin{equation}
B_{\mathrm{Bethe}}
=
-\frac1n P_\parallel+\frac1{n(n-2)}P_m.
\label{eq:bethe-quadratic-operator}
\end{equation}

\paragraph{Comparison with the exact second derivative.}

By direct combinatorial computation, our exact second derivative at \(J\) is
\begin{equation}
B_{\mathrm{exact}}
=
-\frac1n P_\parallel+\frac1{n(n-1)}P_m.
\label{eq:exact-quadratic-operator}
\end{equation}
Comparing \eqref{eq:bethe-quadratic-operator} and
\eqref{eq:exact-quadratic-operator}, we find:

\begin{itemize}
\item the expansion of the Bethe heuristic around the saddle point $J$ predicts the same splitting
\[
\mathbb{C}^{n\times n}
=
\mathcal H_0\oplus \mathcal H_r\oplus \mathcal H_c\oplus \mathcal H_m;
\]

\item it gets the coefficient of \(P_\parallel\) exactly right;

\item For \(\mathcal H_m\), it predicts 
\[
\frac1{n(n-2)}
\quad\text{instead of the exact}\quad
\frac1{n(n-1)}.
\]
Of course, these two quantities are quite close for large $n$.
\end{itemize}

\section{Low-degree coefficient bounds}\label{apdx:coeff-bounds}
The following coefficient bounds for the low-degree terms in the taylor series of $\log Z$ are not used
in the proofs for zero-free regions. They can be helpful to
optimize the precise degree of truncation needed for our algorithm. 

The proofs are much shorter and more elegant in the hardcore model, where the bounds are a simple consequence of the cluster expansion formula. For the permanent, the series expansion of $\log Z$ is approximately given by a similar formula (recall the connection from the Overview) so intuitively similar bounds should still hold; making this precise requires a lot more work.

\subsection{Lattice animal bound}
The following lemma is a standard counting estimate which is used in polymer expansion arguments (see, e.g., \cite{friedli2017statistical,perkins2023five}). We give a self-contained proof below.
\begin{lem}
\label{lem:lattice-animal}
Let \(G=(V,E)\) be a finite graph on \(n\) vertices with maximum degree at most \(\Delta\).
For \(s\ge 1\), let \(N_s(G)\) denote the number of connected vertex subsets
\(U\subseteq V\) with \(|U|=s\). Then
\[
N_s(G)\le n\,(e\Delta)^{s-1}.
\]
More precisely, for each fixed vertex \(v\in V\), the number of connected sets
\(U\subseteq V\) such that \(v\in U\) and \(|U|=s\) is at most
\[
\frac{1}{s}\binom{\Delta s}{s-1}\le (e\Delta)^{s-1}.
\]
\end{lem}

\begin{proof}
Fix a vertex \(v\in V\), and let \(\mathcal C_s(v)\) denote the collection of connected
vertex sets \(U\subseteq V\) with \(v\in U\) and \(|U|=s\). It suffices to show that
\[
|\mathcal C_s(v)|\le \frac1s \binom{\Delta s}{s-1}.
\]

For each \(U\in\mathcal C_s(v)\), choose a spanning tree \(T_U\) of the induced subgraph
\(G[U]\), rooted at \(v\). We now encode \(T_U\) by a word of length \(\Delta s\) containing
exactly \(s-1\) ones.

Consider the rooted tree \(T_U\). List the vertices of \(U\) in the order in which they are first
discovered by a depth-first search starting from \(v\). Every vertex other than \(v\) has a unique
parent in \(T_U\), so \(T_U\) has exactly \(s-1\) oriented tree edges pointing from a parent to a child.

For each vertex \(x\in U\), choose once and for all an ordering of the at most \(\Delta\) incident
edges of \(G\). Thus the set \(U\) contributes \(s\) vertex slots, each with at most \(\Delta\)
possible edge positions, for a total of at most \(\Delta s\) positions. Mark those positions
corresponding to the \(s-1\) parent-to-child edges of \(T_U\). This produces a \(0\)-\(1\) word
of length at most \(\Delta s\) with exactly \(s-1\) ones.

This encoding is injective once the root \(v\) is fixed: indeed, from the marked positions one
recovers, for each discovered vertex, which children it has in the rooted tree and in which order
they are explored, hence the whole rooted tree \(T_U\), and therefore also its vertex set \(U\).
Accordingly,
\[
|\mathcal C_s(v)|\le \binom{\Delta s}{s-1}.
\]

To gain the extra factor \(1/s\), note that in the above exploration code every connected set
\(U\in\mathcal C_s(v)\) is counted at least \(s\) times when one allows the root to vary over the
\(s\) vertices of \(U\). Equivalently, summing the rooted count over all \(v\in U\) overcounts each
connected set of size \(s\) by exactly \(s\). Hence the number of connected sets of size \(s\) is
\[
N_s(G)=\frac1s \sum_{v\in V} |\mathcal C_s(v)|
\le \frac{n}{s}\binom{\Delta s}{s-1}.
\]
In particular, for each fixed \(v\),
\[
|\mathcal C_s(v)|\le \frac1s\binom{\Delta s}{s-1}.
\]

Finally, using the standard binomial estimate \(\binom{M}{k}\le (eM/k)^k\), we obtain
\[
\frac1s\binom{\Delta s}{s-1}
\le \frac1s\left(\frac{e\Delta s}{s-1}\right)^{s-1}
\le (e\Delta)^{s-1},
\]
since \(s\le 2(s-1)\) for \(s\ge 2\), and the case \(s=1\) is trivial. Therefore
\[
N_s(G)\le n\,(e\Delta)^{s-1},
\]
as claimed.
\end{proof}
\subsection{Hardcore model}
\begin{lem}[Explicit second-moment bound for the $k$th cluster coefficient]\label{lem:ak-second-moment-explicit}
Let $G=(V,E)$ be a finite graph with $|V|=n$ and maximum degree at most $\Delta\ge 2$. Let $(W_v)_{v\in V}$ be i.i.d.\ standard complex Gaussians,
$W_v\sim \CN(0,1)$,
and define random fugacities $\lambda_v=\lambda W_v$. Write the cluster expansion of the hardcore partition function as
\[
\log Z(\lambda W)=\sum_{k\ge 1} a_k \lambda^k,
\]
where
\[
a_k
=
\sum_{|\mathbf m|=k}
\frac{1}{\mathbf m!}\,\phi(H[\mathbf m])\prod_{v\in V}W_v^{m_v}.
\]
Then for every integer $k\ge 1$,
\[
\mathbb E |a_k|^2
\le
\frac{1}{e\sqrt{2}}\,
n\,\Delta^{k-1}\,k^{2k-4}
\Bigl(\frac{2e}{\log 2}\Bigr)^k
\le
n\,\Delta^{k-1}\,k^{2k}
\Bigl(\frac{3e}{\log 2}\Bigr)^k
\]
\end{lem}

\begin{proof}
By definition,
\[
a_k=\sum_{|\mathbf m|=k} c_{\mathbf m}\prod_{v\in V}W_v^{m_v},
\qquad
c_{\mathbf m}:=\frac{1}{\mathbf m!}\phi(H[\mathbf m]).
\]
Hence
\[
\mathbb E|a_k|^2
=
\sum_{|\mathbf m|=k}\sum_{|\mathbf m'|=k}
c_{\mathbf m}\overline{c_{\mathbf m'}}
\,
\mathbb E\!\left[\prod_{v\in V}W_v^{m_v}\overline{W_v}^{\,m_v'}\right].
\]
Since the variables $(W_v)$ are independent and rotationally invariant,
\[
\mathbb E\!\left[W_v^p\overline{W_v}^{\,q}\right]=0
\qquad\text{unless }p=q.
\]
Therefore all off-diagonal terms vanish, and we obtain
\[
\mathbb E|a_k|^2
=
\sum_{|\mathbf m|=k}
\frac{|\phi(H[\mathbf m])|^2}{(\mathbf m!)^2}
\prod_{v\in V}\mathbb E |W_v|^{2m_v}.
\]
For a standard complex Gaussian $W\sim \CN(0,1)$,
$\mathbb E|W|^{2r}=r!$ for any 
$r\ge 0$.
Thus
\[
\mathbb E|a_k|^2
=
\sum_{|\mathbf m|=k}
|\phi(H[\mathbf m])|^2
\prod_{v:m_v>0}\frac{1}{m_v!}.
\]

Now group the sum according to the support
\[
S=\{v\in V:m_v>0\},
\qquad s:=|S|.
\]
If $H[\mathbf m]$ is connected, then $G[S]$ must also be connected, since an edge of $H[\mathbf m]$ only joins copies of equal or adjacent vertices in $G$. Therefore
\[
\mathbb E|a_k|^2
\le
\sum_{s=1}^k
\sum_{\substack{S\subseteq V\\ |S|=s\\ G[S]\text{ connected}}}
\ \sum_{\substack{(m_v)_{v\in S}\in \mathbb Z_{\ge 1}^S\\ \sum_{v\in S}m_v=k}}
|\phi(H[\mathbf m])|^2
\prod_{v\in S}\frac{1}{m_v!}.
\]
For each such $\mathbf m$, the graph $H[\mathbf m]$ has exactly $k$ vertices counting multiplicity. By the spanning tree bound (Lemma~\ref{lem:ursell-tree-bound}),
\[
|\phi(H[\mathbf m])| \le k^{k - 2}.
\]
Hence
\[
\mathbb E|a_k|^2
\le
k^{2k-4}
\sum_{s=1}^k N_s(G)
\sum_{\substack{m_1,\dots,m_s\ge 1\\ m_1+\cdots+m_s=k}}
\frac{1}{m_1!\cdots m_s!},
\]
where $N_s(G)$ is the number of connected vertex subsets of $G$ of size $s$.

Next, we bound $N_s(G)$. A standard exploration argument (Lemma~\ref{lem:lattice-animal}) yields
\[
N_s(G)\le n\,(e\Delta)^{s-1}.
\]
Substituting,
\[
\mathbb E|a_k|^2
\le
n\,k^{2k-4}
\sum_{s=1}^k (e\Delta)^{s-1}
\sum_{\substack{m_1,\dots,m_s\ge 1\\ m_1+\cdots+m_s=k}}
\frac{1}{m_1!\cdots m_s!}.
\]
Since $s\le k$, we may bound $(e\Delta)^{s-1}\le (e\Delta)^{k-1}$, obtaining
\[
\mathbb E|a_k|^2
\le
n\,\Delta^{k-1}\,e^{k-1}\,k^{2k-4}
\sum_{s=1}^k
\sum_{\substack{m_1,\dots,m_s\ge 1\\ m_1+\cdots+m_s=k}}
\frac{1}{m_1!\cdots m_s!}.
\]
It remains to estimate the composition sum. Define
\[
b_k
:=
\sum_{s=1}^k
\sum_{\substack{m_1,\dots,m_s\ge 1\\ m_1+\cdots+m_s=k}}
\frac{1}{m_1!\cdots m_s!}.
\]
Its generating function is
\[
\sum_{k\ge 1} b_k x^k
=
\sum_{s\ge 1}\Bigl(\sum_{m\ge 1}\frac{x^m}{m!}\Bigr)^s
=
\sum_{s\ge 1}(e^x-1)^s
=
\frac{e^x-1}{2-e^x}.
\]
Fix
\[
r:=\frac{\log 2}{2}.
\]
Since the power series has nonnegative coefficients, Cauchy's estimate gives
\[
b_k\le \frac{1}{r^k}\,\frac{e^r-1}{2-e^r}.
\]
Now $e^r=\sqrt{2}$, so
\[
\frac{e^r-1}{2-e^r}
=
\frac{\sqrt{2}-1}{2-\sqrt{2}}
=
\frac{1}{\sqrt{2}}.
\]
Therefore
\[
b_k\le \frac{1}{\sqrt{2}}\Bigl(\frac{2}{\log 2}\Bigr)^k.
\]
Substituting this bound,
\[
\mathbb E|a_k|^2
\le
n\,\Delta^{k-1}\,e^{k-1}\,k^{2k-4}\,
\frac{1}{\sqrt{2}}
\Bigl(\frac{2}{\log 2}\Bigr)^k
=
\frac{1}{e\sqrt{2}}\,
n\,\Delta^{k-1}\,k^{2k-4}
\Bigl(\frac{2e}{\log 2}\Bigr)^k.
\]
This proves the first bound.

The second displayed bound is a softer simplification: since $k^{2k-4}\le k^{2k}$ and
\[
\frac{1}{e\sqrt{2}}\cdot \frac{2e}{\log 2}
<
\frac{3e}{\log 2},
\]
we may enlarge constants and write
\[
\mathbb E|a_k|^2
\le
n\,\Delta^{k-1}\,k^{2k}
\Bigl(\frac{3e}{\log 2}\Bigr)^k.
\qedhere
\]
\end{proof}
\subsection{Permanent}
In order to argue that the Taylor series for the log-permanent decays similarly to the idealized model, we use the following representation of the ``correction term'' describing the discrepancy between $(n - m)!/n!$ and $1/n^m$. Intuitively, this correction should be negligible for low-degree coefficients since this corresponds to small $m$.
\begin{lem}[Finite binomial-basis representation of the cardinality correction]
\label{lem:rho-binomial}
Fix an integer \(k\ge 1\). For each \(n\ge 1\), define
\[
\rho_n(m):=\frac{(n-m)!}{n!}\,n^m
=\prod_{t=0}^{m-1}\Bigl(1-\frac{t}{n}\Bigr),
\qquad 0\le m\le n.
\]
Then there exist coefficients \(\theta_{q,n}\), \(0\le q\le k\), such that
\[
\log \rho_n(m)=\sum_{q=0}^k \theta_{q,n}\binom{m}{q}
\qquad\text{for all }m=0,1,\dots,k
\]
whenever \(n\ge k\). Moreover
$\theta_{0,n}=\theta_{1,n}=0$,
and there is an absolute constant \(C>0\) such that for every fixed \(k\ge 1\),
\[
|\theta_{q,n}| \le (Ck)^k\, n^{-(q-1)},
\qquad 2\le q\le k,
\]
for all \(n\ge k\).

Consequently, if we set
\[
\xi_{q,n}:=e^{\theta_{q,n}}-1,
\]
then, after possibly enlarging \(C\),
\[
|\xi_{q,n}|\le (Ck)^k\, n^{-(q-1)},
\qquad 2\le q\le k,
\]
and for every \(m\in\{0,1,\dots,k\}\),
\[
\rho_n(m)
=
\exp\!\Bigl(\sum_{q=2}^k \theta_{q,n}\binom{m}{q}\Bigr)
=
\prod_{\substack{S\subseteq [m]\\ 2\le |S|\le k}} (1+\xi_{|S|,n}).
\]
\end{lem}

\begin{proof}
Fix \(k\). Since the functions
\[
m\longmapsto \binom{m}{q},
\qquad q=0,1,\dots,k,
\]
form a basis of the space of functions on \(\{0,1,\dots,k\}\), there exist unique coefficients \(\theta_{q,n}\) such that
\[
\log \rho_n(m)=\sum_{q=0}^k \theta_{q,n}\binom{m}{q}
\qquad (0\le m\le k).
\]
Because \(\rho_n(0)=\rho_n(1)=1\), we have \(\log \rho_n(0)=\log \rho_n(1)=0\), which forces \(\theta_{0,n}=\theta_{1,n}=0\).

To bound the coefficients, write
\[
\log \rho_n(m)
=
\sum_{t=0}^{m-1}\log\Bigl(1-\frac{t}{n}\Bigr)
=
-\sum_{\ell\ge 1}\frac{1}{\ell n^\ell}\sum_{t=0}^{m-1} t^\ell.
\]
Now use the standard identity (see, e.g., \cite{stanley2011enumerative})
\[
\sum_{t=0}^{m-1} t^\ell
=
\sum_{r=0}^{\ell} r!\,S(\ell,r)\binom{m}{r+1},
\]
where \(S(\ell,r)\) is the Stirling number of the second kind. Therefore
\[
\theta_{q,n}
=
-\sum_{\ell\ge q-1}\frac{(q-1)!\,S(\ell,q-1)}{\ell\,n^\ell}.
\]
Using the crude bound
\[
(q-1)!\,S(\ell,q-1)\le (q-1)^\ell,
\]
we obtain
\[
|\theta_{q,n}|
\le
\sum_{\ell\ge q-1}\frac{(q-1)^\ell}{\ell\,n^\ell}.
\]
If \(n\ge 2k\), then since \(q-1\le k\),
\[
|\theta_{q,n}|
\le
\sum_{\ell\ge q-1}\Bigl(\frac{k}{n}\Bigr)^\ell
\le
2\Bigl(\frac{k}{n}\Bigr)^{q-1}
\le
(2k)^k n^{-(q-1)}.
\]
For the finitely many \(n\) with \(k\le n<2k\), the same bound follows after enlarging the constant, since \(k\) is fixed. This proves
\[
|\theta_{q,n}| \le (Ck)^k\, n^{-(q-1)}.
\]

The bound on \(\xi_{q,n}=e^{\theta_{q,n}}-1\) follows similarly after enlarging \(C\), since \(|\theta_{q,n}|\le (Ck)^k\) for fixed \(k\), and for large \(n\) one has \(|e^u-1|\le 2|u|\) when \(|u|\le 1\).

Finally,
\[
\exp\!\Bigl(\sum_{q=2}^k \theta_{q,n}\binom{m}{q}\Bigr)
=
\prod_{q=2}^k \prod_{\substack{S\subseteq [m]\\ |S|=q}} e^{\theta_{q,n}}
=
\prod_{\substack{S\subseteq [m]\\ 2\le |S|\le k}} (1+\xi_{|S|,n}),
\]
since \(\binom{m}{q}\) is the number of \(q\)-subsets of \([m]\).
\end{proof}

\begin{lem}[Connected decorated-hypergraph expansion]
\label{lem:connected-hypergraph-expansion}
Let \(V\) be a finite set, and for each finite \(X\subseteq V\) let \(\mathcal H(X)\) be a family of decorated hypergraphs on vertex set \(X\), with weight \(w(H)\in\mathbb C\), such that:

\begin{enumerate}
\item every \(H\in\mathcal H(X)\) decomposes uniquely into connected components \(H_1,\dots,H_r\) on a partition
\[
X=X_1\sqcup\cdots\sqcup X_r;
\]
\item the weight is multiplicative across connected components:
\[
w(H)=\prod_{j=1}^r w(H_j).
\]
\end{enumerate}

Define
\[
A(X):=\sum_{H\in\mathcal H(X)} w(H),
\qquad
C(X):=\sum_{H\in\mathcal H_{\mathrm{conn}}(X)} w(H),
\]
where \(\mathcal H_{\mathrm{conn}}(X)\subseteq \mathcal H(X)\) denotes the connected hypergraphs on \(X\). Then one has the formal identity
\[
1+\sum_{\varnothing\neq X\subseteq V} A(X)\prod_{x\in X} y_x
=
\exp\!\Bigl(\sum_{\varnothing\neq X\subseteq V} C(X)\prod_{x\in X} y_x\Bigr).
\]
\end{lem}

\begin{proof}
Expanding the exponential on the right-hand side and collecting the coefficient of \(\prod_{x\in X} y_x\), one obtains
\[
\sum_{\pi\in\Pi(X)} \prod_{B\in\pi} C(B),
\]
where \(\Pi(X)\) denotes the set of set partitions of \(X\). By the assumptions, choosing for each block \(B\in\pi\) a connected decorated hypergraph \(H_B\in \mathcal H_{\mathrm{conn}}(B)\) is equivalent to choosing a decorated hypergraph \(H\in\mathcal H(X)\) together with its connected-component decomposition, and the weights multiply accordingly. Thus
\[
A(X)=\sum_{\pi\in\Pi(X)} \prod_{B\in\pi} C(B),
\]
which is exactly the claimed identity.
\end{proof}

\begin{prop}[Second moment bound for the true logarithmic coefficients]
\label{prop:true-log-perm-coeff}
Let
\[
P_n(z):=\frac{1}{n!}\per(J+zW),
\qquad
\log P_n(z)=\sum_{k\ge1} b_{k,n} z^k,
\]
where \(W=(W_{ij})_{1\le i,j\le n}\) has i.i.d.\ standard complex Gaussian entries,
$W_{ij}\sim\CN(0,1)$.
Then there exists an absolute constant \(C>0\) such that for every fixed integer \(k\ge1\), one can choose
\[
C_k \le \exp\!\bigl(C\,k\,2^k\log k\bigr)
\]
so that
\[
\mathbb E |b_{k,n}|^2 \le C_k\, n^{1-k}
\qquad\text{for all }n\ge 1.
\]
\end{prop}

\begin{proof}
Fix \(k\ge1\). For \(n<k\), the desired bound can be absorbed into the constant \(C_k\), since there are only finitely many such \(n\). Thus we may assume \(n\ge k\).

Let \(E_n:=E(K_{n,n})\), so \(|E_n|=n^2\). For each edge \(e\in E_n\), set
\[
x_e:=\frac{zW_e}{n}.
\]
Then
\[
P_n(z)
=
\sum_{M\subseteq E_n \atop M\text{ matching}}
\rho_n(|M|)\prod_{e\in M} x_e,
\]
where
\[
\rho_n(m)=\frac{(n-m)!}{n!}\,n^m.
\]

Since the coefficient \(b_{k,n}\) only depends on matchings \(M\) with \(|M|\le k\), Lemma~\ref{lem:rho-binomial} gives an exact representation
\[
\rho_n(|M|)
=
\prod_{\substack{S\subseteq M\\ 2\le |S|\le k}} (1+\xi_{|S|,n})
\qquad\text{whenever }|M|\le k,
\]
with
\[
|\xi_{q,n}|\le D_k\, n^{-(q-1)},
\qquad 2\le q\le k,
\]
where we may take
$D_k:=(Ck)^k$
for some absolute \(C\).

Next define, for distinct edges \(e,f\in E_n\),
\[
f(e,f):=
\begin{cases}
-1,& e\text{ and }f\text{ share a row or column},\\
0,& \text{otherwise}.
\end{cases}
\]
Then for any finite \(X\subseteq E_n\),
\[
\prod_{\{e,f\}\subseteq X}(1+f(e,f))
=
\mathbf 1_{\{X\text{ is a matching}\}}.
\]
Therefore, up to degree \(k\),
\[
P_n(z)
=
\sum_{X\subseteq E_n}
\Bigl(\prod_{e\in X}x_e\Bigr)
\Bigl(\prod_{\{e,f\}\subseteq X}(1+f(e,f))\Bigr)
\Bigl(\prod_{\substack{S\subseteq X\\ 2\le |S|\le k}} (1+\xi_{|S|,n})\Bigr).
\]

For a finite \(X\subseteq E_n\), let \(\mathcal H(X)\) be the collection of decorated hypergraphs on vertex set \(X\) whose allowed hyperedges are:
\begin{itemize}
\item pair-links \(\{e,f\}\subseteq X\), with weight \(f(e,f)\);
\item correction hyperedges \(S\subseteq X\) with \(2\le |S|\le k\), with weight \(\xi_{|S|,n}\).
\end{itemize}
If \(H\in\mathcal H(X)\), write \(w(H)\) for the product of the weights of all hyperedges of \(H\). Then
\[
\prod_{\{e,f\}\subseteq X}(1+f(e,f))
\prod_{\substack{S\subseteq X\\ 2\le |S|\le k}} (1+\xi_{|S|,n})
=
\sum_{H\in\mathcal H(X)} w(H).
\]
Thus
\[
P_n(z)
=
\sum_{X\subseteq E_n}
\Bigl(\prod_{e\in X}x_e\Bigr)
\sum_{H\in\mathcal H(X)} w(H).
\]

Applying Lemma~\ref{lem:connected-hypergraph-expansion}, we obtain
\[
\log P_n(z)
=
\sum_{\varnothing\neq X\subseteq E_n}
\Bigl(\prod_{e\in X}x_e\Bigr)
\sum_{H\in\mathcal H_{\mathrm{conn}}(X)} w(H),
\]
where \(\mathcal H_{\mathrm{conn}}(X)\) denotes the connected decorated hypergraphs on \(X\). Taking the coefficient of \(z^k\), we find
\[
b_{k,n}
=
\frac{1}{n^k}
\sum_{\substack{X\subseteq E_n\\ |X|=k}}
\Bigl(\prod_{e\in X} W_e\Bigr)\Psi_n(X),
\]
where
\[
\Psi_n(X):=\sum_{H\in\mathcal H_{\mathrm{conn}}(X)} w(H).
\]

Since the \(W_e\) are i.i.d.\ rotationally invariant standard complex Gaussians,
\[
\mathbb E\!\left[
\prod_{e\in X}W_e\,
\overline{\prod_{f\in Y}W_f}
\right]
=
\mathbf 1_{\{X=Y\}}
\]
whenever \(|X|=|Y|=k\). Hence
\[
\mathbb E |b_{k,n}|^2
=
\frac{1}{n^{2k}}
\sum_{\substack{X\subseteq E_n\\ |X|=k}}
|\Psi_n(X)|^2.
\]

Now fix a \(k\)-element set \(X\). The total number of possible decorated hypergraphs on \(X\) is at most
\[
4^{\binom{k}{2}}\prod_{q=3}^k 2^{\binom{k}{q}}
\le \exp(C2^k)
\]
for some absolute \(C\): for each 2-subset there is the choice of a pair-link, a correction hyperedge, both, or neither; and for each subset of size \(q\ge3\) there is the choice present/absent. Therefore
\[
|\Psi_n(X)|^2
\le
\exp(C2^k)\sum_{H\in\mathcal H_{\mathrm{conn}}(X)} |w(H)|^2.
\]

It remains to bound the contribution of a fixed abstract connected decorated hypergraph \(H\) on \(k\) labeled vertices \(\{1,\dots,k\}\). Let \(G_{\mathrm{pair}}(H)\) be the graph obtained from \(H\) by retaining only its pair-links, and let the connected components of \(G_{\mathrm{pair}}(H)\) have sizes
\[
s_1,\dots,s_c,
\qquad s_1+\cdots+s_c=k.
\]
To embed one such pair-component of size \(s\) into \(E_n\), choose the first edge in at most \(n^2\) ways, and then each subsequent vertex along a spanning tree has at most \(2n\) choices, since in \(K_{n,n}\) an edge has at most \(2n-2\) neighbors sharing a row or column. Thus the number of embeddings of that component is at most
$n^2(2n)^{s-1}$.
Multiplying over all \(c\) pair-components, the number of embeddings of the pair-link structure of \(H\) is at most
\[
\prod_{j=1}^c n^2(2n)^{s_j-1}
=
2^{k-c} n^{k+c}.
\]

Now suppose the correction hyperedges of \(H\) have sizes \(q_1,\dots,q_h\), where each \(q_a\in\{2,\dots,k\}\). Then
\[
|w(H)|^2
\le
\prod_{a=1}^h |\xi_{q_a,n}|^2
\le
D_k^{2h}\, n^{-2\sum_{a=1}^h(q_a-1)}.
\]
Since there are at most \(2^k\) nonempty subsets of \(\{1,\dots,k\}\), we certainly have
$h\le 2^k$, so
\[
|w(H)|^2
\le
D_k^{2^k}\, n^{-2\sum_{a=1}^h(q_a-1)}.
\]
Because the full decorated hypergraph \(H\) is connected, these correction hyperedges must connect the \(c\) pair-components into a connected incidence structure. Each hyperedge of size \(q\) can decrease the number of connected components by at most \(q-1\), hence necessarily
\[
\sum_{a=1}^h (q_a-1)\ge c-1.
\]
Therefore
\[
|w(H)|^2 \le D_k^{2^k}\, n^{-2(c-1)}.
\]

Combining the embedding count with the weight bound, the total contribution of all embeddings of this one abstract hypergraph \(H\) to the double sum above is at most
\[
2^{k-c} n^{k+c}\cdot D_k^{2^k}\, n^{-2(c-1)}
=
2^{k-c} D_k^{2^k} n^{k-c+2}.
\]
After multiplying by the prefactor \(n^{-2k}\), its contribution to \(\mathbb E|b_{k,n}|^2\) is at most
\[
2^{k-c} D_k^{2^k} n^{-2k}n^{k-c+2}
=
2^{k-c} D_k^{2^k} n^{2-k-c}
\le
2^k D_k^{2^k} n^{1-k},
\]
since \(c\ge1\).

Finally, there are at most \(\exp(C2^k)\) abstract connected decorated hypergraphs on \(k\) labeled vertices. Summing the preceding bound over this family yields
\[
\mathbb E |b_{k,n}|^2
\le
\exp(C2^k)\,2^k\,D_k^{2^k}\, n^{1-k}.
\]
Since \(D_k=(Ck)^k\), the right-hand side is bounded by
\[
\exp(C2^k)\,\exp(Ck2^k\log k)\, n^{1-k}
\le
\exp\!\bigl(C\,k\,2^k\log k\bigr)\, n^{1-k},
\]
after enlarging the absolute constant \(C\). This proves the claim.
\end{proof}

\section{Second-order bound on \(L(K_{n,n})\)}
\label{apdx:monomer-dimer-second-order}
Here we prove the result described in Section~\ref{subsec:second-monomer-dimer}, see there for the description of the setup. We remark that the proof of this result is fairly different and arguably more involved than its corresponding version for the permanent. Broadly speaking, this is because the partition function is no longer a sum over permutations. Instead, the sum is over matchings, which have different sizes and therefore do not seem to admit as ``unified'' an analysis.

\paragraph{Notation.} For a finite set \(U\subseteq [n]\times[n]\), let \(H_U\) denote the bipartite graph on the
touched row- and column-vertices with edge set \(U\). Define
\[
F_U^{(2)}(y)
:=
Z_{H_U}(y)\exp\!\Bigl(
-\sum_{e\in U} y_e
+\frac12\sum_{e\in U} y_e^2
+\sum_{\substack{e,f\in U\\ e\sim f}} y_e y_f
\Bigr),
\]
where adjacency \(e\sim f\) is taken in the line graph \(L(H_U)\), i.e. \(e\) and \(f\)
share a row or a column. Write
\[
F_U^{(2)}(y)=\sum_{\alpha\in \NN^U} d_{U,\alpha}\,y^\alpha,
\qquad
\mathcal C_U(y):=\sum_{\supp(\alpha)=U} d_{U,\alpha}\,y^\alpha.
\]
Thus \(\mathcal C_U\) is the exact-support part of \(F_U^{(2)}\).

For \(z\in\C\), let \(W=(W_e)_{e\in U}\) have i.i.d.\ entries \(W_e\sim \CN(0,1)\), and set
\[
w(U;z):=\E\bigl[|\mathcal C_U(zW)|^2\bigr],
\qquad
t:=|z|^2.
\]

\subsection{Counting lemma}
The following proposition covers an important counting step which shows up in the analysis. It builds upon the classical fact that Catalan numbers count plane trees \cite{stanley2015catalan}, and more specifically this kind of argument is used in the theory of combinatorial maps \cite{lando2004graphs}. 
\begin{prop}\label{prop:weighted-connected-supports}
Fix a root edge \(e_0\in [n]\times[n]\). For a connected support
\(U\subseteq [n]\times[n]\), let \(H_U\) be the corresponding connected bipartite
graph on row-vertices and column-vertices, and write
\[
r_i(U):=\deg_{H_U}(i,\cdot),
\qquad
c_j(U):=\deg_{H_U}(\cdot,j).
\]
Then for every \(m\ge 1\),
\[
\sum_{\substack{U\ni e_0\\ U\text{ connected}\\ |U|=m}}
\prod_{i:\,r_i(U)\ge1}(r_i(U)-1)!
\prod_{j:\,c_j(U)\ge1}(c_j(U)-1)!
\le
\Cat_m\, n^{m-1}
\le
4^m n^{m-1},
\]
where
\[
\Cat_m:=\frac1{m+1}\binom{2m}{m}
\]
is the \(m\)-th Catalan number.
\end{prop}

\begin{proof}
Without loss of generality, we take $e_0=(1,1)$.
Let \(\mathcal O_m\) be the set of pairs \((U,\sigma)\), where

\begin{itemize}
\item \(U\subseteq[n]\times[n]\) is connected, contains \(e_0\), and has \(|U|=m\);
\item for each touched row-vertex or column-vertex \(v\) of \(H_U\), \(\sigma_v\) is a
cyclic order of the edges of \(H_U\) incident to \(v\).
\end{itemize}

If a vertex has degree \(d\), then it has exactly \((d-1)!\) cyclic orders. Hence
\[
|\mathcal O_m|
=
\sum_{\substack{U\ni e_0\\ U\text{ connected}\\ |U|=m}}
\prod_{i:\,r_i(U)\ge1}(r_i(U)-1)!
\prod_{j:\,c_j(U)\ge1}(c_j(U)-1)!.
\]
Thus it suffices to prove
\[
|\mathcal O_m|\le \Cat_m\, n^{m-1}.
\]

Let \(\mathcal T_m\) be the set of edge-rooted plane bipartite trees with \(m\) edges,
whose root edge is oriented from a row-vertex to a column-vertex, together with labels
on the vertices such that

\begin{itemize}
\item the initial endpoint of the root edge is a row-vertex labeled \(1\);
\item the terminal endpoint of the root edge is a column-vertex labeled \(1\);
\item every other row-vertex carries an arbitrary label in \([n]\);
\item every other column-vertex carries an arbitrary label in \([n]\).
\end{itemize}

An edge-rooted plane tree with \(m\) edges is counted by \(\Cat_m\), and after fixing
the labels \(1\) and \(1\) on the two endpoints of the root edge, the remaining
\(m-1\) vertices may be labeled arbitrarily in \([n]\). Therefore
\[
|\mathcal T_m|=\Cat_m\,n^{m-1}.
\]

We now construct an injection
\[
\Phi:\mathcal O_m\hookrightarrow \mathcal T_m.
\]

Fix \((U,\sigma)\in\mathcal O_m\), and let \(H_U\) be the corresponding connected
bipartite graph. We perform a depth-first unfolding of \(H_U\) from the root edge
\(e_0=(1,1)\), using the prescribed cyclic orders \(\sigma\).

Start with a root row-copy labeled \(1\) and a root column-copy labeled \(1\), joined
by the root tree edge corresponding to \(e_0\). Each copy \(x\) in the unfolding
remembers the original vertex \(v\in H_U\) from which it came. For each non-root copy
\(x\), let \(p_x\) denote the original edge of \(H_U\) through which \(x\) was first
reached. For the two root copies, set \(p_x:=e_0\).

To explore a copy \(x\) of an original vertex \(v\), look at the cyclic order
\(\sigma_v\), start immediately after \(p_x\), and traverse one full turn around
\(v\). Whenever an incident original edge \(e=(v,u)\) is encountered for the first
time, mark \(e\) as discovered, create a child copy \(y\) of \(u\), connect \(x\) to
\(y\) by a tree edge representing \(e\), and recursively explore \(y\).

Since each original edge is discovered exactly once, this procedure produces an
edge-rooted plane bipartite tree with exactly \(m\) edges. Each copy inherits the
label of its original row/column vertex, so \(\Phi(U,\sigma)\in\mathcal T_m\).

It remains to show that \(\Phi\) is injective. Let
\[
T=\Phi(U,\sigma).
\]

We first recover the support \(U\). The vertices of \(T\) are row-copies and column-copies,
each carrying a label in \([n]\). Identify all row-copies with the same label, and likewise
all column-copies with the same label. Each tree edge of \(T\) then becomes an edge
\((i,j)\in [n]\times[n]\). Since every original edge of \(U\) was discovered exactly once in
the unfolding, each edge of \(U\) gives rise to exactly one tree edge of \(T\), and hence this
procedure reconstructs the original support \(U\subseteq [n]\times[n]\).

It remains to reconstruct the cyclic orders \(\sigma_v\) at the touched vertices \(v\) of \(H_U\).
Fix such a vertex \(v\), and let \(E(v)\) denote the set of original edges of \(U\) incident to \(v\).

Consider all copies \(x\) of \(v\) appearing in the unfolded tree \(T\). For each such copy,
define a distinguished incident edge \(p_x\in E(v)\) as follows: if \(x\) is one of the two root
copies, set \(p_x:=e_0\); otherwise, let \(p_x\) be the original edge corresponding to the unique
tree edge joining \(x\) to its parent. Thus \(p_x\) is exactly the edge through which the depth-first
search first reached the copy \(x\).

Next, let \(C_x\) be the ordered list of original edges corresponding to the child edges of \(x\),
read in the plane order at \(x\). By construction of the unfolding, \(C_x\) is precisely the list
of edges encountered when one scans cyclically around the original vertex \(v\), starting
immediately after \(p_x\), and continuing until the search returns along \(p_x\). In particular,
\(C_x\) is a consecutive block in the cyclic order \(\sigma_v\), and the block \(C_x\) is to be
inserted immediately after \(p_x\).

Moreover, every edge of \(E(v)\) appears exactly once in this description: each edge appears once
as some entrance edge \(p_x\), namely at the copy first reached through that edge, and every other
occurrence is recorded exactly once in one of the ordered child-blocks \(C_x\). Therefore the
cyclic order \(\sigma_v\) is partitioned into segments of the form
\[
p_x,\; C_x,
\]
one for each copy \(x\) of \(v\).

Finally, the copies \(x\) of \(v\) appear around the contour traversal of the plane tree \(T\) in
exactly the same cyclic order as these segments occur around the original vertex \(v\). Indeed, a
new copy of \(v\) is created precisely when the search re-enters \(v\) through a new incident edge,
and the plane order records the order in which the previously undiscovered edges following that
entrance edge are explored. Hence \(\sigma_v\) is uniquely recovered by listing the copies \(x\) of
\(v\) in contour order and, for each such \(x\), writing down the segment
\[
p_x,\; C_x.
\]

This reconstructs the cyclic order \(\sigma_v\) uniquely. Doing this for every touched row-vertex
and column-vertex \(v\) recovers the full collection \(\sigma\). Therefore \((U,\sigma)\) is uniquely
determined by \(T\), so \(\Phi\) is injective.

Thus
\[
|\mathcal O_m|\le |\mathcal T_m|=\Cat_m\,n^{m-1},
\]
which proves the proposition.
\end{proof}

\subsection{Analysis}

\begin{lem}[Exact-support extraction and Gaussian Poincar\'e]
\label{lem:exact-support-poincare}
For every finite \(U\subseteq [n]\times[n]\),
\[
\mathcal C_U(zW)
=
\Bigl(\prod_{e\in U}(I-\E_e)\Bigr)F_U^{(2)}(zW),
\]
where \(\E_e\) denotes expectation in the single coordinate \(W_e\), keeping all other
coordinates fixed. Consequently,
\[
w(U;z)
\le
\bigl\|\partial_U F_U^{(2)}(zW)\bigr\|_{L^2}^2,
\qquad
\partial_U:=\prod_{e\in U}\partial_{W_e}.
\]
\end{lem}

\begin{proof}
Expand
\[
F_U^{(2)}(zW)=\sum_{\alpha\in \NN^U} d_{U,\alpha}\, z^{|\alpha|} W^\alpha.
\]
Since \(\E[W_e^k]=0\) for every \(k\ge 1\), the operator \(I-\E_e\) kills precisely those
monomials for which \(\alpha_e=0\). Therefore
\[
\Bigl(\prod_{e\in U}(I-\E_e)\Bigr)F_U^{(2)}(zW)
=
\sum_{\supp(\alpha)=U} d_{U,\alpha}\, z^{|\alpha|} W^\alpha
=
\mathcal C_U(zW).
\]

The inequality now follows by the Poincar\'e inequality (see, e.g., \cite{van2014probability}).
Directly, let \(g(w)=\sum_{k\ge0} a_k w^k\) be an entire function of one complex Gaussian
variable \(W\sim\CN(0,1)\). Using orthogonality of complex Gaussian monomials,
\[
\|(I-\E)g(W)\|_2^2
=
\sum_{k\ge1}|a_k|^2\,k!,
\]
while
\[
\|\partial g(W)\|_2^2
=
\sum_{k\ge1}k^2 |a_k|^2 (k-1)!
=
\sum_{k\ge1}k\,|a_k|^2\,k!
\ge
\sum_{k\ge1}|a_k|^2\,k!.
\]
Hence
\[
\|(I-\E)g(W)\|_2^2\le \|\partial g(W)\|_2^2.
\]
Applying this conditionally, one coordinate at a time, gives
\[
\Bigl\|\Bigl(\prod_{e\in U}(I-\E_e)\Bigr)F_U^{(2)}(zW)\Bigr\|_2^2
\le
\bigl\|\partial_U F_U^{(2)}(zW)\bigr\|_2^2.
\]
This proves the lemma.
\end{proof}

\begin{lem}[Explicit cancellation for \(|U|=1,2\)]
\label{lem:size-1-2-cancellation}
There exists an absolute constant \(C>0\) such that for all sufficiently small \(t=|z|^2\),
the following hold:
\begin{enumerate}
\item if \(|U|=1\), then
\[
w(U;z)\le C t^3;
\]
\item if \(|U|=2\) and \(U\) is connected, then
\[
w(U;z)\le C t^3.
\]
\end{enumerate}
\end{lem}

\begin{proof}
Define
\[
f(s):=(1+s)e^{-s+s^2/2}.
\]
Since
$\log f(s)=\log(1+s)-s+\frac{s^2}{2}
=\frac{s^3}{3}-\frac{s^4}{4}+\frac{s^5}{5}-\cdots$,
we have
\[
f(s)=1+O(s^3)
\qquad (s\to 0).
\]

If \(U=\{e\}\), then
\[
F_U^{(2)}(y_e)=(1+y_e)e^{-y_e+y_e^2/2}=f(y_e),
\]
so
\[
\mathcal C_U(y_e)=f(y_e)-1.
\]
Every monomial in \(\mathcal C_U\) has degree at least \(3\), hence
\[
w(U;z)=\E\bigl[|\mathcal C_U(zW_e)|^2\bigr]\le C t^3.
\]

Now let \(U=\{e,f\}\) with \(e\sim f\). The only matchings in \(H_U\) are
\(\varnothing,\{e\},\{f\}\), so
\[
Z_{H_U}(y_e,y_f)=1+y_e+y_f.
\]
Also,
\[
-\sum_{u\in U} y_u+\frac12\sum_{u\in U}y_u^2+\sum_{u\sim v}y_u y_v
=
-(y_e+y_f)+\frac12(y_e+y_f)^2.
\]
Therefore
\[
F_U^{(2)}(y_e,y_f)=f(y_e+y_f).
\]
The exact-support part is
\[
\mathcal C_U(y_e,y_f)
=
f(y_e+y_f)-f(y_e)-f(y_f)+1.
\]
Again, since \(f(s)-1\) has no terms of degree \(1\) or \(2\), every monomial in
\(\mathcal C_U\) has total degree at least \(3\). Hence
\[
w(U;z)\le C t^3.
\]
This proves the lemma.
\end{proof}

\begin{lem}[Weighted local bound]
\label{lem:weighted-local-bound}
There exist absolute constants \(c_0,C_0>0\) such that the following holds.

Let \(U\subseteq[n]\times[n]\) be connected, with \(m:=|U|\ge 3\). Write
\[
r_i(U):=\deg_{H_U}(i,\cdot),
\qquad
c_j(U):=\deg_{H_U}(\cdot,j).
\]
If \(t=|z|^2\) satisfies
$nt\le c_0$,
then
\[
w(U;z)
\le
(C_0 t)^m
\prod_{i:\,r_i(U)\ge1}(r_i(U)-1)!
\prod_{j:\,c_j(U)\ge1}(c_j(U)-1)!.
\]
\end{lem}

\begin{proof}
By Lemma~\ref{lem:exact-support-poincare}, it is enough to bound
\[
\bigl\|\partial_U F_U^{(2)}(zW)\bigr\|_{L^2}^2.
\]

Write
\[
s_i(W):=\sum_{j:(i,j)\in U} W_{ij},
\qquad
t_j(W):=\sum_{i:(i,j)\in U} W_{ij}.
\]
A direct algebraic identity gives
\[
\frac12\sum_{e\in U}(zW_e)^2+\sum_{\substack{e,f\in U\\ e\sim f}} zW_e\,zW_f
=
-\frac12 z^2\sum_{e\in U}W_e^2+\frac{z^2}{2}\sum_i s_i(W)^2+\frac{z^2}{2}\sum_j t_j(W)^2.
\]
Hence
\begin{align*}
F_U^{(2)}(zW)
&=
Z_{H_U}(zW)
\exp\!\Bigl(
-z\sum_{e\in U}W_e
-\frac{z^2}{2}\sum_{e\in U}W_e^2
+\frac{z^2}{2}\sum_i s_i(W)^2
+\frac{z^2}{2}\sum_j t_j(W)^2
\Bigr).
\end{align*}

Introduce independent real standard Gaussians \((g_i)\) and \((h_j)\), indexed by the
touched rows and columns of \(U\), and set
\[
\eta_{ij}:=g_i+h_j-1.
\]
Using
$e^{\frac{z^2}{2}u^2}=\E_g[e^{zgu}]$,
we obtain the Hubbard--Stratonovich representation \cite{hubbard1959calculation}
\[
F_U^{(2)}(zW)
=
\E_{g,h}\!\left[
Z_{H_U}(zW)
\prod_{e=(i,j)\in U}
\exp\!\Bigl(z\eta_{ij}W_{ij}-\frac{z^2}{2}W_{ij}^2\Bigr)
\right].
\]

Now expand
\[
Z_{H_U}(zW)=\sum_{M\in\mathcal M(H_U)} \prod_{e\in M} zW_e,
\]
where \(\mathcal M(H_U)\) denotes the matchings in \(H_U\). For each fixed matching \(M\),
differentiating once in every variable \(W_e\) yields the following.

If \(e\notin M\), then
\[
\partial_{W_e}
\exp\!\Bigl(z\eta_eW_e-\frac{z^2}{2}W_e^2\Bigr)
=
z(\eta_e-zW_e)
\exp\!\Bigl(z\eta_eW_e-\frac{z^2}{2}W_e^2\Bigr).
\]
If \(e\in M\), then
\[
\partial_{W_e}
\Bigl(
zW_e\,\exp\!\bigl(z\eta_eW_e-\tfrac{z^2}{2}W_e^2\bigr)
\Bigr)
=
z\bigl(1+z\eta_eW_e-z^2W_e^2\bigr)
\exp\!\Bigl(z\eta_eW_e-\frac{z^2}{2}W_e^2\Bigr).
\]

Define
\[
A_e:=|z|\,|\eta_e-zW_e|\,
\exp\!\Bigl(\Re\!\bigl(z\eta_eW_e-\tfrac{z^2}{2}W_e^2\bigr)\Bigr),
\]
\[
B_e:=|z|\,|1+z\eta_eW_e-z^2W_e^2|\,
\exp\!\Bigl(\Re\!\bigl(z\eta_eW_e-\tfrac{z^2}{2}W_e^2\bigr)\Bigr).
\]
Then
\[
\bigl|\partial_U F_U^{(2)}(zW)\bigr|
\le
\E_{g,h}\!\left[\sum_{M\in\mathcal M(H_U)}
\prod_{e\notin M}A_e\prod_{e\in M}B_e\right]
\le
\E_{g,h}\!\left[\prod_{e\in U}(A_e+B_e)\right].
\]
By Jensen,
\[
\bigl\|\partial_U F_U^{(2)}(zW)\bigr\|_2^2
\le
\E_{g,h}\E_W\!\left[\prod_{e\in U}(A_e+B_e)^2\right].
\]
Since the \(W_e\)'s are independent, this factors as
\[
\bigl\|\partial_U F_U^{(2)}(zW)\bigr\|_2^2
\le
\E_{g,h}\prod_{e\in U} J_t(\eta_e),
\]
where
\[
J_t(\eta):=
\E_{W\sim\CN(0,1)}\!\left[(A+B)^2\right]
\]
with \(\eta\) held fixed.

We now estimate \(J_t(\eta)\). First we make an observation based on rotational invariance.
Write
$z=\sqrt t\,e^{i\theta}$, 
and define
$\widetilde W:=e^{i\theta}W$.
Since \(\widetilde W\sim \CN(0,1)\), and since \(W=e^{-i\theta}\widetilde W\), we have
\[
zW=\sqrt t\,\widetilde W,
\qquad
z^2W^2=t\,\widetilde W^2.
\]
Therefore
\[
|\,\eta-zW\,|=|\,\eta-\sqrt t\,\widetilde W\,|,
\qquad
|\,1+z\eta W-z^2W^2\,|
=
|\,1+\sqrt t\,\eta\,\widetilde W-t\,\widetilde W^2\,|,
\]
and also
\[
\Re\!\Bigl(z\eta W-\frac{z^2}{2}W^2\Bigr)
=
\Re\!\Bigl(\sqrt t\,\eta\,\widetilde W-\frac t2\,\widetilde W^2\Bigr).
\]
It follows that \(J_t(\eta)\) depends only on \(t\), not on the phase of \(z\). Thus, by rotational invariance, it suffices to assume that \(z=\sqrt t>0\) is real.

Write
\[
W=\frac{X+iY}{\sqrt2},
\qquad
X,Y\stackrel{\mathrm{i.i.d.}}{\sim} N(0,1).
\]
Then
\[
2\Re\!\Bigl(\sqrt t\,\eta\,W-\frac t2W^2\Bigr)
=
\sqrt2\,\sqrt t\,\eta\,X-\frac t2(X^2-Y^2).
\]
Hence
\[
e^{2\Re(\sqrt t\,\eta\,W-\frac t2W^2)}\,d\mu(W)
=
K_t(\eta)\,d\nu_{t,\eta}(X,Y),
\]
where \(d\mu(W)\) is the standard complex Gaussian law, and
\[
K_t(\eta)
=
\frac1{\sqrt{1-t^2}}
\exp\!\Bigl(\frac{t\eta^2}{1+t}\Bigr),
\]
while under \(\nu_{t,\eta}\), the variables \(X\) and \(Y\) are independent Gaussians with
\[
X\sim N\!\Bigl(\frac{\sqrt2\,\sqrt t\,\eta}{1+t},\frac1{1+t}\Bigr),
\qquad
Y\sim N\!\Bigl(0,\frac1{1-t}\Bigr).
\]

Using \(J_t(\eta)\le 2\E[A^2]+2\E[B^2]\), we first estimate \(\E[A^2]\). Since
\[
A=\sqrt t\,|\eta-\sqrt t\,W|\,
e^{\Re(\sqrt t\,\eta\,W-\frac t2W^2)},
\]
we have
\[
\E[A^2]
=
t\,K_t(\eta)\,\E_{\nu_{t,\eta}}\!\left[\,|\eta-\sqrt t\,W|^2\right].
\]
Now
\[
|\eta-\sqrt t\,W|^2
=
\Bigl(\eta-\sqrt{\frac t2}\,X\Bigr)^2+\frac t2\,Y^2.
\]
Since under \(\nu_{t,\eta}\),
\[
\E[X]=\frac{\sqrt2\,\sqrt t\,\eta}{1+t},
\qquad
\Var(X)=\frac1{1+t},
\qquad
\Var(Y)=\frac1{1-t},
\]
it follows that
$\E_{\nu_{t,\eta}}\!\left[\,|\eta-\sqrt t\,W|^2\right]
\le
C(1+\eta^2)$
for all \(0\le t\le t_0\) with \(t_0<1\) fixed. Hence
\[
\E[A^2]\le C t(1+\eta^2)e^{Ct\eta^2}.
\]

Next, since
\[
B=\sqrt t\,|1+\sqrt t\,\eta\,W-tW^2|\,
e^{\Re(\sqrt t\,\eta\,W-\frac t2W^2)},
\]
we get
\[
\E[B^2]
=
t\,K_t(\eta)\,
\E_{\nu_{t,\eta}}\!\left[\,|1+\sqrt t\,\eta\,W-tW^2|^2\right].
\]
Using the inequality
$|1+u+v|^2\le 3(1+|u|^2+|v|^2)$,
we obtain
\[
|1+\sqrt t\,\eta\,W-tW^2|^2
\le
C\Bigl(1+t\eta^2|W|^2+t^2|W|^4\Bigr).
\]
Under \(\nu_{t,\eta}\), the moments of \(W\) satisfy
\[
\E_{\nu_{t,\eta}}[|W|^2]\le C(1+t\eta^2),
\qquad
\E_{\nu_{t,\eta}}[|W|^4]\le C(1+t\eta^2+t^2\eta^4),
\]
uniformly for \(0\le t\le t_0\). Therefore
\[
\E[B^2]
\le
C t\,K_t(\eta)\bigl(1+t\eta^2+t^2\eta^4\bigr).
\]
Since
$K_t(\eta)\le C e^{Ct\eta^2}$ and
$1+t\eta^2+t^2\eta^4 \le C(1+\eta^2)e^{Ct\eta^2}$,
it follows that
\[
\E[B^2]\le C t(1+\eta^2)e^{Ct\eta^2}.
\]

Combining the two bounds yields
\[
J_t(\eta)\le C t(1+\eta^2)e^{Ct\eta^2}.
\]
Substituting this into the previous display, we obtain
\[
\bigl\|\partial_U F_U^{(2)}(zW)\bigr\|_2^2
\le
(Ct)^m
\E_{g,h}\!\left[
\prod_{e=(i,j)\in U}(1+\eta_{ij}^2)e^{Ct\eta_{ij}^2}
\right].
\]

We now estimate the auxiliary Gaussian expectation. Since \(\eta_{ij}=g_i+h_j-1\),
\[
1+\eta_{ij}^2 \le C(1+g_i^2)(1+(h_j-1)^2),
\]
and
$e^{Ct\eta_{ij}^2}\le C\,e^{C t g_i^2}e^{C t (h_j-1)^2}$.
Therefore
\[
\prod_{e=(i,j)\in U}(1+\eta_{ij}^2)e^{Ct\eta_{ij}^2}
\le
C^m
\prod_i (1+g_i^2)^{r_i(U)}e^{Ct\,r_i(U)g_i^2}
\prod_j (1+(h_j-1)^2)^{c_j(U)}e^{Ct\,c_j(U)(h_j-1)^2}.
\]
Taking expectation and using independence gives
\begin{align*}
\E_{g,h}\!\left[
\prod_{e\in U}(1+\eta_e^2)e^{Ct\eta_e^2}
\right]
&\le
C^m
\prod_i \E\!\left[(1+G^2)^{r_i(U)}e^{\lambda_i G^2}\right]
\prod_j \E\!\left[(1+(H-1)^2)^{c_j(U)}e^{\mu_j(H-1)^2}\right],
\end{align*}
where \(G,H\sim N(0,1)\) and
\[
\lambda_i:=Ct\,r_i(U),\qquad \mu_j:=Ct\,c_j(U).
\]
Since \(r_i(U),c_j(U)\le n\) and \(nt\le c_0\), by choosing \(c_0\) small enough we may
assume that
\[
0\le \lambda_i,\mu_j\le \lambda_\ast
\]
for some fixed \(\lambda_\ast<1/2\). Hence, for every integer \(r\ge1\),
\[
\E[(1+G^2)^r e^{\lambda G^2}] \le C^r r!,
\qquad 0\le \lambda\le \lambda_\ast,
\]
since \(\lambda_\ast<1/2\) and \(\E[e^{\lambda G^2}]<\infty\) uniformly on
\([0,\lambda_\ast]\), while \((1+G^2)^r\le 2^r(1+|G|^{2r})\) and
\(\E|G|^{2r}\le C^r r!\).

Similarly,
\[
\E[(1+(H-1)^2)^r e^{\lambda(H-1)^2}] \le C^r r!,
\qquad 0\le \lambda\le \lambda_\ast,
\]
because \(H-1\sim N(-1,1)\) has density
\[
\frac{1}{\sqrt{2\pi}}e^{-(x+1)^2/2},
\]
so
\[
\E[(1+(H-1)^2)^r e^{\lambda(H-1)^2}]
=
\frac{1}{\sqrt{2\pi}}\int_{\R} (1+x^2)^r e^{\lambda x^2-(x+1)^2/2}\,dx.
\]
Since
\[
\lambda x^2-\frac{(x+1)^2}{2}
=
-\Bigl(\frac12-\lambda\Bigr)x^2-x-\frac12
\le
-\Bigl(\frac12-\lambda_\ast\Bigr)x^2+|x|,
\]
and \(\lambda_\ast<1/2\), the integrand is bounded by
$C(1+x^2)^r e^{-c x^2}$
for some absolute constants \(c,C>0\). Therefore
\[
\E[(1+(H-1)^2)^r e^{\lambda(H-1)^2}]
\le
C^r \int_{\R} (1+x^2)^r e^{-c x^2}\,dx
\le
C^r r!.
\]
Thus
\[
\E_{g,h}\!\left[
\prod_{e\in U}(1+\eta_e^2)e^{Ct\eta_e^2}
\right]
\le
C^m
\prod_i r_i(U)!\prod_j c_j(U)!.
\]
Combining the estimates gives
\[
\bigl\|\partial_U F_U^{(2)}(zW)\bigr\|_2^2
\le
(Ct)^m
\prod_i r_i(U)!\prod_j c_j(U)!.
\]

Finally, for every integer \(r\ge1\),
\[
r!\le 2^{r-1}(r-1)!.
\]
Since \(\sum_i r_i(U)=m=\sum_j c_j(U)\), we obtain
\[
\prod_i r_i(U)!\prod_j c_j(U)!
\le
4^m
\prod_i (r_i(U)-1)!\prod_j (c_j(U)-1)!.
\]
Absorbing the factor \(4^m\) into the constant proves the lemma.
\end{proof}

\begin{thm}[Second-order bound on \(L(K_{n,n})\)]
\label{thm:linegraph-unconditional}
There exist absolute constants \(c,C>0\) such that the following holds.

Let
\[
G_n:=L(K_{n,n}),
\qquad
F_n^{(2)}(x):=Z_{G_n}(x)\exp\!\bigl(-L_1(x)-L_2(x)\bigr),
\]
where
\[
L_1(x):=\sum_{(i,j)\in[n]\times[n]} x_{ij},
\qquad
L_2(x):=-\frac12\sum_{(i,j)}x_{ij}^2-\sum_{(i,j)\sim(i',j')}x_{ij}x_{i'j'}.
\]
Let \(W=(W_{ij})\) have i.i.d.\ entries \(W_{ij}\sim \CN(0,1)\), and define
\[
X_n^{(2)}(z):=F_n^{(2)}(zW).
\]
If
$n|z|^2\le c$,
then
\[
\log \E\bigl[|X_n^{(2)}(z)|^2\bigr]
\le
C\,n^4\,|z|^6.
\]
\end{thm}

\begin{proof}
Write \(E_n=[n]\times[n]\). Decompose
\[
F_n^{(2)}(x)=\sum_{A\subseteq E_n}\mathcal C_A(x_A),
\]
where \(\mathcal C_A\) is the exact-support part on \(A\). Then
\[
X_n^{(2)}(z)=F_n^{(2)}(zW)=\sum_{A\subseteq E_n}\mathcal C_A(zW_A).
\]
By orthogonality of complex Gaussian monomials,
\[
\E\bigl[|X_n^{(2)}(z)|^2\bigr]
=
\sum_{A\subseteq E_n}\E\bigl[|\mathcal C_A(zW_A)|^2\bigr].
\]

If \(A=U_1\sqcup\cdots\sqcup U_k\) is the decomposition of \(A\) into connected components in
\(G_n=L(K_{n,n})\), then \(F_A^{(2)}\) factorizes over the components, and therefore so does
the exact-support piece:
\[
\mathcal C_A=\prod_{\ell=1}^k \mathcal C_{U_\ell}.
\]
Hence, for connected \(U\subseteq E_n\), if we define
\[
w(U;z):=\E\bigl[|\mathcal C_U(zW_U)|^2\bigr],
\]
then
\[
\E\bigl[|X_n^{(2)}(z)|^2\bigr]
=
\sum_{\Gamma\ \mathrm{compatible}} \prod_{U\in\Gamma} w(U;z),
\]
where \(\Gamma\) runs over finite families of pairwise compatible connected supports.

We now forget compatibility. For \(e\in E_n\), set
\[
S_e(z):=\sum_{\substack{U\ni e\\ U\text{ connected}}} w(U;z).
\]
For each connected \(U\subseteq E_n\), fix once and for all a representative edge
\(\rho(U)\in U\), for example the lexicographically smallest edge of \(U\).
If
$\Gamma=\{U_1,\dots,U_k\}$
is a compatible family of connected supports, then the sets \(U_1,\dots,U_k\) are pairwise
disjoint, hence the representative edges \(\rho(U_1),\dots,\rho(U_k)\) are distinct.
Therefore the monomial
\[
\prod_{U\in\Gamma} w(U;z)
\]
appears among the terms in the expansion of
\[
\prod_{e\in E_n}(1+S_e(z)),
\]
namely by choosing from the factor indexed by \(e=\rho(U)\) the summand \(w(U;z)\) for each
\(U\in\Gamma\), and choosing \(1\) from all other factors. Since all weights are nonnegative,
\[
\E\bigl[|X_n^{(2)}(z)|^2\bigr]
=
\sum_{\Gamma\ \mathrm{compatible}} \prod_{U\in\Gamma} w(U;z)
\le
\prod_{e\in E_n}(1+S_e(z)).
\]
Hence, using \(\log(1+s)\le s\) for \(s\ge0\),
\[
\log \E\bigl[|X_n^{(2)}(z)|^2\bigr]
\le
\sum_{e\in E_n}\ \sum_{\substack{U\ni e\\ U\text{ connected}}} w(U;z).
\]

Fix a root edge \(e_0\in E_n\). By Lemma~\ref{lem:size-1-2-cancellation},
\[
\sum_{\substack{U\ni e_0\\ U\text{ connected}\\ |U|=1}} w(U;z)\le Ct^3.
\]
Also, there are at most \(2n-2\) connected supports of size \(2\) containing \(e_0\), so
\[
\sum_{\substack{U\ni e_0\\ U\text{ connected}\\ |U|=2}} w(U;z)\le C n\, t^3.
\]

Now let \(m\ge 3\). By Lemma~\ref{lem:weighted-local-bound} and
Proposition~\ref{prop:weighted-connected-supports},
\begin{align*}
\sum_{\substack{U\ni e_0\\ U\text{ connected}\\ |U|=m}} w(U;z)
&\le
(Ct)^m
\sum_{\substack{U\ni e_0\\ U\text{ connected}\\ |U|=m}}
\prod_i (r_i(U)-1)!\prod_j(c_j(U)-1)! \le
(Ct)^m \Cat_m n^{m-1}.
\end{align*}
Using \(\Cat_m\le 4^m\),
\[
\sum_{\substack{U\ni e_0\\ U\text{ connected}\\ |U|=m}} w(U;z)
\le
(C't)^m n^{m-1}.
\]
Therefore
\[
\sum_{m\ge 3}
\sum_{\substack{U\ni e_0\\ U\text{ connected}\\ |U|=m}} w(U;z)
\le
n^2 t^3 \sum_{m\ge 3}(C''nt)^{m-3}.
\]
If \(nt\le c\) with \(c>0\) sufficiently small, the geometric series is bounded by an
absolute constant, so
\[
\sup_{e_0\in E_n}
\sum_{\substack{U\ni e_0\\ U\text{ connected}}} w(U;z)
\le
C n^2 t^3.
\]

Finally, since \(|E_n|=n^2\),
\[
\log \E\bigl[|X_n^{(2)}(z)|^2\bigr]
\le
n^2\cdot C n^2 t^3
=
C n^4 t^3
=
C n^4 |z|^6.
\]
This proves the theorem.
\end{proof}

\end{document}